%% file: holant.tex
\documentclass[11pt]{article}
\usepackage[affil-it]{authblk}
\usepackage{fullpage}

\usepackage{amsthm,amsmath,amssymb}

\usepackage[bookmarks=true,hypertexnames=false,pagebackref]{hyperref}
\hypersetup{colorlinks=true, citecolor=blue, linkcolor=red,
  urlcolor=blue}
\usepackage[dvipsnames]{xcolor}
\usepackage{amsthm}
\usepackage[lining,semibold,type1]{libertine} 
\usepackage[T1]{fontenc} 
\usepackage{textcomp} 
\usepackage[varqu,varl]{inconsolata}
\usepackage[libertine,vvarbb]{newtxmath}
\usepackage[scr=rsfso,cal=cm]{mathalfa}
\usepackage{bm}
\usepackage{cleveref}
\usepackage{graphicx}
\usepackage{todonotes}
\usepackage{enumerate}
\usepackage{thmtools}

\newtheorem{theorem}{Theorem}
\newtheorem{lemma}[theorem]{Lemma}
\newtheorem{corollary}[theorem]{Corollary}
\newtheorem{definition}[theorem]{Definition}

\newtheorem{claim}[theorem]{Claim}

\newtheorem*{remark}{Remark}
\newtheorem{proposition}[theorem]{Proposition}
\newtheorem{condition}{Condition}

\crefname{condition}{Condition}{Conditions}
\crefname{proposition}{Proposition}{Propositions}
\crefname{lemma}{Lemma}{Lemmas}
\crefname{theorem}{Theorem}{Theorems}

\def\Lovasz{Lov\'asz\xspace}

\newcommand{\abs}[1]{\left\vert#1\right\vert}
\newcommand{\set}[1]{\left\{#1\right\}}
 \newcommand{\eps}{\varepsilon}
 
\renewcommand{\mid}{\;\middle\vert\;} \newcommand{\cmid}{\,:\,}
 
\newcommand{\defeq}{\triangleq}

\newcommand{\id}[1]{\mathbb{1}\left[#1\right]}

\usepackage[ruled,linesnumbered,vlined]{algorithm2e}
\usepackage{algpseudocode}
\SetKwRepeat{Do}{do}{while}

\def\*#1{\mathbf{#1}} \def\+#1{\mathcal{#1}} \def\-#1{\mathrm{#1}} \def\^#1{\mathbb{#1}} \def\$#1{\mathtt{#1}}
\newcommand{\bb}{\mathbb}

\def\!#1{\mathtt{#1}}
\def\@#1{\mathscr{#1}}

\def\EG{\emph{e.g.}\xspace}
\def\IE{\emph{i.e.}\xspace}

\def\bad{\!{bad}}
\def\cp{\!{cp}}
\def\good{\!{good}}
\def\poly{\mathrm{poly}}

\def\symbolwidth{\phantom{{}={}}}

\newcommand{\wh}[1]{\widehat{#1}}
\newcommand{\BuildTime}{\poly\left(\Delta^{\Delta\ell}\right)}
\newcommand{\Couple}{\!{Couple}}
\newcommand{\Ham}{\!{Ham}}

\newcommand{\RunningTimeExponent}{\poly\left(\Delta(G), 1/B(\Phi)\right)}

\newcommand{\zero}{\boldsymbol{0}}
\newcommand{\vecb}{\boldsymbol{b}}
\newcommand{\vecf}{\boldsymbol{f}}
\newcommand{\seqS}{\boldsymbol{s}}

\usepackage{xifthen}
\usepackage{thm-restate}

\renewcommand{\Pr}[2][]{ \ifthenelse{\isempty{#1}}
  {\mathbf{Pr}\left[#2\right]} {\mathbf{Pr}_{#1}\left[#2\right]} }
\newcommand{\E}[2][]{ \ifthenelse{\isempty{#1}}
  {\mathbf{E}\left[#2\right]}
  {\mathop{\mathbf{E}}_{#1}\left[#2\right]} }
\newcommand{\Var}[2][]{ \ifthenelse{\isempty{#1}}
  {\mathbf{Var}\left[#2\right]}
  {\mathop{\mathbf{Var}}_{#1}\left[#2\right]} }

\newcommand{\qgl}[1]{{\color{purple}{#1}}}

\allowdisplaybreaks

\newboolean{doubleblind}
\setboolean{doubleblind}{False}

\title{FPTAS for Holant Problems with Log-Concave Signatures}
\ifdoubleblind
    \author{Author(s)}
\else
    \author[1]{Kun He}
    \author[2]{Zhidan Li}
    \author[2]{Guoliang Qiu}
    \author[2]{Chihao Zhang}
    \affil[1]{Renmin University of China}
    \affil[2]{Shanghai Jiao Tong University}
\fi

\date{\today}
\begin{document}

\maketitle

\begin{abstract}
     For an integer $b\ge 0$, a $b$-matching in a graph $G=(V,E)$ is a set $S\subseteq E$ such that each vertex $v\in V$ is incident to at most $b$ edges in $S$.  We design a fully polynomial-time approximation scheme (FPTAS) for counting the number of $b$-matchings in graphs with bounded degrees. Our FPTAS also applies to a broader family of counting problems, namely Holant problems with log-concave signatures.

     Our algorithm is based on Moitra's linear programming approach (JACM'19). Using a novel construction called the \emph{extended coupling tree}, we derandomize the coupling designed by Chen and Gu (SODA'24).
\end{abstract}

\section{Introduction}

Holant problems, namely the read-twice counting CSPs, are an  expressive framework for counting problems. An instance of (Boolean domain) Holant problem is a tuple $\Phi=(G=(V,E),\set{f_v}_{v\in V})$ where each $f_v:\set{0,1}^{E_v}\to \bb C$ is a constraint function (usually called a signature) with domain $E_v$, the set of edges incident to $v$. The problem is to compute the \emph{Holant}, or the \emph{partition function} defined as 
\[
    \-{Holant}(\Phi) = \sum_{\sigma\in\set{0,1}^E} \prod_{v\in V}f_v\left(\sigma(E(v))\right),
\]
where $\sigma(E_v)$ is the restriction of $\sigma$ on $E_v$. The framework encodes a broad family of counting problems including counting matchings ($f_v(\tau) = \id{\abs{\tau}\le 1}$)\footnote{For an assignment $\sigma\in\set{0,1}^S$, we use $\abs{\sigma}$ to denote its Hamming weight, namely $\sum_{i\in S}\sigma(i)$.}, perfect matchings ($f_v(\tau) = \id{\abs{\tau}= 1}$), counting edge covers ($f_v(\tau) = \id{\abs{\tau}\ge 1}$) etc. 

Despite the success in classifying the exact counting complexity of Holant problems (\EG,~\cite{CGW16,CGW16edgecoloring,SC20,Bac21,CF23}), understanding their approximability remains a long-term yet incomplete task (\EG,~\cite{LWZ14,HLZ16,CLLY20,GLLZ21Holant,CLV22,HQZ23,CG24bMatching}). Even for \emph{symmetric} constraint functions with Boolean domains, namely those $f_v(\tau)$ whose value only depends on the Hamming weights of $\tau$, the complete picture is far from clear. 

One notable benchmark problem for algorithmic techniques is the problem of counting $b$-matchings, which asks for the number of $\set{0,1}$-assignments to edges such each vertex is adjacent to \emph{at most} $b$ edges with value $1$. It can be phrased in the Holant framework where each $f_v(\tau)=\id{\abs{\tau}\le b}$. Clearly, when $b=1$, it is the counting matching problem, which has been a central problem in approximate counting. However, when $b$ becomes larger, many good properties of matchings break down, posing new challenges in algorithm design. The problem was perhaps first studied in the work of ~\cite{HLZ16}, in which the rapid mixing of a particular Markov chain was established for $b\le 7$ using the method of canonical paths, and the chain can be used to uniformly sample from $b$-matchings. 
Recently, using a simple and neat coupling argument, among many other things, the \emph{spectral independence} property for the uniform distribution of $b$-matching was established in ~\cite{CG24bMatching}, which implies the rapid mixing of Glauber dynamics for sampling $b$-matchings. As a result, one obtains a fully polynomial-time randomized approximation scheme (FPRAS) for counting $b$-matchings for any $b\ge 0$.

In this work, we focus on \emph{deterministic} approximate counting algorithms. There are a few popular techniques for designing deterministic algorithms for Holant-type problems, including the method of correlation decay~\cite{BGKNT07,LWZ14,LLL14,LLZ14} and polynomial interpolation~\cite{GLLZ21Holant,BCW22,CFFGL22}, 
which result in fully polynomial-time approximation schemes (FPTAS) for counting problems. Both the methods of correlation decay and polynomial interpolation are successful for counting matchings and many other Holant problems.

However, for $b$ larger than $1$, the problem of $b$-matching resists both methods: the problem lacks a concise and lossless recursion for computing marginals necessary to apply the correlation decay technique, and it is challenging to determine the location of zeros for the partition function in order to use the polynomial interpolation method (in particular, the $H_\eps$-stable property in ~\cite{GLLZ21Holant} for local polynomials does not hold for a general $d$). 

We design an FPTAS for $b$-matching and more generally Holant problems with log-concave signature by developing the method of linear programming, invented by Moitra in \cite{Moitra19}, which was previously only applied to counting problems on hypergraphs under \Lovasz-local-lemma-type conditions~\cite{Moitra19,GLLZ19,GGGY21,JPV21,WY24}. We will summarize our main results in \Cref{sec:main-results} and explain our technical contributions in \Cref{sec:techniques}, respectively.

\subsection{Main results}\label{sec:main-results}

\paragraph{\ensuremath{\boldsymbol{b}}-matchings}

Given a graph $G = (V, E)$, recall that $E_v = \set{e \in E \cmid \mbox{$e$ is incident to $v$}}$ is the collection of all edges incident to $v$ for each $v \in V$. For a vector $\vecb = \set{b_v}_{v \in V} \in \mathbb{N}_{> 0}^V$, we say that $S \subseteq E$ is a \emph{$\vecb$-matching of $G$} if $\abs{S \cap E_v} \le b_v$ for every $v \in V$. Note that when $b_v = b$ for every $v \in V$ we obtain the typical $b$-matchings.

\begin{theorem}[Informal version of~\Cref{thm:formal-counting-b-matchings}] \label{thm:counting-b-matchings}
    Given any positive integers $\Delta$ and $b$, there exists an FPTAS for counting the number of $\vecb$-matchings for any graph with maximum degree $\Delta$ and any $\vecb = \set{b_v}_{v \in V}$ satisfying $b_v \le b$ for every $v \in V$.
\end{theorem}

Closely to $\vecb$-matchings, another problem is \emph{counting $\vecb$-edge covers}. For an edge subset $S \subseteq E$, we say that $S$ is a \emph{$\vecb$-edge cover of $G$} if $\abs{S \cap E_v} \ge b_v$ for every $v \in V$. Note that for every $\vecb$-edge cover $S$, its complement $E \setminus S$ forms a $\vecb'$-matching where $\vecb' = \set{b_v'}_{v \in V}$ is defined as $b_v' = \deg_G(v) - b_v$. Then one can easily derive the following result for counting $\vecb$-edge covers as an immediate corollary of counting $\vecb'$-matchings.

\begin{corollary}[$\vecb$-edge covers]
    Given any positive integers $\Delta$ and $b \le \Delta$, there exists an FPTAS for counting the number of $\vecb$-edge covers for any graph with maximum degree $\Delta$ and any $\vecb = \set{b_v}_{v \in V}$ satisfying $b_v \ge b$ for every $v \in V$.
\end{corollary}



\paragraph{Holant problems with log-concave signatures}

Beyond counting the number of $b$-matchings, we also consider the family of Holant problems with \emph{Boolean domain symmetric log-concave signatures}. 
Recall that an instance of the Holant problem is a tuple $ \Phi = \left(G = (V, E), \vecf = \set{f_v}_{v \in V}\right)$ where for each $v \in V$, $f_v:\set{0,1}^{E_v}\rightarrow \^C$. For each $v\in V$, the constraint function $f_v$ is usually referred to as the \emph{signature} on $v$.



In this work, we only consider signatures taking non-negative values. As a result, an instance $\Phi$ induces the Gibbs distribution $\mu=\mu_\Phi$ over $\Omega = \Omega_\Phi=\set{0,1}^{E}$ such that
\begin{equation}\label{def-holant-mu}
\begin{aligned}
\forall \sigma\in \Omega, \quad \mu(\sigma)=\frac{1}{Z}  \prod_{v \in V} f_v\left(\sigma(E_v)\right),
\end{aligned}
\end{equation}
where $Z = \-{Holant}(\Phi)$ is the partition function of $\Phi$.

Recall that a $d$-ary Boolean domain function $f$ is symmetric if its value only depends on the Hamming weight of the input. We can represent it in the compact way $f=[f(0),f(1),\dots, f(d)]$ where $f(i)\ge 0$ is the value of $f$ on inputs with Hamming weight $i$. We say $f$ is \emph{log-concave} if and only if it satisfies: 
\begin{enumerate}
    \item $f(k)^2 \ge f(k - 1) f(k + 1)$ for every $1 \le k \le d - 1$;
    \item If $f(k_1) > 0$ and $f(k_2) > 0$ for some $0 \le k_1 \le k_2 \le d$, then for every $k_1 \le j \le k_2$, $f(j) > 0$.
\end{enumerate}
We say that $\Phi = \left(G = (V, E), \set{f_v}_{v \in V}\right)$ is log-concave if all its signatures are log-concave. Clearly an instance of $b$-matchings, when stated as a Holant problem, is log-concave.


We consider a family of Holant instances satisfying the following conditions.

\begin{condition} \label{cond:Holant-condition}
Consider the Holant instance $\Phi=(G,\vecf)$ satisfying that $\vecf$ is a family of Boolean domain symmetric log-concave signatures where $f(0)>0$ for each $f \in \vecf$.
\end{condition}



\begin{theorem}[Informal version of~\Cref{thm:formal-counting-Holant}] \label{thm:counting-Holant}
    Given any positive integer $\Delta$ and any positive real number $r$, there exists an FPTAS for the partition function of every Holant instance $\Phi = \left(G = (V, E), \vecf = \set{f_v}_{v \in V}\right)$ satisfying~\Cref{cond:Holant-condition} where the maximum degree of $G$ is $\Delta$ and $f_v(1) \le r f_v(0)$ for each $v \in V$.
\end{theorem}

We remark that an FPRAS for Holant problems with log-concave signatures under the same condition was found in ~\cite{CG24bMatching}. 



\subsection{Overview of our techniques}\label{sec:techniques}

Since the (approximate) counting task can be reduced to estimating the {marginal ratios} of the Gibbs measure on a single edge, at a high level, our algorithm follows the idea of Moitra~\cite{Moitra19} to bootstrap the {marginal ratios} using a linear program (LP):
\begin{itemize}
    \item Design a random process to mimic the coupling of two conditional Gibbs distributions, whose boundary conditions only differ on a single-edge. 
    \item Establish a linear program to certify the random process, whose solution can recover the desired marginal ratios.
\end{itemize}

The Moitra's LP-based framework achieved significant success in designing deterministic algorithms for counting problems that involve high-order constraints~\cite{Moitra19,GLLZ19,GGGY21,JPV21,WY24}. In this work, we extend the applicability of the method to Holant problems with Boolean symmetric signatures. 
Similar to previous work, the design of a specific LP formulation and the retrieval of the marginal ratio from the solution of LP are the primary technical challenges in applying this framework. We provide an overview of the key technical aspects of our algorithm below.

Our starting point is the coupling process described in~\cite{CG24bMatching}, which has been employed to establish the spectral independence property of the underlying Gibbs measure. We first sketch the coupling for instances of $b$-matchings. The process maintains a pair of $b$-matching instances $(\Phi,\Phi')$ differing only on the signature on a unique vertex $v$. We assume the distinct signature is $f_v(\tau) = \id{\abs{\tau}\le k}$ and $f_v(\tau) = \id{\abs{\tau}\le k-1}$ on the vertex $v$ in $\Phi$ and $\Phi'$ respectively. The coupling proceeds as follows, trying to eliminate the discrepancy:
\begin{enumerate}
    \item The property of the $b$-matchings guarantees that for \emph{at least} one incident edge of $v$, say $e=\set{v,u}$, the marginal probability of assigning $1$ on $e$ in $\Phi$ is less than that in $\Phi'$. Then one applies the optimal (marginal) coupling on the edge $e$.
    \item If the assignments on $e$ successfully couple, one can repeat the argument to the remaining graph (with the discrepancy still at $v$).
    \item If the assignment on $e$ does not couple, the discrepancy moves to the vertex $u$ and one recursively couples the remaining graph (with the discrepancy at $u$).
\end{enumerate}

It is shown in~\cite{CG24bMatching} that the process terminates in logarithmic steps with high probability which is the key property of the coupling that our algorithm relies on. However, there are several technical challenges to certify the coupling using Moitra's method. 


\subsubsection{The extended coupling tree}
One crucial point in the first step of the above coupling is that the choice of $e$ cannot be arbitrary, at least for general log-concave instances of Holant problems. We provide an example illustrating this in \Cref{sec:counter-example}. This poses the main difficulty for us: 
To write an explicit linear program that certifies the entire coupling process, one has to know the identity of $e$ at each step, as in previous work, which is currently difficult to determine. 
This task of identifying $e$ might be as challenging as the original task of estimating the marginal in general Holant cases.



We overcome this difficulty by introducing a new structure called the \emph{extended coupling tree} based on the original coupling process. Specifically, we introduce a new set of auxiliary variables in the linear program for each edge $e'$ incident to $v$ when dealing with a pair of instances differing at the vertex $v$. We add constraints to ensure that the value of these variables forms a distribution. Clearly, the coupling itself corresponds to a feasible solution, namely setting the correct edge with probability one. These new variables have the following intended meaning: since we do not know which edge $e$ is the correct one (used in the coupling), pick one with a probability proportional to what the solution of the linear program indicates. In other words, instead of certifying an unknown edge $e$, we certify a \emph{fractional} choice of edges. {It is important to note that this approach may lead us to select ``wrong'' edges. However,} with appropriate constraints to control the error, we show that this larger linear program define on the extended coupling tree can bootstrap the desired marginal probability within the required precision as well. 

\subsubsection{Encoding the coupling error}

Another technical ingredient in Moitra's method is to trade off the efficiency and the precision of the marginal ratios estimation. Ideally, we can perfectly recover the marginal ratios by the LP based on certain couplings. However, it is intractable to write down the whole coupling tree as the number of possible states grows exponentially with the problem size. As a result, one must truncate the coupling tree to a tractable size and ensure the error incurred remains tolerable. 
In Moitra's method, the coupling error is hard-coded in the LP via linear constraints. In~\cite{Moitra19,GLLZ19,GGGY21,JPV21}, this is possible due to the local uniformity property under local-lemma-type conditions. In these work, variables are partitioned into marked one and unmarked ones. In the coupling, one selects an unpinned marked variable to apply the optimal coupling where the local uniformity ensures that the assignments on the variables can be successfully coupled with a good chance. This property can be encoded in LP by introducing a linear constraint relating the ratio of the marginals in consecutive steps in the coupling tree. Therefore, errors were encoded \emph{locally}.

In the very recent work of~\cite{WY24}, a new method to encode the coupling error has been introduced. They utilized the constraint-wise coupling and completely dispensed with the local uniformity to obtain tighter bounds. 
 They wrote the error caused along root-to-leaf paths in the coupling tree in a single constraint, which can be viewed as a \emph{global} way to encode the error.

Our method of encoding error can be viewed as a middle ground between previous approaches. We write the error caused by the coupling in a few steps in a single constraint, reflecting the nature of the coupling in~\cite{CG24bMatching} and the structure of our extended coupling tree. Our analysis for constraints of this form also demonstrates the flexibility of controlling the error in Moitra's method, which might find applications in other problems.





\subsection{Organization of the paper}

After introducing basic notations and a few standard properties of Holant instances in \Cref{sec:prelim}, we introduce our main construction, the extended coupling tree, in \Cref{sec:extended-coupling-tree}. Using this structure, we apply Moitra's linear programming approach to estimate the marginals in \Cref{sec:LP}. Finally in \Cref{sec:counting}, we present our approximate counting algorithms and prove the main theorems.

\section{Preliminaries}\label{sec:prelim}

\subsection{Notations}

For a natural number $n \in \mathbb N$, we use $[n]$ to denote the set $\set{1, 2, \ldots, n}$.
We use $\log(\cdot)$ to denote the logarithm with the natural base.
For an event $\+E$, we use $\id{\+E}$ to denote the indicator of it. 
We will use boldface type, e.g., $\boldsymbol{S},\boldsymbol{c}$, for vectors.
Specifically, we use $\boldsymbol{0}$ and $\boldsymbol{1}$ to denote the all-zero vector and the all-one vector, respectively.
For any sequence $\seqS = (s_1, \ldots, s_\ell)$ and any element $t$, let $\seqS \circ t$ denote $(s_1, \ldots, s_\ell,t)$, the concatenation of $\seqS$ and $t$.
With a slight abuse of notation, we will occasionally use the sequence $\seqS= (s_1, \ldots, s_\ell)$ to denote the set $\{s_1, \ldots, s_\ell\}$ when it is clear from the context.
Moreover, we use the notation $\varnothing$ to denote an empty sequence.

We consider the graphs with ``half-edges'' in this paper.
Formally, given a set of nodes $V$, an \emph{normal} edge on $V$ is a pair of different nodes $\{u,v\}$ where $u,v\in V$,
and a half-edge on $v$ is a singleton tuple $\{u\}$ where $u \in V$. Moreover, we sometimes abuse $e$ as a set of vertices, \IE, we will use ``$v \in e$'' to denote the event that $v$ is an endpoint of $e$ and $v \notin e$ otherwise.
A graph with half-edges is a pair $(V,E_1\cup E_2)$ where $V$ is a set of nodes, $E_1$ is a set of normal edges,
and $E_2$ is a set of half-edges. 
When referring to a graph with half-edges, we will always use two subsets to denote its normal edges and half-edges separately.
Thus, when we use the notation $G=(V,E)$, it implies that 
$G$ has no half-edges.
One can verify that the notations $\mu_\Phi$, $Z_\Phi$, e.t.c., can be naturally generalized to instances where the underlying graph has half edges.
Given a graph $G$, we always use $V(G)$ to denote the vertex set of $G$, $E(G)$ to denote the edge set, and $\Delta(G)$ to denote the maximum degree. 
For every $v \in V$, we use $\deg_G(v)$ to denote the degree of $v$ in $G$, \IE, the number of edges incident to $v$.
If $G$ is clear from the context, we may omit $G$ from these notations.
For convenience, we suppose that an arbitrary order is assumed over all the edges in $E_1$.


We call $\sigma: E\rightarrow \set{0,1}$ an assignment of $G$.
For any subset $S\subseteq E$ and any assignment $\sigma$, we use $\sigma(S)$ to denote the assignments of $\sigma$ on $S$.
For simplicity, let $\sigma(e)$ denote $\sigma(\{e\})$ for each $e\in E$. For any $S\subseteq E$, we call $\sigma: S\rightarrow \set{0,1}$ a \emph{partial assignment} of $G$ defined on $S$.
For any subset $S \subseteq E$ and any partial assignment $\sigma$ defined on $S$, 
let $\Lambda(\sigma)$ denote $S$.
For any assignment $\sigma$ and any partial assignment $\tau$,
we abuse the notation $\sigma\in \tau$ to denote the event that $\sigma(\Lambda(\tau)) =\tau$.

Given any partial assignment $\sigma$ and any vertex $v \in V$, we use $\Ham(\sigma, E_v)$ to denote $\abs{\sigma(E_v \cap \Lambda(\sigma))}$, i.e., the Hamming weight of $\sigma$ restricted to $E_v \cap \Lambda(\sigma)$.
For two partial assignments $\sigma$ and $\tau$ where $\Lambda(\sigma)\cap \Lambda(\tau) = \emptyset$, we use $\sigma \land \tau : \Lambda(\sigma) \cup \Lambda(\tau) \to \set{0, 1}$ to denote the concatenation of $\sigma$ and $\tau$, \IE, $(\sigma \land \tau)(e)=\sigma(e)$ for each $e\in \Lambda(\sigma)$ and $(\sigma \land \tau)(e)=\tau(e)$ for each $e\in \Lambda(\tau)$. 
{Additionally, for any $S\subseteq E$ and any $\boldsymbol{c} \in \{0,1\}^S$, we use $S \gets \boldsymbol{c}$ to denote the partial assignment $\sigma$ where $\Lambda(\sigma) = S$ and $\sigma(S) = \boldsymbol{c}$.}

Given any instance $ \Phi = \left(G = (V, E), \vecf = \set{f_v}_{v \in V}\right)$ of the Holant problem,
we always assume $\emph{w.l.o.g.}$ that the signature $f_v = [f_v(0),\cdots,f_v(d)]$ at $v$ satisfies $d = \deg_G(v)$ for each vertex $v\in V(G)$.
Recall the distribution $\mu \triangleq \mu_{\Phi}$ defined in \eqref{def-holant-mu}.
We say an assignment $\sigma$ of $G$ is feasible if $\mu(\sigma)>0$. 
We say a partial assignment $\tau$ is \emph{feasible} if there exists a feasible assignment $\sigma\in \tau$.




For any partial assignment $\tau$ and any vertex $v\in V$, we use $E_v^\tau \defeq E_v \setminus \Lambda(\tau)$ to denote the set of unpinned edges incident to $v$ under $\tau$. 
Given any feasible partial assignment $\tau$ and any subset $S \subseteq E \setminus \Lambda(\tau)$, we denote by $\mu_S^\tau$ the marginal probability distribution of $\mu$ projected to $S$ conditional on $\tau$.
Similarly, let $\mu_S$ denote the marginal probability distribution of $\mu$ projected to $S$.
Furthermore, given any partial assignment $\sigma$ where $S\subseteq \Lambda(\sigma)$ and any $e\in E$ where $e\not\in \Lambda(\tau)$, we will use the following simplified notations:
\begin{itemize}
\item Let $\mu_S^\tau(\sigma),\mu_S(\sigma)$ denote $\mu_S^\tau(\sigma(S)),\mu_S(\sigma(S))$, respectively;
\item Let $\mu^\tau(\sigma),\mu(\sigma)$ denote $\mu_{\Lambda(\sigma)\setminus\Lambda(\tau)}^\tau(\sigma),\mu_{\Lambda(\sigma)}(\sigma)$, respectively;
\item Let $\mu_e^\tau(\sigma),\mu_e(\sigma)$ denote $\mu_{\{e\}}^\tau(\sigma),\mu_{\{e\}}(\sigma)$, respectively.
\end{itemize}

\subsection{Properties of Holant instances}

For a symmetric signature $f = [f(0), f(1), \ldots, f(d)]$ of arity $d$, its \emph{local polynomial} introduced by~\cite{GLLZ21Holant} is defined as follows:
\begin{align} \label{eq:local-polynomial}
    P_f(x) \defeq \sum_{i = 0}^{d} \binom{d}{i} f(i) x^i.
\end{align}
Let $\Phi = \left(G = (V, E = E_1 \cup E_2), \vecf = \set{f_v}_{v \in V}\right)$ be a Holant instance satisfying~\Cref{cond:Holant-condition}. 

We define the pinning of Holant instances as follows.
\begin{definition}[Pinning of Holant instances] \label{def:pinning-Holant}
\emph{
    Given a Holant instance $\Phi = \left(G = (V, E = E_1 \cup E_2), \vecf = \set{f_v}_{v \in V}\right)$ satisfying~\Cref{cond:Holant-condition}, an edge $e \in E$ and a value $c \in \set{0, 1}$, the \emph{pinning of $\Phi$ with $e$ assigned to $c$}, denoted by $\Phi^{e \gets c} = \left(G^{e \gets c}, \vecf^{e \gets c} = \set{f_v^{e \gets c}}_{v \in V}\right)$, is the Holant instance satisfying
    \begin{itemize}
        \item 
        $G^{e \gets c} = \left(V, E \setminus \set{e}\right)$;
        \item For every $v \in V$, if $v \not\in e$, then $f_v^{e \gets c} = f_v$; otherwise, $f_v^{e \gets c} = \left[f_v(0 + c),\cdots, f_v({\deg_G(v)} - 1 + c)\right]$.
    \end{itemize}
}
\end{definition}

For $e \in E$ and $c \in \set{0, 1}$, let $Z^{e \gets c} = Z_{\Phi}^{e \gets c} \defeq Z_{\Phi^{e \gets c}}$ denote the partition function conditional on $e \gets c$. One can verify that
$$
    Z^{e \gets c} = \sum_{\sigma \in \set{0, 1}^E \cmid \sigma(e) = c} \prod_{v \in V} f_v\left(\abs{\sigma ({E_v})}\right)
$$
and $Z = Z^{e \gets 0} + Z^{e \gets 1}$. We remark here that $Z^{e \gets c} > 0$ if and only if for every $v \in e$, $f_v(c) > 0$ and for $\Phi$ satisfying~\Cref{cond:Holant-condition}, $Z^{e \gets 0}$ is positive.

For the instance $\Phi$, we define the following quantities:
\begin{align}\label{eq-def-bmin}
    B = B(\Phi) \defeq \min_{v \in V} \frac{P_{f_v}(0)}{P_{f_v}(r_{\max})}
\end{align}
where
\begin{align}\label{eq-def-rmax}
    r_{\max} = r_{\max}(\Phi) \defeq \max_{v \in V} \frac{f_v(1)}{f_v(0)}
\end{align}
with convention $f_v(1) = 0$ for every isolated vertex $v \in V$.
Note that since $\Phi$ satisfies~\Cref{cond:Holant-condition}, it holds that $r_{\max} \ge 0$. Hence $0 \le P_{f_v}(0) \le P_{f_v}(r_{\max})$ by~\eqref{eq-def-bmin}, leading to $0 < B \le 1$. When $B = 1$, it follows from the definition that $f_{v}(1) = 0$ for all $v \in V$. In this case, the only feasible assignment of $\Phi$ is $\sigma(e) = 0$ for every $e \in E$ and the partition function is $Z = \prod_{v \in V} f_v(0)$, which is trivial. Thus throughout this paper's subsequent part, we assume that $0 < B < 1$.

The following lemma in~\cite{CG24bMatching} shows the monotonicity of $r_{\max}$ and $B$ under the pinnings. For completeness, we postpone the proof of it in~\Cref{sec:proof-Holant-properties}.

\begin{restatable}[Observation 15 in~\cite{CG24bMatching}]{lemma}{Monotonicity} \label{lem:monotonicity-under-pinning}
    Given a Holant instance $\Phi = (G = (V, E = E_1 \cup E_2), \vecf)$ satisfying~\Cref{cond:Holant-condition}, for an edge $e \in E$ and $c \in \set{0, 1}$, if $Z_{\Phi}^{e \gets c} > 0$, then the following inequalities hold for $r_{\max}$ and $B$:
    $$
        r_{\max}(\Phi) \ge r_{\max}(\Phi^{e \gets c}), \quad B(\Phi) \le B(\Phi^{e \gets c}).
    $$
\end{restatable}

Recall that $\mu = \mu_\Phi$ is the induced Gibbs distribution. Given any partial assignment $\sigma$ and vertex $v\in V$, recall that $E_v^\sigma$ is the set of unpinned edges incident to $v$ under $\sigma$. The following lemma in~\cite{CG24bMatching} establishes lower bounds of the marginal probability of $\mu$.
\begin{restatable}[Lemma 18 in~\cite{CG24bMatching}]{lemma}{MarginalBound} \label{lem:marginal-bound}
    Given any Holant instance $\Phi = \left(G=(V,E=E_1\cup E_2) , \vecf \right)$ satisfying~\Cref{cond:Holant-condition}, it holds that
    \begin{align} \label{eq:marginal-bound}
        \mu_{E_v^\sigma}^{\sigma}(\zero) \ge B(\Phi),
    \end{align}
    for any feasible partial assignment $\sigma$ and vertex $v\in V$ where $E^{\sigma}_v\subseteq E_1$.
\end{restatable}

As shown in subsequent sections, to apply~\Cref{lem:marginal-bound}, we also need to decide the feasibility of the partial assignment by the following lemma. The proof of this lemma is provided in~\Cref{sec:proof-Holant-properties}.

\begin{restatable}{lemma}{PartialFeasibility}\label{lem:partial-assignment-feasibility}
    Given any instance $\Phi = \left(G =(V,E=E_1\cup E_2), \vecf \right)$ satisfying~\Cref{cond:Holant-condition} with any partial assignment $\sigma$, there is an algorithm deciding whether $\sigma$ is feasible in time $O(\abs{\Lambda(\sigma)})$.
\end{restatable}


For every edge $e \in E$, define its \emph{marginal ratio} as
\begin{align}\label{def-rphie}
    R(e)=R_{\Phi}(e) \defeq \frac{\Pr[X \sim \mu]{X(e) = 1}}{\Pr[X \sim \mu]{X(e) = 0}} = \frac{\mu_e(1)}{\mu_e(0)},
\end{align}
which is well-defined since $f_v(0)>0$ for each $v\in V$ by~\Cref{cond:Holant-condition}. Note that $R(e) = Z^{e \gets 1} / Z^{e \gets 0}$. The marginal ratio of a half-edge can be bounded as follows.
\begin{restatable}{lemma}{MarginalRatioBound} \label{lem:marginal-ratio-upper-bound}
    Given any Holant instance $\Phi = \left(G =(V,E=E_1\cup E_2), \vecf \right)$ satisfying~\Cref{cond:Holant-condition} , for each half-edge $e \in E_2$, it holds that
    \begin{align*}
        R_{\Phi}(e) \le r_{\max}(\Phi).
    \end{align*}
\end{restatable}
For completeness, we prove~\Cref{lem:marginal-ratio-upper-bound} in~\Cref{sec:proof-Holant-properties}.




\section{The Extended Coupling Trees}\label{sec:extended-coupling-tree}
This section introduces the key ingredient of our counting algorithm, the extended coupling tree.
As mentioned in \Cref{sec:techniques}, our counting algorithm employs
Moitra's linear programming approach~\cite{Moitra19} to derandomize the coupling procedure for Holant problems~\cite{CG24bMatching}.
Firstly, we describe the coupling procedure in \Cref{sec-coupling-procedure}.
Next, we construct a random process to simulate the coupling procedure in \Cref{sec-rp}.
Following this, we encode the states of the random process within the extended coupling tree in \Cref{sec-tct}. 
At last, we extend the marginals in the random process to the extended coupling tree in \Cref{sec-marginal}.
The extended coupling tree and the marginals form the foundation for establishing the linear program.
Unlike typical methods, our extended coupling tree includes impossible states in the random process, allowing for its efficient construction.


\subsection{The coupling procedure}\label{sec-coupling-procedure}
In this section, we review the coupling procedure for Holant problems from \cite{CG24bMatching}.
The procedure inspires our random process in \Cref{sec-rp} for simulating the coupling and our extended coupling tree in \Cref{sec-tct} for encoding the states of the random process.

In the subsequent discussion, we consider instances satisfying the following condition.
 
\begin{condition}\label{cond-instancepair}
    The following holds for the tuple $\left(\Phi = (G, \vecf),  \sigma_\bot, \tau_\bot,v_{\bot}\right)$:
    \begin{itemize}
        \item $G = (V, E = E_1 \cup \set{e_\bot})$ is a graph with a unique half edge $e_\bot = \set{v_{\bot}}$ where $v_\bot \in V$;
        \item $\Phi=(G,\vecf)$ satisfying~\Cref{cond:Holant-condition};
		\item $\sigma_\bot = (e_\bot \gets 1)$ and $\tau_\bot = (e_\bot \gets 0)$.
    \end{itemize}
\end{condition}

For an instance $(\Phi = (G, \vecf), \sigma_\bot, \tau_\bot, v_\bot)$ satisfying~\Cref{cond-instancepair}, recall that an arbitrary order is assumed over all the edges in $E_1$.
Also recall the distribution $\mu \triangleq \mu_{\Phi}$ defined in \eqref{def-holant-mu}.
    
Let $\sigma, \tau$ be two partial assignments of $G$ where $\Lambda(\sigma) = \Lambda(\tau)$.
We say that a vertex $v \in V$ is \emph{weight-distinct under $(\sigma, \tau)$} if ${\!{Ham}\left(\sigma,{E_v}\right)}\neq{\!{Ham}\left(\tau,{E_v}\right)}$. Furthermore, a vertex $v$ is called \emph{$1$-distinct} under $(\sigma, \tau)$ if $\abs{{\!{Ham}\left(\sigma,{E_v}\right)}-{\!{Ham}\left(\tau,{E_v}\right)}}=1$. We say that $(\sigma, \tau)$ is a pair of \emph{$1$-discrepancy partial assignments} if there is a unique weight-distinct vertex $v \in V$ and $v$ is 1-distinct.
Given a pair of $1$-discrepancy partial assignments $(\sigma,\tau)$ where $v$ is the unique weight-distinct vertex,
we call an edge $e \in E_v^\sigma$ \emph{amenable}
under $(\sigma,\tau)$
if $\mu_e^\sigma(1) \ge \mu_e^\tau(1)$ when $\Ham(\sigma, E_v) < \Ham(\tau, E_v)$ or $\mu_e^\sigma(1) \le \mu_e^\tau(1)$ when $\Ham(\sigma, E_v) > \Ham(\tau, E_v)$.

Given a pair of 1-discrepancy partial assignments $(\sigma,\tau)$, Algorithm \ref{algo:Couple} from \cite{CG24bMatching} defines a coupling of $\mu^{\sigma}$ and $\mu^{\tau}$. 
The following proposition demonstrates the correctness of the coupling process, which has been proved in~\cite{CG24bMatching}.
\begin{algorithm}[htbp]
    \caption{$\!{Couple}(\Phi, \sigma, \tau, v)$}
    \label{algo:Couple}
         \KwIn{A Holant instance $\Phi = \left(G = (V, E = E_1\cup \set{e_\bot}), \vecf = \set{f_v}_{v \in V}\right)$, a pair of 1-discrepancy partial assignments $(\sigma, \tau)$ of $G$ where $\sigma(e_\bot) = 1, \tau(e_\bot) = 0$, and the unique weight-distinct vertex $v \in V$ under $(\sigma, \tau)$.}
    
      \KwOut{A pair of assignments drawn from a coupling between $\mu^{\sigma}$ and $\mu^{\tau}$.}	
       
        $S \leftarrow \emptyset$\;
        \While{
            $E_v^{\sigma} \neq \emptyset$\label{line-while}
        }
        {
            \If{${\!{Ham}\left(\sigma, {E_v}\right)} < {\!{Ham}\left(\tau, {E_v}\right)}$}
            {
                let $e = \set{u,v}$ be the first edge satisfying $e \in E_v^{\sigma}$ and   $\mu^{\sigma}_e(1) \geq \mu^{\tau}_e(1)$\; \label{line:pick-dominating-edge-1}
                \quad \tcp{pick the amenable edge}
            }
            \Else{
                let $e = \set{u,v}$ be the first edge satisfying $e\in E_v^{\sigma}$ and 
                $\mu^{\sigma}_e(1) \leq \mu^{\tau}_e(1)$\; \label{line:pick-dominating-edge-2}
                \quad \tcp{pick the amenable edge}
            }
            sample $(\sigma_e, \tau_e)$ from an optimal coupling of $(\mu^{\sigma}_e, \mu^{\tau}_e)$\label{line-edge-sample}\;
            $\sigma \gets \sigma \land \sigma_e$, $\tau \gets \tau \land \tau_e$, $S \gets S\cup \{e\}$\; 
            \If{$\sigma_e \neq \tau_e$}
            {
                $(\sigma',\tau')\leftarrow \!{Couple}(\Phi,\sigma, \tau, u)$; \quad 
                \tcp{The weight-distinct vertex has been changed}\label{line-recursive-call}
                \Return $\left(\sigma(S)\land \sigma',\tau(S)\land \tau'\right)$
            }    
        }
        Sample $\sigma'\sim \mu^{\sigma}$\;
       \Return $(\sigma(S)\land \sigma', \tau(S)\land \sigma')$\;     
\end{algorithm}

\begin{proposition}[\cite{CG24bMatching}] \label{prop:coupling-correctness}
    The procedure $\!{Couple}(\Phi, \sigma_\bot, \tau_\bot, v_\bot)$ 
    satisfies the following properties:
    \begin{enumerate}
        \item \textbf{(Soundness of the coupling)} It always terminates.
        
        \item \textbf{(Validity of the recursive call)} At each call of $\!{Couple}(\Phi, \sigma, \tau, u)$ in Line \ref{line-recursive-call}, $(\sigma, \tau)$ is always a pair of 1-discrepancy partial assignments where $\sigma(e_\bot) = 1, \tau(e_\bot) = 0$, and $u$ is always the unique weight-distinct vertex under $(\sigma, \tau)$. Moreover, $\sigma$ and $\tau$ are feasible.
        
        \item \textbf{(Existence of amenable edges)} When $E_v^\sigma \neq \emptyset$ in Line \ref{line-while}, there always exists an amenable edge $e \in E_v^\sigma$.
        
        \item \textbf{(Correctness of the coupling)} The outcome of $\!{Couple}(\Phi, \sigma, \tau, v)$ is a coupling of $(\mu^{\sigma}, \mu^{\tau})$.
    \end{enumerate}
\end{proposition}

\subsection{Random process simulating the truncated coupling procedure}\label{sec-rp}

In this section, given an instance $(\Phi = (G, \vecf), \sigma_\bot, \tau_\bot, v_\bot)$ satisfying~\Cref{cond-instancepair}, we construct a random process to simulate $\!{Couple}(\Phi, \sigma_\bot, \tau_\bot, v_\bot)$ with truncation. 

In the subsequent discussion, we always consider the tuples satisfying the following condition.
\begin{condition}\label{condition-sigma-tau}
The tuples $(\sigma,\tau,\seqS)$, $(\sigma,\tau,\seqS,e)$,  $(\sigma,\tau,\seqS, v, L)$ and $(\sigma,\tau,\seqS, v, L,e)$ satisfy the following:
\begin{itemize}
\item $(\sigma, \tau)$ is a pair of 1-discrepancy partial assignments and $\sigma(e_\bot) = 1, \tau(e_\bot) = 0$;
\item $\seqS$ records the sequence of assigned edges in $\Lambda(\sigma)\setminus \{e_\bot\}$;
\item $v$ is the unique weight-distinct vertex under $(\sigma, \tau)$;
\item $L = \abs{\{e'\in \Lambda(\sigma)\mid (e'\neq e_{\bot})\land (\sigma(e')\neq \tau(e'))\}}$;
\item $e$ is an edge in $E^{\sigma}_v$.
\end{itemize}
\end{condition}

For any tuple generated by certain processes, we will also ensure that \Cref{condition-sigma-tau} is satisfied, including those processes in Definitions  \ref{def:truncated-random-process} and \ref{def:truncated-coupling-tree}.
Given any $(\sigma,\tau)$, let $v(\sigma,\tau)$ denote the unique weight-distinct vertex under $(\sigma, \tau)$ and $L(\sigma,\tau)$ denote $\abs{\{e\in \Lambda(\sigma)\mid (e\neq e_{\bot})\land (\sigma(e)\neq \tau(e))\}}$.

The random process to simulate $\!{Couple}(\Phi, \sigma_\bot, \tau_\bot, v_\bot)$ is as follows.

\begin{definition}[$\ell$-truncated random process] \label{def:truncated-random-process}
    \emph{
    For any instance $\left(\Phi = (G, \vecf),  \sigma_\bot, \tau_\bot,v_{\bot}\right)$ satisfying \Cref{cond-instancepair} and  integer $\ell>0$, the $\ell$-truncated 
    random process $P^{\!{cp}} \triangleq P^{\!{cp}}_\ell(\Phi, \sigma_\bot, \tau_\bot, v_\bot) = \set{(\sigma_t, \tau_t, \seqS_t, v_t, L_t)}_{0\leq t \leq T}$ repeats the following operations:
    \begin{enumerate}
        \item The initial state is $(\sigma_0, \tau_0, \seqS_0, v_0, L_0) = (\sigma_\bot, \tau_\bot, \varnothing, v_\bot, 0)$.
        \item For $t = 0, 1, \cdots$: for the state $(\sigma_t, \tau_t, \seqS_t, v_t, L_t)$:
        \begin{enumerate}
            \item          
            If $L_t \geq \ell$ or $E_{v_t}^{\sigma_t}=\emptyset$, then the process lets $T \leftarrow t$, outcomes $(\sigma_t, \tau_t, \seqS_t, v_t, L_t)$, and stops.
            \item Otherwise, $E_{v_t}^{\sigma_t}\neq \emptyset$. 
            If ${\!{Ham}\left(\sigma_t, {E_{v_t}}\right)} < {\!{Ham}\left(\tau_t, {E_{v_t}}\right)}$,
            let $e = \set{u,v_t}$ be the first edge in $E_{v_t}^{\sigma_t}$ with 
            $\mu^{\sigma_t}_e(1) \geq \mu^{\tau_t}_e(1)$.
            Otherwise, let $e = \set{u,v_t}$ be the first edge in $E_{v_t}^{\sigma_t}$ with $\mu^{\sigma_t}_e(1) \leq \mu^{\tau_t}_e(1)$.
                \begin{enumerate}[(i)]
                    \item Sample $(\sigma_e,\tau_e)$ from an optimal coupling of $(\mu_e^{\sigma_t},\mu_e^{\tau_t})$.
                    \item Let $\sigma_{t + 1} \gets \sigma_t \land \sigma_e,\tau_{t + 1} \gets \tau_t \land \tau_e$, $\seqS_{t + 1} \gets \seqS_t \circ e$.
                    If $\sigma_e=\tau_e$, let $v_{t+1}\gets v_t, L_{t+1}\gets L_t$; otherwise, let $v_{t+1}\gets u$ and $L_{t+1}\gets L_t+1$.
                \end{enumerate} \label{item:construction-expanding-step}
        \end{enumerate}
    \end{enumerate} 
    }
\end{definition}

Intuitively, for any state $(\sigma_t, \tau_t, \seqS_t, v_t, L_t)$ in $P^{\!{cp}}$, $\seqS_t$ records the sequence of chosen edges and $L_t$ records the total number of edges $e$ in $\seqS_t$ where $\sigma_t(e)\neq \tau_t(e)$.
By comparing \Cref{def:truncated-random-process} with Algorithm \ref{algo:Couple}, one can verify that $P^{\!{cp}} \triangleq P^{\!{cp}}_\ell(\Phi, \sigma_\bot, \tau_\bot, v_\bot)$ simulates $\!{Couple}(\Phi, \sigma_\bot, \tau_\bot, v_\bot)$ while ensuring that the depth of recursion does not exceed $\ell$.
Moreover, one can verify the following lemma by \Cref{def:truncated-random-process} and the induction.


\begin{lemma}\label{lemma-trp-correctness}
In the $\ell$-truncated random process $P^{\!{cp}} \triangleq P^{\!{cp}}_\ell(\Phi, \sigma_\bot, \tau_\bot, v_\bot) = \set{(\sigma_t, \tau_t, \seqS_t, v_t, L_t)}_{0\leq t \leq T}$, we have 
$(\sigma_t, \tau_t,\seqS_t, v_t, L_t)$ satisfies \Cref{condition-sigma-tau} for each $0\leq t \leq T$.
\end{lemma}

The following notations related to the random process $P^{\!{cp}} \triangleq P^{\!{cp}}_\ell(\Phi, \sigma_\bot, \tau_\bot, v_\bot) = \set{(\sigma_t, \tau_t, \seqS_t, v_t, L_t)}_{0\leq t \leq T}$ will be used in its analysis. 
\begin{definition}\label{def-notation-trp}
In the $\ell$-truncated random process $P^{\!{cp}}$, define $\mathbf{Pr}_{\!{cp}}$ as follows.
\begin{itemize}
\item For any $(\sigma, \tau, \seqS, v, L)$, let $\Pr[\!{cp}]{(\sigma, \tau, \seqS, v, L)}$ denote the probability that $(\sigma, \tau, \seqS, v, L)$ is a state in the random process $P^{\!{cp}}$.
Formally, for each $(\sigma, \tau, \seqS, v, L)$  where  $\abs{\seqS}=t$,
\begin{align}\label{eq-def-pro-stsvl}
    \Pr[\!{cp}]{(\sigma, \tau, \seqS, v, L)} \triangleq \Pr{\left(T\geq t\right)\land \left((\sigma, \tau, \seqS, v, L) = (\sigma_t, \tau_t, \seqS_t, v_t, L_t)\right)}.
\end{align}

\item For any $(\sigma, \tau, \seqS, v, L)$ and $e\in E^{\sigma}_v$, let $\Pr[\!{cp}]{(\sigma, \tau, \seqS, v, L,e)}$ denote the probability that $(\sigma, \tau, \seqS, v, L)$ is a state in the random process $P^{\!{cp}}$ and $e$ is the chosen edge at that state.
Formally, for each $(\sigma, \tau, \seqS, v, L)$  where  $\abs{\seqS}=t$ and each $e\in E^{\sigma}_v$,
{\begin{align}\label{eq-def-pro-stsvle}
\Pr[\!{cp}]{(\sigma, \tau, \seqS, v, L,e)} \triangleq \Pr{\left(T > t\right)\land \left((\sigma, \tau, \seqS, v, L) = (\sigma_t, \tau_t, \seqS_t, v_t, L_t)\right)\land \left(\seqS_{t+1} =\seqS\circ e\right) }.
\end{align}}
\end{itemize}
\end{definition}

The following property on the random process $P^{\!{cp}}$ will be used in its analysis. For completeness, we provide its proof in~\Cref{sec:omitted-proofs-coupling}.

\begin{restatable}{lemma}{PropertyDefrpc}\label{lemma-property-def-rpc}
In \Cref{def:truncated-random-process}, the following properties hold:
\begin{enumerate}
\item For each $(\sigma, \tau, \seqS, v, L)$, we have 
\begin{align}\label{eq-def-trp-perperty-sigma}
\Pr[\!{cp}]{(\sigma, \tau, \seqS,v,L)}\leq \mu^{\sigma_{\bot}}_{\seqS}(\sigma), \quad \Pr[\!{cp}]{(\sigma, \tau, \seqS,v,L)}\leq\mu^{\tau_{\bot}}_{\seqS}(\tau).
\end{align}
\item For each $(\sigma, \tau, \seqS,v,L)$ and $e\in E^{\sigma}_v$, we have 
\begin{align}\label{eq-def-trp-perperty-sigma-e}
\Pr[\!{cp}]{(\sigma, \tau, \seqS,v,L,e)}\leq \mu^{\sigma_{\bot}}_{\seqS}(\sigma),\quad \Pr[\!{cp}]{(\sigma, \tau, \seqS,v,L, e)}\leq\mu^{\tau_{\bot}}_{\seqS}(\tau).
\end{align}
\end{enumerate}
\end{restatable}

\subsection{Truncated extended coupling tree}\label{sec-tct}
Given an instance $(\Phi, \sigma_\bot, \tau_\bot, v_\bot)$ satisfying~\Cref{cond-instancepair}, 
we follow the idea of Moitra~\cite{Moitra19} to estimate the marginal probability distribution $\mu_{e_{\bot}}$, which recovers the probabilities in \Cref{def-notation-trp} using a linear program.
Typically, this linear program is built based on a coupling tree that encodes all possible states of the random process $P^{\!{cp}}$.
For the linear program to be constructed efficiently, the coupling tree itself must be constructed efficiently.
However, this is a challenging task because identifying an amenable edge $e \in E_{v_t}^{\sigma_t}$ at each state $(\sigma_t,\tau_t,\seqS_t,v_t,L_t)$ as described in \Cref{def:truncated-random-process} is difficult.
The difficulty arises because it seems to bring us back to the original problem of estimating the marginal probability distribution.
To address this challenge, we extend the coupling tree by adding some extra auxiliary states, which are impossible in $P^{\!{cp}}$, allowing for its efficient construction.

We say a tuple $(\sigma, \tau, \seqS, v, L)$ is \emph{feasible} if the partial assignments $\sigma$ and $\tau$ are both feasible.
Our extended coupling tree for encoding the states of $P^{\!{cp}}$ is as follows.

\begin{definition}[$\ell$-truncated extended coupling tree] \label{def:truncated-coupling-tree}
    \emph{
    For any instance $\left(\Phi = (G, \vecf),  \sigma_\bot, \tau_\bot,v_{\bot}\right)$ satisfying \Cref{cond-instancepair} and  any positive integer $\ell$, the $\ell$-truncated extended coupling tree $\+T \triangleq \+T_{\ell}(\Phi, \sigma_\bot, \tau_\bot, v_\bot)$ is a rooted tree constructed as follows:      \begin{enumerate}
        \item The root of $\+T$ is the node with label $(\sigma_\bot, \tau_\bot, \varnothing, v_\bot, 0)$ of depth $0$.
        \item For $i = 0, 1, \cdots$: for each node $u$  of depth $i$ in the current $\+T$, suppose the label of $u$ is $(\sigma, \tau, \seqS, v, L)$.
        \begin{enumerate}
            \item          
            If $L \geq \ell$ or $E_v^{\sigma}=\emptyset$ or $(\sigma, \tau, \seqS, v, L)$ is infeasible, then $u$ is a leaf node in $\+T$;
            \item Otherwise, for each $e = \set{v,v'} \in E_v^{\sigma}$,
                \begin{enumerate}[(i)]
                    \item If ${\!{Ham}\left(\sigma,{E_v}\right)}<{\!{Ham}\left(\tau,{E_v}\right)}$, add three nodes with labels $(\sigma \land (e \gets 0), \tau \land (e \gets 0), \seqS \circ e, v, L)$, $(\sigma \land (e\gets 1), \tau \land (e \gets 0), \seqS \circ e, v', L+1)$, and $(\sigma \land (e \gets 1), \tau \land (e \gets 1), \seqS \circ e, v, L)$ as children of $u$;
                    \item Otherwise, add three nodes with labels $(\sigma \land (e \gets 0), \tau \land (e \gets 0), \seqS \circ e,v,L)$, $(\sigma \land (e \gets 0), \tau \land (e \gets 1), \seqS \circ e, v', L + 1)$, and $(\sigma \land e \gets 1, \tau \land (e \gets 1), \seqS\circ e,v,L)$ as children of $u$.
                \end{enumerate} 
        \end{enumerate}
    \end{enumerate} 
    {For simplicity, let $V(\+T)$ denote the set of nodes in $\+T$.}
    }
\end{definition}


Comparing~\Cref{def:truncated-coupling-tree} with~\Cref{def:truncated-random-process}, one can verify that the $\ell$-truncated extended coupling tree $\+T$ encodes all possible states in the $\ell$-truncated random process $P^{\cp}$. 
Furthermore, additional auxiliary states are also introduced in $\+T$.
In $P^{\cp}$, only one amenable edge in $E_{v_t}^{\sigma_t}$ selected to generate the next state at each state $(\sigma_t,\tau_t,\seqS_t,v_t,L_t)$.
In contrast, in $\+T$, all edges in $E_{v}^{\sigma}$ are enumerated to generate the children of each node labeled $(\sigma,\tau,\seqS,v,L)$ in ~\Cref{def:truncated-coupling-tree}.

By \Cref{def:truncated-coupling-tree} and the induction, one can verify the following properties:
\begin{itemize}
\item For each node $u$ with label $(\sigma,\tau,\seqS,v,L)$ and depth $i$, we have $\abs{\seqS} = i$.
Thus, any two nodes at different depths have distinct labels.
\item For each node $u\in V(\+T)$, the children of $u$ have distinct labels. Moreover, for any two nodes $u,v\in V(\+T)$ with distinct labels, the child $u'$ of $u$ and the child $v'$ of $v$ have distinct labels.
Thus, by using induction on depth, we can also prove that any two nodes at the same depth have distinct labels.
\end{itemize}
In summary, we have the following lemma.
\begin{lemma}\label{lemma-truncated-coupling-tree}
Any two different nodes in the $\ell$-truncated extended coupling tree have distinct labels.
\end{lemma}

By \Cref{lemma-truncated-coupling-tree}, one can use a label $(\sigma, \tau, \seqS, v, L)$ to refer to the unique node with this label in the $\ell$-truncated extended coupling tree $\+T$. 
Thus, we will say `the node $(\sigma, \tau, \seqS, v, L)$', `the feasible node $(\sigma, \tau, \seqS, v, L)$' rather than `the node with label $(\sigma, \tau, \seqS, v, L)$' and `the node with feasible label $(\sigma, \tau, \seqS, v, L)$', respectively.


The following property on the nodes in $V(\+T)$ is immediate by \Cref{def:truncated-coupling-tree} and the induction.
\begin{lemma}\label{lemma-property-tct-vl}
Each node $(\sigma, \tau, \seqS, v, L)\in V(\+T)$ satisfies \Cref{condition-sigma-tau}.
\end{lemma}

The following notations related to the $\ell$-truncated extended coupling tree $\+T_{\ell}$ will be used in its analysis.
\begin{definition}\label{def-notation-v-tct}
Given any $\ell$-truncated extended coupling tree $\+T = \+T_{\ell}$,
define the sets $\+V,\+L,\+L_{\!{good}}$ and the function $\+D(\cdot)$ as follows:
\begin{itemize}
\item Let $\+V$ denote the set of \emph{feasible} nodes in $V(\+T)$;
\item Let $\+L$ denote the set of leaf nodes in $V(\+T)$;
\item Define $\+L_{\!{good}}= \set{(\sigma,\tau, \seqS,v,L) \in \+L\mid L<\ell}$ and $\+L_{\!{bad}}\triangleq \+L_{\!{bad}}(\ell) = \+L\setminus \+L_{\!{good}}$.
\item Given any $(\sigma,\tau,\seqS,v,L)\in \+V\setminus \+L$, define 
\[
\+D(\sigma,\tau,\seqS,v,L)\triangleq \left\{(\sigma',\tau',\seqS',v',L')\in V(\+T) \mid \sigma' = \sigma \land (E^{\sigma}_v \gets \boldsymbol{0}),\tau' = \tau \land (E^{\sigma}_v \gets \boldsymbol{0})\right\}.
\]
\end{itemize}
\end{definition}

Intuitively, $\+L_{\!{good}}$ contains the leaves $(\sigma,\tau,\seqS,v,L)$ where $\sigma$ and $\tau$ are coupled successfully, i.e., $\mu^\sigma = \mu^\tau$.
These leaves are easily handled in our marginal ratio estimator, as discussed in \Cref{sec-mre}.
Conversely, the leaves in $\+L_{\!{bad}}$ are more challenging to handle and can introduce errors into our estimator.
The notation $\+D(\cdot)$ is crucial for bounding such errors.
We remark that if $\abs{E^{\sigma}_v} > 1$, then $\abs{\+D(\sigma,\tau,\seqS,v,L)}$ can be larger than $1$, because each permutation $e_1,e_2,\cdots,e_{\abs{E^{\sigma}_v}}$ of the edges in $E^{\sigma}_v$ results in a possible sequence $\seqS' = \seqS\circ e_1 \circ e_2\cdots\circ e_{\abs{E^{\sigma}_v}}$.

The following proposition provides some useful properties of the $\ell$-truncated extended coupling tree. We defer its proof in~\Cref{sec:omitted-proofs-coupling}.

\begin{restatable}{proposition}{PropertyTruncateTree} \label{prop:property-of-truncated-tree}
The following holds for the $\ell$-truncated extended coupling tree $\+T$:
\begin{enumerate}[(1)]
\item $V(\+T)\setminus \+L = \+V\setminus\+L$. \label{item:CT-property-1}
\item $\+T$ is of degree at most $3\Delta$, of depth at most $\Delta \ell$.
Thus, $\abs{V(\+T)}\leq \left(3\Delta\right)^{\Delta \ell + 1}$.
In addition, for each node $(\sigma,\tau,\seqS,v,L)\in \+T$, we have $\abs{\Lambda(\sigma)} = \abs{\Lambda(\tau)} \leq \Delta \ell + 1$. \label{item:CT-property-size}
\item For each node $(\sigma,\tau,\seqS,v,L)\in \+L_{\!{good}}\cap \+V$, we have $\mu^\sigma = \mu^\tau$. Thus, 
\[\frac{\mu(\tau)}{\mu(\sigma)} =  \frac{f_v\left(\abs{\tau(E_v)}\right)}{f_v\left(\abs{\sigma({E_v})}\right)}.\] \label{item:CT-property-ratio}
\end{enumerate}
\end{restatable}

\subsection{Marginals for the extended coupling tree}\label{sec-marginal}
In this section, we extend the marginals in $P^{\!{cp}}$ to the extended coupling tree and prove several useful properties. These marginals and properties form the foundation for establishing the linear program.

For simplification, we will use $(\sigma,\tau,\seqS)$ to denote $(\sigma,\tau,\seqS,v(\sigma,\tau),L(\sigma,\tau))$.
Specifically, we will use $\Pr[\!{cp}]{(\sigma, \tau, \seqS)},\Pr[\!{cp}]{(\sigma, \tau, \seqS,e)}$ to denote $\Pr[\!{cp}]{(\sigma, \tau, \seqS,v(\sigma,\tau),L(\sigma,\tau))},\Pr[\!{cp}]{(\sigma, \tau, \seqS, v(\sigma,\tau),L(\sigma,\tau),e)}$, respectively.

\begin{definition}\label{def-key-quantity}
Define the following quantities related to \Cref{def:truncated-random-process}:
\begin{itemize}
\item For each node $(\sigma, \tau, \seqS) \in \+V$, define
    \begin{align}\label{eqn-marginal-all}
        p^{\sigma}_{\sigma, \tau, \seqS} \defeq \Pr[\!{cp}]{(\sigma,\tau, \seqS)}/ \mu^{\sigma_{\bot}}_{\seqS}(\sigma), \quad p^{\tau}_{\sigma, \tau, \seqS} \defeq 
        \Pr[\!{cp}]{(\sigma, \tau, \seqS)}/\mu^{\tau_{\bot}}_{\seqS}(\tau).
    \end{align}
\item  For each node $(\sigma, \tau, \seqS) \in \+V \setminus \+L$ and $e\in E_{{v(\sigma,\tau)}}^{\sigma}$, define
\begin{align}\label{eqn-marginal-inner}
    p^{\sigma}_{\sigma, \tau, \seqS, e} \triangleq 
    \Pr[\!{cp}]{(\sigma, \tau, \seqS,e)}/\mu^{\sigma_{\bot}}_{\seqS}(\sigma),
     \quad p^{\tau}_{\sigma, \tau, \seqS, e} \triangleq 
     \Pr[\!{cp}]{(\sigma, \tau, \seqS,e) }/\mu^{\tau_{\bot}}_{\seqS}(\tau).
\end{align}
\item For each node $(\sigma, \tau, \seqS)\in V(\+T)\setminus \+V$, define $ p^{\sigma}_{\sigma, \tau, \seqS}= p^{\tau}_{\sigma, \tau, \seqS} = 0$. 
\end{itemize}
\end{definition}

\begin{remark}
\emph{
For each $(\sigma, \tau, \seqS) \in \+V\setminus \+L$ and $e\in E_{{v(\sigma,\tau)}}^{\sigma}$, by the definition of $\+V$, we have $\sigma$ is feasible. Thus, $\mu(\sigma)>0$.
In addition, recall that $\sigma(e_{\bot}) = 1$, $\sigma_{\bot} = (e_{\bot}\leftarrow 1)$.
We have  
$\mu^{\sigma_{\bot}}_{\seqS}(\sigma) \geq \mu(\sigma)>0$.
Thus, $p^{\sigma}_{\sigma, \tau, \seqS}$ and $p^{\sigma}_{\sigma, \tau, \seqS,e}$
are well defined.
Similarly, we also have $p^{\tau}_{\sigma, \tau, \seqS}$ and $p^{\tau}_{\sigma, \tau, \seqS, e}$ are well defined.
}
\end{remark}

Intuitively, the ratio $p^{\sigma}_{\sigma, \tau, \seqS}$ is a normalized probability,
where the randomness of $\sigma$ is eliminated, leaving only the randomness of $\tau$.
The same applies to $p^{\sigma}_{\sigma, \tau, \seqS}, p^{\tau}_{\sigma, \tau, \seqS},p^{\sigma}_{\sigma, \tau, \seqS, e}$ and $p^{\tau}_{\sigma, \tau, \seqS, e}$.
In the following, when we use the notations $p^{\sigma}_{\sigma, \tau, \seqS}$ and $p^{\tau}_{\sigma, \tau, \seqS}$,
we always assume  $(\sigma,\tau,\seqS)\in V(\+T)$.
Similarly, when we use the notations $p^{\sigma}_{\sigma, \tau, \seqS, e}$ and $p^{\tau}_{\sigma, \tau, \seqS, e}$, we always assume  $(\sigma,\tau,\seqS)\in \+V \setminus \+L$ and $e\in E_{v(\sigma,\tau)}^{\sigma}$.

One can verify that the following proposition holds for the above quantities. The proof of the proposition is deferred to~\Cref{sec:omitted-proofs-coupling}.

    \begin{restatable}{proposition}{LinearConstraints}\label{prop:coupling-linear-constraint}
    The following holds for the ratios in \Cref{def-key-quantity}:
        \begin{enumerate}[(1)]
            \item All $p^{\sigma}_{\sigma, \tau, \seqS}, p^{\tau}_{\sigma, \tau, \seqS},p^{\sigma}_{\sigma, \tau, \seqS,e}, p^{\tau}_{\sigma, \tau, \seqS,e}$ are in $[0, 1]$. 
            In particular, $p^{\sigma_\bot}_{\sigma_\bot, \tau_\bot, \varnothing} = p^{\tau_\bot}_{\sigma_\bot, \tau_\bot, \varnothing} = 1$.\label{prop-cp-first}
            \item For each $(\sigma, \tau, \seqS)\in \+V\setminus\+L$, let $v = v(\sigma,\tau)$. Then 
            \begin{align}\label{eqn-inter-sum1}
               p^{\sigma}_{\sigma,\tau,\seqS} = \sum_{e \in E_v^{\sigma}} p^{\sigma}_{\sigma, \tau, \seqS, e}, \quad  p^{\tau}_{\sigma,\tau,\seqS}=\sum_{e \in  E_v^{\sigma}} p^{\tau}_{\sigma,\tau, \seqS, e}.
            \end{align}
            \item For each $(\sigma, \tau, \seqS)\in \+V\setminus\+L$ and $e\in E_{v}^{\sigma}$ where $v=v(\sigma,\tau)$, if {${\!{Ham}\left(\sigma, {E_{v}}\right)} < {\!{Ham}\left(\tau,{E_{v}}\right)}$}, we have 
            \begin{align}\label{eqn-inner-child-sum1}
                p^{\sigma}_{\sigma,\tau, \seqS,e} = p^{\sigma \land (e\gets 0)}_{\sigma \land (e\gets 0),\tau\land (e\gets 0), \seqS \circ e}, \quad  p^{\sigma}_{\sigma, \tau, \seqS, e}=p^{\sigma\land (e\gets 1)}_{\sigma\land (e\gets 1),\tau\land (e\gets 0), \seqS\circ e} + p^{\sigma\land (e\gets 1)}_{\sigma\land (e\gets 1),\tau\land (e\gets 1),\seqS \circ e},
            \end{align}
            \begin{align}\label{eqn-inner-child-sum2}
                p^{\tau}_{\sigma, \tau, \seqS, e} = p^{\tau \land (e \gets 0)}_{\sigma \land (e\gets 0),\tau\land (e\gets 0), \seqS \circ e} + p^{\tau\land (e\gets 0)}_{\sigma\land (e\gets 1),\tau\land (e\gets 0), \seqS \circ e}, \quad  p^{\tau}_{\sigma, \tau, \seqS, e}=p^{\tau \land (e\gets 1)}_{\sigma \land (e\gets 1), \tau\land (e \gets 1), \seqS\circ e}.
            \end{align}
            Otherwise, we have
            \begin{align}\label{eqn-inner-child-sum3}
                p^{\sigma}_{\sigma,\tau, \seqS,e}=p^{\sigma \land (e\gets 0)}_{\sigma\land (e\gets 0),\tau\land (e\gets 0), \seqS\circ e}+p^{\sigma\land (e\gets 0)}_{\sigma\land (e\gets 0),\tau\land (e\gets 1), \seqS \circ e}, \quad  p^{\sigma}_{\sigma,\tau, \seqS,e}=p^{\sigma\land (e\gets 1)}_{\sigma\land (e\gets 1),\tau\land (e\gets 1), \seqS\circ e},
            \end{align}
            \begin{align}\label{eqn-inner-child-sum4}
                p^{\tau}_{\sigma,\tau, \seqS,e}=p^{\tau\land (e\gets 0)}_{\sigma\land (e\gets 0),\tau\land (e\gets 0), \seqS\circ e}, \quad  p^{\tau}_{\sigma,\tau,S,e}=p^{\tau\land (e\gets 1)}_{\sigma\land (e\gets 0),\tau\land (e\gets 1), \seqS \circ e} + p^{\tau \land (e\gets 1)}_{\sigma\land (e\gets 1),\tau\land (e\gets 1), \seqS\circ e}.
            \end{align}
            \item For each $(\sigma, \tau, \seqS) \in \+V$, we have
            \begin{align}\label{eqn-ratio}
            	{p^{\sigma}_{\sigma,\tau, \seqS}} = p^{\tau}_{\sigma,\tau, \seqS} \cdot \frac{\mu_{e_{\bot}}(1)}{\mu_{e_{\bot}}(0)}\cdot \frac{ \mu(\tau)}{ \mu(\sigma)}.
            \end{align}
        \end{enumerate}
    \end{restatable}
    
    Recall the definitions of $B$ and $\+D(\cdot)$ in \eqref{eq-def-bmin} and \Cref{def-notation-v-tct}. 
    The next lemma provides a key property for bounding the error of our marginal ratio estimator in \Cref{sec-mre}, introduced by the bad leaves in $\+L_{\!{bad}}$.

    \begin{restatable}{lemma}{CouplingError} \label{lem:coupling-error}
       For each $(\sigma, \tau, \seqS) \in \+V\setminus \+L$,
       let $\+D \triangleq \+D(\sigma, \tau, \seqS)$.
       Then $\+D\neq \emptyset$. In addition, 
       \begin{align}\label{eqn-error-bound}
           \sum_{(\sigma', \tau', \seqS') \in \+D} p^{\sigma'}_{\sigma', \tau', \seqS'}\geq  B \cdot p^{\sigma}_{\sigma, \tau, \seqS}, \quad \sum_{(\sigma',\tau', \seqS') \in \+D} p^{\tau'}_{\sigma', \tau', \seqS} \geq  B \cdot p^{\tau}_{\sigma,\tau, \seqS}.
       \end{align}
    \end{restatable}

    The lemma states a lower bound for the probability of a successful coupling at the node $(\sigma,\tau,\seqS)$, which is  derived from the marginal bound in~\Cref{lem:marginal-bound}. A formal proof of \Cref{lem:coupling-error} is provided in \Cref{sec:omitted-proofs-coupling}.

\section{The Marginal Ratio Estimator via Linear Program}\label{sec:LP}
In this section, we design a deterministic algorithm to efficiently approximate the marginal ratio for any Holant instance satisfying~\Cref{cond:Holant-condition}. We start by establishing a linear program based on the truncated extended coupling tree, and then design a marginal estimator using this linear program. This section is dedicated to proving the following theorem.

\begin{theorem} \label{thm:Holant-marginal-ratio-estimator}
    There exists a deterministic algorithm $\+A$ such that given as input any $\varepsilon\in(0,1/4)$ and any instance $\Phi = (G=(V,E), \vecf)$ satisfying~\Cref{cond:Holant-condition} with any edge $e\in E$, it outputs a number $\wh{R}$ such that
    $$
        (1 - \varepsilon) R_{\Phi}(e) \le \wh{R} \le (1 + \varepsilon) R_{\Phi}(e),
    $$
    within time $O\left(\abs{V} \cdot \varepsilon^{-\RunningTimeExponent}\right)$.
\end{theorem}



\subsection{Setting up the linear program}
We first introduce a linear program induced by the coupling process $P^{\!{cp}}$ described in \Cref{def:truncated-random-process}.
Our linear program is built on the extended coupling tree of $P^{\!{cp}}$ rather than the original coupling tree.


\begin{definition}[Linear program induced by the coupling] \label{def:induced-LP}
    For any instance $\left(\Phi = (G, \vecf),  \sigma_\bot, \tau_\bot,v_{\bot}\right)$ satisfying \Cref{cond-instancepair} and  any positive integer $\ell$, 
    let $\+T = \+T_{\ell}(\Phi, \sigma_\bot, \tau_\bot, v_\bot)$ be the $\ell$-truncated extended coupling tree in \Cref{def:truncated-coupling-tree} and $0 \leq r^- \leq r^+$ be two parameters.
    The LP induced by the $\ell$-truncated random process in \Cref{def:truncated-random-process} is as follows.
    The variables of the LP are:
    \begin{itemize}
    \item For each $(\sigma, \tau, \seqS) \in V(\+T)$, there are two variables $\widehat{p}_{\sigma, \tau, \seqS}^{\sigma}$, $\widehat{p}_{\sigma,\tau, \seqS}^{\tau}$.
    \item For each $(\sigma,\tau,\seqS)\in \+V\setminus \+L$ and each $e\in E_{v(\sigma,\tau)}^{\sigma}$, there are two variables $\widehat{p}_{\sigma,\tau,\seqS,e}^{\sigma}$, and $\widehat{p}_{\sigma,\tau,\seqS,e}^{\tau}$.
    \end{itemize}
    The constraints of the LP are:
    \begin{enumerate}
        \item All $\widehat{p}^{\sigma}_{\sigma, \tau, \seqS}, \widehat{p}^{\tau}_{\sigma, \tau, \seqS},\widehat{p}^{\sigma}_{\sigma, \tau, \seqS,e}, \widehat{p}^{\tau}_{\sigma, \tau, \seqS,e}$ are in $[0, 1]$. In particular, $\widehat{p}^{\sigma_\bot}_{\sigma_\bot,\tau_\bot, \varnothing}=\widehat{p}^{\tau_\bot}_{\sigma_\bot,\tau_\bot,\varnothing} = 1$.\label{item-first-LP}
        \item For each $(\sigma, \tau, \seqS)\in \+V\setminus\+L$, 
        let $v=v(\sigma,\tau)$. Then
            \begin{align}\label{eqn-hat-inter-sum1}
               \widehat{p}^{\sigma}_{\sigma, \tau, \seqS} = \sum_{e \in  E_v^{\sigma}} \widehat{p}^{\sigma}_{\sigma, \tau, \seqS, e}, \quad  \widehat{p}^{\tau}_{\sigma,\tau, \seqS}=\sum_{e\in  E_v^{\sigma}} \widehat{p}^{\tau}_{\sigma,\tau, \seqS,e}.
            \end{align}\label{item-second-LP}
        \item For each $(\sigma, \tau, \seqS)\in \+V\setminus\+L$ and $e\in E_{v}^{\sigma}$ where $v = v(\sigma,\tau)$,
            if ${\!{Ham}\left(\sigma,{E_{v}}\right)}<{\!{Ham}\left(\tau,{E_{v}}\right)}$, then
            \begin{align}\label{eqn-hat-inner-child-sum1}
                \widehat{p}^{\sigma}_{\sigma,\tau,\seqS,e}=\widehat{p}^{\sigma\land (e\gets 0)}_{\sigma\land (e\gets 0),\tau\land (e\gets 0), \seqS \circ e}, \quad  \widehat{p}^{\sigma}_{\sigma,\tau, \seqS,e}=\widehat{p}^{\sigma\land (e\gets 1)}_{\sigma\land (e\gets 1),\tau\land (e\gets 0), \seqS \circ e} + \widehat{p}^{\sigma \land (e\gets 1)}_{\sigma\land (e\gets 1),\tau \land (e\gets 1), \seqS \circ e},
            \end{align}
            \begin{align}\label{eqn-hat-inner-child-sum2}
                \widehat{p}^{\tau}_{\sigma,\tau, \seqS,e}=\widehat{p}^{\tau\land (e\gets 0)}_{\sigma\land (e\gets 0),\tau\land (e\gets 0), \seqS\circ e}+\widehat{p}^{\tau\land (e\gets 0)}_{\sigma\land (e\gets 1),\tau\land (e\gets 0), \seqS \circ e}, \quad  \widehat{p}^{\tau}_{\sigma,\tau,\seqS,e}=\widehat{p}^{\tau\land (e\gets 1)}_{\sigma\land (e\gets 1),\tau\land (e\gets 1),\seqS\circ e}.
            \end{align}
            Otherwise, 
            \begin{align}\label{eqn-hat-inner-child-sum3}
                 \widehat{p}^{\sigma}_{\sigma,\tau,\seqS,e}= \widehat{p}^{\sigma \land (e\gets 0)}_{\sigma\land (e\gets 0),\tau\land (e\gets 0),\seqS\circ e}+ \widehat{p}^{\sigma\land (e\gets 0)}_{\sigma\land (e\gets 0),\tau\land (e\gets 1), \seqS\circ e}, \quad  \widehat{p}^{\sigma}_{\sigma,\tau,\seqS,e}=\widehat{p}^{\sigma\land (e\gets 1)}_{\sigma\land (e\gets 1),\tau\land (e\gets 1), \seqS\circ e},
            \end{align}
            \begin{align}\label{eqn-hat-inner-child-sum4}
                 \widehat{p}^{\tau}_{\sigma,\tau,\seqS,e}= \widehat{p}^{\tau\land (e\gets 0)}_{\sigma\land (e\gets 0),\tau\land (e\gets 0),\seqS\circ e}, \quad   \widehat{p}^{\tau}_{\sigma,\tau,\seqS,e}= \widehat{p}^{\tau\land (e\gets 1)}_{\sigma\land (e\gets 0),\tau\land (e\gets 1),\seqS\circ e} +  \widehat{p}^{\tau \land (e\gets 1)}_{\sigma\land (e\gets 1),\tau\land (e\gets 1),\seqS\circ e}.
            \end{align} \label{item-third-LP}
        \item For any $(\sigma,\tau,\seqS)\in \+L_{\!{good}}\cap \+V$,
            \begin{align}\label{eqn-hat-ratio}
                 \frac{ \mu(\tau)}{ \mu(\sigma)}\cdot r^-\cdot {\widehat{p}^{\tau}_{\sigma,\tau, \seqS}}\leq {\widehat{p}^{\sigma}_{\sigma,\tau, \seqS}}\leq  \frac{ \mu(\tau)}{ \mu(\sigma)}\cdot  r^+ \cdot{\widehat{p}^{\tau}_{\sigma,\tau, \seqS}}.
            \end{align}\label{item-forth-LP}
        \item For any $(\sigma,\tau, \seqS)\in \+V\setminus \+L$ and $\+D \triangleq \+D(\sigma,\tau, \seqS)$,
           \begin{align}\label{eqn-hat-error-bound}
               \sum_{(\sigma',\tau',\seqS')\in \+D}\widehat{p}^{\sigma'}_{\sigma',\tau', \seqS'} \geq  B \cdot \widehat{p}^{\sigma}_{\sigma,\tau, \seqS}, \quad \sum_{(\sigma',\tau',\seqS')\in \+D}\widehat{p}^{\tau'}_{\sigma',\tau',\seqS'}\geq  B \cdot \widehat{p}^{\tau}_{\sigma,\tau,\seqS}.
           \end{align}\label{item-fifth-LP}
        \item  For each node $(\sigma, \tau, \seqS)$ in $V(\+T)\setminus \+V$, $ \widehat{p}^{\sigma}_{\sigma,\tau, \seqS}= \widehat{p}^{\tau}_{\sigma,\tau, \seqS}=0$. \label{item-sixth-LP}
    \end{enumerate}
\end{definition}

Unlike typical LPs in the literature~\cite{Moitra19,GLLZ19,GGGY21,JPV21,WY24}, we introduce a new set of auxiliary constraints to mimic the fractional choice of edges in~\eqref{eqn-hat-inter-sum1}.
Moreover, the above LP bounds the error caused by the coupling in a few steps within a single constraint \eqref{eqn-hat-error-bound}.
These design reflects the nature of the coupling in~\cite{CG24bMatching} and the structure of the extended coupling tree.
As we shall see, the constraints presented in~\eqref{eqn-hat-error-bound} play a crucial role in bounding the deviation of our bootstrap paradigm from the coupling procedure $P^{\!{cp}}$, particularly in scenarios where the ``wrong'' edges may be selected.

\subsection{The analysis of the linear program}
Recall the definition of the marginal ratio $R_{\Phi}(e_\bot)$ in \eqref{def-rphie}. In this section, 
we demonstrate that the feasibility of the LP in \Cref{def:induced-LP} can be efficiently checked. Furthermore, if the LP is feasible, one can obtain a good approximation of $R_{\Phi}(e_\bot)$.

At first, we show the efficiency of constructing and verifying the feasibility of the LP.
\begin{restatable}[Cost for constructing LP]{lemma}{BuildingCost} 
\label{lem:building-cost-of-LP}
    For any instance $(\Phi, \sigma_\bot, \tau_\bot, v_\bot)$ satisfying~\Cref{cond-instancepair}, any non-negative integer $\ell \ge 0$ and any two real numbers $0 \le r^- \le r^+$,
    the LP in \Cref{def:induced-LP} can be constructed in time $\poly\left(\Delta^{\Delta \ell}\right)$. Thus one can check the feasibility of the LP in time $\poly\left(\Delta^{\Delta \ell}\right)$.
\end{restatable}

The following lemma ensures the feasibility of the LP by the coupling process.
\begin{restatable}[Feasibility of LP]{lemma}{LPFeasibility}\label{lem:feasibility-of-LP}
    If $r^- \le R_\Phi(e_\bot) \le r^+$, a feasible solution to the LP in~\Cref{def:induced-LP} exists.
\end{restatable}
For completeness, we provide the full proofs of the above two lemmas in~\Cref{sec:omitted-proofs-LP}.

To approximate the marginal ratio $R_{\Phi}(e_\bot)$, the key property is that a feasible solution to the LP provides both upper and lower bounds for it.
\begin{theorem} \label{thm:ratio-bound-by-LP}
    Assume that all the constraints of the LP in~\Cref{def:induced-LP} hold. Then it holds that
    \begin{align*}
        r^{-} \left(1 - \left(1 - B^2\right)^{\ell}\right)\le R_{\Phi}(e_\bot) \le {r^+}\left(1 - \left(1 - B^2\right)^{\ell}\right)^{-1}.
    \end{align*}
\end{theorem}

To prove~\Cref{thm:ratio-bound-by-LP}, we show some properties of the feasible solution to the LP.

\begin{lemma} \label{lem:ratio-identity}
     Assume that all the constraints of the LP in~\Cref{def:induced-LP} hold. Then it holds that
    \begin{align*}
        \mu_{e_{\bot}}(1) = \sum_{(\sigma, \tau, \seqS) \in \+L}\mu(\sigma)\cdot\widehat{p}_{\sigma, \tau, \seqS}^{\sigma},
    \end{align*}
    \begin{align*}
        \mu_{e_{\bot}}(0)  = \sum_{(\sigma, \tau, \seqS) \in \+L}\mu(\tau)\cdot\widehat{p}_{\sigma, \tau, \seqS}^{\tau}.
    \end{align*}
\end{lemma}

The following lemma is the key ingredient to prove~\Cref{lem:ratio-identity}. 

\begin{restatable}{lemma}{RatioSum} \label{lem:ratio-identity-partial}
    Assume that all the constraints of the LP in~\Cref{def:induced-LP} hold. Then it holds that
    \begin{align}\label{eqn-sum-leftside}
        \forall x\in \sigma_\bot, \quad \sum_{(\sigma,\tau, \seqS)\in \+L: \ x\in \sigma} \widehat{p}^{\sigma}_{\sigma,\tau,\seqS}=1,
    \end{align}
    \begin{align}\label{eqn-sum-rightside}
        \forall y\in \tau_\bot, \quad \sum_{(\sigma,\tau, \seqS)\in \+L: \ y\in \tau} \widehat{p}^{\tau}_{\sigma,\tau, \seqS}=1.
    \end{align}
\end{restatable}

A formal proof of this lemma is in \Cref{sec:omitted-proofs-LP}. Here we give an intuitive explanation of \eqref{eqn-sum-leftside}. For any $x\in\sigma_\bot$, the collection of values $\wh{p}_{\sigma, \tau, \seqS}^\sigma$ and $\wh{p}_{\sigma, \tau, \seqS, e}^\sigma$ induces a random walk starting from $(\sigma_\bot,\tau_\bot,\varnothing)$ on the $\ell$-truncated extended coupling tree:
\begin{itemize}
    \item When the current node is $(\sigma,\tau,\seqS)$, move to $(\sigma,\tau,\seqS,e)$ with probability proportional to $\wh{p}_{\sigma, \tau, \seqS, e}^\sigma$ according to the rule in \eqref{eqn-hat-inter-sum1}.
    \item When the current node is $(\sigma,\tau,\seqS,e)$, move to a node $(\sigma',\tau',\seqS')$ such that $x\in \sigma'$ according to the rule in \eqref{eqn-hat-inner-child-sum1} or \eqref{eqn-hat-inner-child-sum3}, depending on whether ${\!{Ham}\left(\sigma,{E_{v}}\right)}<{\!{Ham}\left(\tau,{E_{v}}\right)}$ or not.
\end{itemize}
The \eqref{eqn-sum-leftside} simply states that the random walk terminates at one of nodes in $\+L$ with probability $1$.

\begin{proof}[Proof of~\Cref{lem:ratio-identity}]
    We have
    \begin{align*}
        \sum_{(\sigma, \tau, \seqS) \in \mathcal L} \mu(\sigma)\cdot \widehat{p}_{\sigma, \tau, \seqS}^{\sigma} 
        = \sum_{(\sigma, \tau, \seqS) \in \mathcal L} \widehat{p}_{\sigma, \tau, \seqS}^{\sigma} \sum_{x \in \sigma} \mu(x) 
        = \sum_{(\sigma, \tau, \seqS) \in \mathcal L} \widehat{p}_{\sigma, \tau, \seqS}^{\sigma} \sum_{x \in \sigma_\bot} \left(\mu(x) \cdot  \id{x \in \sigma}\right).
    \end{align*}
    Swapping the summations, we obtain
    \begin{align*}
    \sum_{(\sigma, \tau, \seqS) \in \mathcal L} \mu(\sigma)\cdot \widehat{p}_{\sigma, \tau, \seqS}^{\sigma} 
        &= \sum_{x \in \sigma_\bot} \mu(x) \sum_{(\sigma, \tau, \seqS) \in \mathcal L} \left(\widehat{p}_{\sigma, \tau, \seqS}^{\sigma} \cdot \id{x \in \sigma}\right) \\
        &= \sum_{x \in \sigma_\bot} \mu(x) \sum_{(\sigma, \tau, \seqS) \in \+L : \ x \in \sigma} \widehat{p}_{\sigma, \tau, \seqS}^{\sigma} \\
        &= \sum_{x \in \sigma_\bot} \mu(x) \tag{by~\eqref{eqn-sum-leftside}} \\
        &= \mu_{e_\bot}(1).
    \end{align*}
    Similarly, by \eqref{eqn-sum-rightside}, we also have 
    \begin{align*}
        \sum_{(\sigma, \tau, \seqS) \in \+L}\mu(\tau)\cdot\widehat{p}_{\sigma, \tau, \seqS}^{\tau} = \mu_{e_{\bot}}(0).
    \end{align*}
    The lemma is proved.
\end{proof}

By~\Cref{lem:ratio-identity}, to bound the marginal ratio $R_{\Phi}(e_\bot) = \mu_{e_\bot}(1) / \mu_{e_\bot}(0)$, it suffices to bound the ratio $\left(\mu(\sigma) \cdot \wh{p}_{\sigma, \tau, \seqS}^{\sigma}\right) / \left(\mu(\tau) \cdot \wh{p}_{\sigma, \tau, \seqS}^{\tau}\right)$ for each leaf $(\sigma, \tau, \seqS) \in \+L$ as in \eqref{eqn-hat-ratio}. 
For each good leaf $(\sigma, \tau, \seqS) \in \+L_{\good}$, it is straightforward to verify whether \eqref{eqn-hat-ratio} holds, because $\mu(\sigma)/\mu(\tau)$ can be calculated efficiently by Item~\eqref{item:CT-property-ratio} in~\Cref{prop:property-of-truncated-tree}.
However, for a bad leaf $(\sigma, \tau, \seqS) \in \+L_{\bad}$, it is challenging to determine whether \eqref{eqn-hat-ratio} holds, as we do not know how to estimate $\mu(\sigma)/\mu(\tau)$.  The following lemma allows us to disregard all bad leaves at the cost of a small error. 

\begin{restatable}{lemma}{LPTruncatedError} \label{lem:LP-truncated-error}
     Assume that all the constraints of the LP in~\Cref{def:induced-LP} hold. Then it holds that
    \begin{align}\label{eqn-error1}
        \sum_{(\sigma,\tau, \seqS)\in \+L_{\!{bad}}} \widehat{p}_{\sigma,\tau, \seqS}^{\sigma}\cdot \mu^{\sigma_{\bot}}_{\seqS}(\sigma)\leq (1 - B^2)^{\ell},
    \end{align}
    \begin{align}\label{eqn-error2}
        \sum_{(\sigma,\tau, \seqS)\in \+L_{\!{bad}}}\widehat{p}_{\sigma,\tau, \seqS}^{\tau}\cdot \mu^{\tau_{\bot}}_{\seqS}(\tau)\leq (1- B^2)^{\ell}.
    \end{align}
\end{restatable}

The proof of \Cref{lem:LP-truncated-error} is via induction on $\ell$ and is provided in \Cref{sec:omitted-proofs-LP}.
Here we provide an intuitive explanation of~\eqref{eqn-error1}. 
The collection of values $\wh{p}_{\sigma, \tau, \seqS}^\sigma$ and $\wh{p}_{\sigma, \tau, \seqS, e}^\sigma$ induces a random walk starting from $(\sigma_\bot,\tau_\bot,\varnothing)$ to a leaf node on the $\ell$-truncated extended coupling tree:
\begin{itemize}
    \item When the current node is $(\sigma, \tau, \seqS)$, move to $(\sigma, \tau, \seqS, e)$ with probability proportional to $\wh{p}_{\sigma, \tau, \seqS, e}^{\sigma}$ according to the rule in \eqref{eqn-hat-inter-sum1}.
    \item When the current node is $(\sigma, \tau, \seqS, e)$, firstly sample the assignment $\sigma_e$ on $e$ according to the marginal distribution $\mu_{e}^\sigma$, and then move to $(\sigma \land \sigma_e, \tau \land \tau_e, \seqS \circ e)$ with probability proportional to $\wh{p}_{\sigma \land \sigma_e, \tau \land \tau_e, \seqS \circ e}^{\sigma \land \sigma_e}$ according to the rule in \eqref{eqn-hat-inner-child-sum1} or \eqref{eqn-hat-inner-child-sum3}, depending on whether ${\!{Ham}\left(\sigma,{E_{v}}\right)}<{\!{Ham}\left(\tau,{E_{v}}\right)}$ or not.
\end{itemize}
For any bad leaf $(\sigma, \tau, \seqS)\in \+L_{\!{bad}}$, we have $L(\sigma,\tau) = \ell$.
Thus, before reaching any bad leaf from $(\sigma_\bot,\tau_\bot,\varnothing)$, the random work will traverse at least $\ell$ internal node $(\sigma', \tau', \seqS')\in \+V \setminus \+L$ along the path. 
At each internal node $(\sigma', \tau', \seqS')$, after \emph{several} rounds of two steps moving, the probability of arriving at a node in $\+D(\sigma', \tau', \seqS') \subseteq \+L_{\good}$ can be lower bounded by $B^2$.
Here, one $B$ is due to Constraint~\ref{item-fifth-LP}, and the other $B$ is due to the marginal lower bound in~\Cref{lem:marginal-bound}.
Thus, the probability of reaching a bad leaf is no more than $(1 - B^2)^{\ell}$.
The left side of~\eqref{eqn-error1} represents the probability that the random walk arrives at a bad leaf.
Hence, we obtain the upper bound in~\eqref{eqn-error1}.

Now we prove~\Cref{thm:ratio-bound-by-LP} with the above lemmas.
\begin{proof}[Proof of~\Cref{thm:ratio-bound-by-LP}]
    By~\Cref{lem:ratio-identity}, it holds that
    \begin{align*}
        \mu_{e_\bot}(1) &= \sum_{(\sigma, \tau, \seqS) \in \+L} \wh{p}_{\sigma, \tau, \seqS}^{\sigma} \cdot \mu(\sigma) \\
        &= \sum_{(\sigma, \tau, \seqS) \in \+L_{\good}} \wh{p}_{\sigma, \tau, \seqS}^{\sigma} \cdot  \mu(\sigma) + \sum_{(\sigma, \tau, \seqS) \in \+L_{\bad}} \wh{p}_{\sigma, \tau, \seqS}^{\sigma} \cdot  \mu(\sigma) \\
        &= \sum_{(\sigma, \tau, \seqS) \in \+L_{\good}} \wh{p}_{\sigma, \tau, \seqS}^{\sigma} \cdot \mu(\sigma) + \mu_{e_\bot}(1) \sum_{(\sigma, \tau, \seqS) \in \+L_{\bad}} \wh{p}_{\sigma, \tau, \seqS}^{\sigma} \cdot \mu_{\seqS}^{\sigma_{\bot}}(\sigma)\\
        &\le \sum_{(\sigma, \tau, \seqS) \in \+L_{\good}} \wh{p}_{\sigma, \tau, \seqS}^{\sigma} \cdot \mu(\sigma) + \left(1 - B^2\right)^{\ell} \mu_{e_\bot}(1)
    \end{align*}
    where the last inequality holds from~\Cref{lem:LP-truncated-error}. Thus we have
    $$
        \sum_{(\sigma, \tau, \seqS) \in \+L_{\good}} \wh{p}_{\sigma, \tau, \seqS}^{\sigma} \cdot \mu(\sigma) \le \mu_{e_\bot}(1) \le \frac{1}{1 - \left(1 - B^2\right)^{\ell}} \sum_{(\sigma, \tau, \seqS) \in \+L_{\good}} \wh{p}_{\sigma, \tau, \seqS}^{\sigma} \cdot \mu(\sigma).
    $$
    Similarly, we also have
    $$
        \sum_{(\sigma, \tau, \seqS) \in \+L_{\good}} \wh{p}_{\sigma, \tau, \seqS}^{\tau} \cdot\mu(\tau)\le \mu_{e_\bot}(0) \le \frac{1}{1 - \left(1 - B^2\right)^{\ell}} \sum_{(\sigma, \tau, \seqS) \in \+L_{\good}} \wh{p}_{\sigma, \tau, \seqS}^{\tau} \cdot \mu(\tau).
    $$
    Hence, by \eqref{eqn-hat-ratio} we have
    \begin{align*}
        \frac{\mu_{e_\bot}(1)}{\mu_{e_\bot}(0)} \le \frac{1}{1 - \left(1 - B^2\right)^{\ell}} \frac{\sum_{(\sigma, \tau, \seqS) \in \+L_{\good}} \wh{p}_{\sigma, \tau, \seqS}^{\sigma} \cdot \mu(\sigma)}{\sum_{(\sigma, \tau, \seqS) \in \+L_{\good}} \wh{p}_{\sigma, \tau, \seqS}^{\tau} \cdot \mu(\tau)} 
        \le \frac{r^+}{1 - \left(1 - B^2\right)^{\ell}},
    \end{align*}
    and
    \begin{align*}
       \frac{\mu_{e_\bot}(1)}{\mu_{e_\bot}(0)} \ge \left(1 - \left(1 - B^2\right)^{\ell}\right) \frac{\sum_{(\sigma, \tau, \seqS) \in \+L_{\good}} \wh{p}_{\sigma, \tau, \seqS}^{\sigma} \cdot \mu(\sigma)}{\sum_{(\sigma, \tau, \seqS) \in \+L_{\good}} \wh{p}_{\sigma, \tau, \seqS}^{\tau} \cdot \mu(\tau)} 
        \ge \left(1 - \left(1 - B^2\right)^{\ell}\right) r^-.
    \end{align*}
\end{proof}

\subsection{The marginal ratio estimator}\label{sec-mre}

Given any instance satisfying~\Cref{cond:Holant-condition}, we are now ready to complete the proof of~\cref{thm:Holant-marginal-ratio-estimator} by employing the linear program in \Cref{def:induced-LP}. The key ingredient to \Cref{thm:Holant-marginal-ratio-estimator} is the following estimator for the marginal ratio of the half-edge based on the linear program.

\begin{lemma} \label{lem:LP-ratio-estimator}
    There exists a deterministic algorithm $\+A_{\bot}$ such that given as input any $\varepsilon\in(0,1/4)$ and any instance $(\Phi = (G, \vecf), \sigma_\bot, \tau_\bot, v_\bot)$ satisfying~\Cref{cond-instancepair}, it outputs a number $\wh{R}$ such that
    $$
        (1 - \varepsilon) R_{\Phi}(e_\bot) \le \wh{R} \le (1 + \varepsilon) R_{\Phi}(e_\bot),
    $$
   within time $\abs{V(G)}\cdot\varepsilon^{-\poly(\Delta(G),1/B(\Phi))}$.
\end{lemma}

\begin{proof}

    Recall the definition of $B(\Phi)$ in \eqref{eq-def-bmin}.
    For simplicity, let $\Delta \triangleq \Delta(G)$, $B \triangleq B(\Phi)$ and $r=r_{\max}(\Phi)$.
    Set
     \begin{align}\label{eq-def-ell}
    \ell = \left\lceil \frac{\log{\varepsilon} - \log 2}{\log\left(1 - B^2\right)} \right\rceil
    \end{align}
    The following binary search calculates $\wh{R}$. Initially, set $r^1_1 \gets 0$ and $r^1_2 \gets r$ and $i \gets 1$. Then, repeat the following steps: 
    \begin{enumerate}[(1)]
        \item Let $\mathsf{LP1}$ denote the LP in~\Cref{def:induced-LP} for the instance $(\Phi, \sigma_\bot, \tau_\bot, v_\bot)$ with parameters $r^- = r^i_1,r^+ = (r^i_1 + r^i_2)/2$ and $\ell$.
        Let $\mathsf{LP2}$ denote the LP in~\Cref{def:induced-LP} for the instance $(\Phi, \sigma_\bot, \tau_\bot, v_\bot)$ with parameters $r^- = (r^i_1+r^i_2)/2,r^+ = r^i_2$ and $\ell$.
        Decide the feasibility of $\mathsf{LP1}$ and $\mathsf{LP2}$.
        \item If both $\mathsf{LP1}$ and $\mathsf{LP2}$ are feasible or $r^i_1 \geq r^i_2\left(1 - \left(1 - B^2\right)^{\ell}\right)$, let $\wh{R} = (r^i_1+r^i_2)/2$ and
        terminate the binary search. Otherwise, proceed to the next step. \label{enu-bs-2}
        \item If $\mathsf{LP1}$ is feasible, let $r^{i+1}_2 \leftarrow (r^i_1 + r^i_2)/2$; otherwise, let $r^{i+1}_1 \leftarrow (r^i_1 + r^i_2)/2$. Let $i \leftarrow i+1$.
        \label{enu-bs-3}
    \end{enumerate}
    We claim $r^i_1 \leq R_\Phi(e_\bot) \leq r^i_2$ for each $i>0$ in the binary search.
    We prove the claim by induction.
    For the base case, by~\Cref{lem:marginal-ratio-upper-bound}, 
    we have $0\leq R_{\Phi}(e_\bot)\leq r_{\max}(\Phi) = r$.
    Thus, initially we have $r^1_1 \leq R_\Phi(e_\bot) \leq r^1_2$.
    For the induction step, by ~\Cref{lem:feasibility-of-LP}, if $r^- \le R_\Phi(e_\bot) \le r^+$, then a feasible solution to the LP exists.
    Combined with the induction hypothesis, we have 
    at least one of $\mathsf{LP1}$ and $\mathsf{LP2}$ is feasible in round $i$.
    In \Cref{enu-bs-3}, at most one of $\mathsf{LP1}$ and $\mathsf{LP2}$ is feasible. 
    Combined with that at least one of them is feasible,     
    we have exactly one of $\mathsf{LP1}$ and $\mathsf{LP2}$ is feasible in \Cref{enu-bs-3}.
    Assume that $\mathsf{LP1}$ is feasible in \Cref{enu-bs-3}.
    Then $\mathsf{LP2}$ is infeasible.
    Combined with ~\Cref{lem:feasibility-of-LP}, we have 
    either $R_\Phi(e_\bot) > r^i_2$ or $R_\Phi(e_\bot) < (r^i_1 + r^i_2)/2$. Combined with the induction hypothesis,
    we have $r^i_1 \leq R_\Phi(e_\bot) \leq (r^i_1 + r^i_2)/2$.
    Thus, we have $r^{i+1}_1 \leq R_\Phi(e_\bot) \leq r^{i+1}_2$. 
    Similarly, if $\mathsf{LP1}$ is infeasible,
    we also have $r^{i+1}_1 \leq R_\Phi(e_\bot) \leq r^{i+1}_2$.
    
    When the binary research stops, we have either both $\mathsf{LP1}$ and $\mathsf{LP2}$ are feasible or $r^i_1 \geq r^i_2\left(1 - \left(1 - B^2\right)^{\ell}\right)$.
    If both $\mathsf{LP1}$ and $\mathsf{LP2}$ are feasible,
    by \Cref{thm:ratio-bound-by-LP} we have
    \begin{gather*}
        r^i_1 \left(1 - \left(1 - B^2\right)^{\ell}\right)\le R_{\Phi}(e_\bot) \le \frac{r^i_1+r^i_2}{2}\left(1 - \left(1 - B^2\right)^{\ell}\right)^{-1},\\
        \frac{r^i_1+r^i_2}{2} \left(1 - \left(1 - B^2\right)^{\ell}\right)\le R_{\Phi}(e_\bot) \le r^i_2\left(1 - \left(1 - B^2\right)^{\ell}\right)^{-1}.
    \end{gather*}
    Combined with $\wh{R} = (r^i_1+r^i_2)/2$, we have
    \begin{align}\label{eq-rhat-rphi-relation}
        \wh{R} \left(1 - \left(1 - B^2\right)^{\ell}\right)\le R_{\Phi}(e_\bot) \le \wh{R}\left(1 - \left(1 - B^2\right)^{\ell}\right)^{-1}.
    \end{align}
    If $r^i_1 \geq r^i_2\left(1 - \left(1 - B^2\right)^{\ell}\right)$, 
    we have 
    \[r^i_2\left(1 - \left(1 - B^2\right)^{\ell}\right)\leq \frac{r^i_1+r^i_2}{2} \leq r^i_1 \left(1 - \left(1 - B^2\right)^{\ell}\right)^{-1}.\]
    Combined with $r^i_1 \leq R_\Phi(e_\bot) \leq r^i_2$ and $\wh{R} = (r^i_1+r^i_2)/2$, we also have \eqref{eq-rhat-rphi-relation}.
    Moreover, by \eqref{eq-def-ell} one can verify  
    \begin{equation}
    \begin{aligned}
    \left(1 - \left(1 - B^2\right)^{\ell}\right) \geq 1 - \varepsilon/2
    \end{aligned}
    \end{equation}
    Combined with \eqref{eq-rhat-rphi-relation},
    we have 
    $$
        (1 - \varepsilon) R_{\Phi}(e_\bot) \le \wh{R} \le (1 + \varepsilon) R_{\Phi}(e_\bot).
    $$
    
    Now we turn to the running time of the algorithm. 
    Recall that we always assume that the signature $f_v = [f_v(0),\cdots,f_v(d)]$ at $v$ satisfies $d = \deg_G(v)\leq \Delta$ for each vertex $v\in V(G)$.
    Combining \eqref{eq-def-bmin}, \eqref{eq:local-polynomial} with the the assumption,
    one can verify that $B$ and $\ell$ can be calculated within time $O(\abs{V(G)}\cdot 2^{\Delta})$.
    Recall that if $r^i_1 \geq r^i_2\left(1 - \left(1 - B^2\right)^{\ell}\right) = r^i_2(1 - \varepsilon/2)$, the binary search stops.
    Then one can verify that the binary search runs for at most $O(\log \frac{2}{\varepsilon}) =  O(\log \frac{1}{\varepsilon})$ rounds. 
    For each round of the binary search, by \Cref{lem:building-cost-of-LP} we have the time cost is $\poly\left(\Delta^{\Delta \ell}\right)$.
    Meanwhile, we have
    \begin{align*}
        &\symbolwidth \Delta^{\Delta \ell} \\
        \left(\text{by \eqref{eq-def-ell}}\right)\quad   &\le \exp\left(\Delta \log{\Delta} \left(1 + \frac{\log{\varepsilon} - \log 2}{\log\left(1 - B^2\right)}\right)\right) \\
        (\mbox{by $\log{(1 - x)} \leq  -x$ for $x\in (0,1)$}) \quad
        &\le \exp\left(\Delta \log{\Delta} \left(1 + (\log 2 - \log{\varepsilon})/B^{2}\right)\right) \\
        &= \varepsilon^{-\poly(\Delta, 1/B)}.
    \end{align*}
    Hence, the total cost of the algorithm is 
    \[O\left(2^{\Delta}\abs{V(G)} +\log \frac{1}{\varepsilon} \cdot \poly\left(\Delta^{\Delta \ell}\right)\right) = O\left(2^{\Delta}\abs{V(G)} +\log \frac{1}{\varepsilon} \cdot \varepsilon^{-\poly(\Delta, 1/B)}\right) = \abs{V(G)}\cdot\varepsilon^{-\poly(\Delta, 1/B)}.\]
    The theorem is proved.
\end{proof}

Now we can prove \Cref{thm:Holant-marginal-ratio-estimator}.

\begin{proof}[Proof of~\Cref{thm:Holant-marginal-ratio-estimator}]
    Given any instance $\Phi = (G = (V, E), \vecf)$ satisfying~\Cref{cond:Holant-condition}, for each $e = \set{u, v} \in E$, if $f_u(1) = 0$ or $f_v(1) = 0$,  by \eqref{def-holant-mu}, we have $\Pr[X \sim \mu_{\Phi}]{X(e) = 1} = 0$.
    Combining with \eqref{def-rphie}, we have $R_\Phi(e) = 0$. Thus, if $f_u(1) = 0$ or $f_v(1) = 0$, the algorithm $\+A$ can simply output $\wh{R} = 0$, satisfying both the error bound and the time complexity stated in the theorem.
    
    In the following, we assume $f_u(1) > 0$ and $f_v(1) > 0$.     
    Let $e_u = \set{u}$ and $e_v = \set{v}$ be two half-edges obtained by splitting the edge $e$. See \Cref{fig:splitting-edge} for an example.
    \input{split-edge}
    Define three instances as follows:
    \begin{itemize}
    \item Let $G_0 = \left(V, E_0 = E \setminus \set{e} \cup \set{e_u, e_v}\right)$. Let $\Phi_0 = \left(G_0, \vecf\right)$.
    \item Let $\Phi_1 = \Phi_0^{e_u \gets 1}$. Denote its underlying graph by $G_1 = (V, E_1 = E \setminus \set{e} \cup \set{e_v})$ and its signatures by $\vecf' = \set{f_w'}_{w \in V}$.
    \item Let $\Phi_2 = \Phi_0^{e_v \gets 0}$. Denote its underlying graph by $G_2 = (V, E_2 = E \setminus \set{e} \cup \set{e_u})$.
    \end{itemize}
    Moreover, we have the following properties:
    \begin{itemize}
    \item All of $\Delta\left(G_0\right), \Delta\left(G_1\right)$ and $\Delta\left(G_2\right)$ are no more than $\Delta(G)$;
    \item By $f_u(1)>0$, we have $f_u'(0) > 0$. Combined with  $\Phi = (G,\vecf)$ satisfies~\Cref{cond:Holant-condition}, one can verify that $\left(\Phi_1, (e_v \gets 1),(e_v \gets 0),v\right)$ satisfies~\Cref{cond-instancepair}.
    Similarly, we have $\left(\Phi_2, (e_u\leftarrow 1),(e_u\leftarrow 0),u\right)$ satisfies~\Cref{cond-instancepair};
    \item Recall the definition of $r_{\max}$ in~\eqref{eq-def-rmax} and $B$ in~\eqref{eq-def-bmin}. For the instance $\Phi_0$, we have $r_{\max}(\Phi_0) = r_{\max}(\Phi)$ and $B(\Phi_0) = B(\Phi)$. By~\Cref{lem:monotonicity-under-pinning}, we have
    \begin{equation} \label{eq:monotonicity}
    \begin{gathered}
        r_{\max}(\Phi_1) \le r_{\max}(\Phi_0) = r_{\max}(\Phi), \quad B(\Phi_1) \ge B(\Phi_0) = B(\Phi), \\
        r_{\max}(\Phi_2) \le r_{\max}(\Phi_0) = r_{\max}(\Phi), \quad B(\Phi_2) \ge B(\Phi_0) = B(\Phi).
    \end{gathered}
    \end{equation}
    \end{itemize}

    Set $\varepsilon_1 = \varepsilon_2 = \varepsilon/3$.
    Combined these properties with \Cref{lem:LP-ratio-estimator}, we have that $\wh{R}_1$ can be obtained within time 
    \[
        \abs{V(G_1)} \cdot \varepsilon_1^{-\poly\left(\Delta(G_1),1/B\left(\Phi_1\right)\right)} = \abs{V} \cdot \varepsilon^{-\RunningTimeExponent}
    \]
    where 
    \begin{align}\label{eq-rhat1}
        (1 - \varepsilon_1)R_{\Phi_1}(e_v) \leq \wh{R}_1 \leq (1 + \varepsilon_1)R_{\Phi_1}(e_v), 
    \end{align}
    Similarly, $\wh{R}_2$ can be obtained within time 
    \[
        \abs{V(G_2)} \cdot \varepsilon_1^{-\poly\left(\Delta(G_2),1/B\left(\Phi_2\right)\right)} = \abs{V} \cdot \varepsilon^{-\RunningTimeExponent}
    \]
    where 
    \begin{align}\label{eq-rhat2}
        (1 - \varepsilon_2)R_{\Phi_2}(e_u) \leq \wh{R}_2 \leq (1 + \varepsilon_2)R_{\Phi_2}(e_u).
    \end{align}
    Let $\wh{R} = \wh{R}_1 \cdot \wh{R}_2$. 
    Note that
    \begin{align}\label{eq-rphi-rphi12}
        R_{\Phi}(e) &= \frac{\Pr[X \sim \mu_{\Phi}]{X(e) = 1}}{\Pr[X \sim \mu_{\Phi}]{X(e) = 0}}\notag  \\
        &= \frac{\Pr[X \sim \mu_{\Phi_0}]{X(e_u) = X(e_v) = 1 \mid X(e_u) = X(e_v)}}{\Pr[X \sim \mu_{\Phi_0}]{X(e_u) = X(e_v) = 0 \mid X(e_u) = X(e_v)}} \notag\\
        &= \frac{\Pr[X \sim \mu_{\Phi_0}]{X(e_u) = X(e_v) = 1}}{\Pr[X \sim \mu_{\Phi_0}]{X(e_u) = X(e_v) = 0}} \notag \\
        &= \frac{\Pr[X \sim \mu_{\Phi_0}]{X(e_u) = 1 \land X(e_v) = 1}}{\Pr[X \sim \mu_{\Phi_0}]{X(e_u) = 1 \land X(e_v) = 0}} \cdot \frac{\Pr[X \sim \mu_{\Phi_0}]{X(e_u) = 1 \land X(e_v) = 0}}{\Pr[X \sim \mu_{\Phi_0}]{X(e_u) = 0 \land X(e_v) = 0}} \\
        &= \frac{\Pr[X \sim \mu_{\Phi_0}]{X(e_v) = 1 \mid X(e_u) = 1}}{\Pr[X \sim \mu_{\Phi_0}]{X(e_v) = 0 \mid X(e_u) = 1}} \cdot \frac{\Pr[X \sim \mu_{\Phi_0}]{X(e_u) = 1 \mid X(e_v) = 0}}{\Pr[X \sim \mu_{\Phi_0}]{X(e_u) = 0 \mid X(e_v) = 0}}\notag \\
        &= \frac{\Pr[X \sim \mu_{\Phi_1}]{X(e_v) = 1}}{\Pr[X \sim \mu_{\Phi_1}]{X(e_v) = 0}} \cdot \frac{\Pr[X \sim \mu_{\Phi_2}]{X(e_u) = 1}}{\Pr[X \sim \mu_{\Phi_2}]{X(e_u) = 0}}\notag \\
        &= R_{\Phi_1}(e_v) \cdot R_{\Phi_2}(e_u).\notag
    \end{align}
    We emphasize that all the denominators in above inequalities are not $0$ because all of $f_u(0),f_v(0), f_u(1),f_v(1)$ are larger than $0$.
    Thus, we have
    \begin{align*}
        \wh{R} &= \wh{R}_1 \cdot \wh{R}_2\\
        \left(\text{by \eqref{eq-rhat1} and \eqref{eq-rhat1}}\right)\quad&\ge (1 - \varepsilon_1)\cdot(1 - \varepsilon_2)\cdot R_{\Phi^{(1)}}(e_v) \cdot R_{\Phi^{(2)}}(e_u) \\
        \left(\text{by $\varepsilon_1 = \varepsilon_2 =\varepsilon/3$}\right) \quad &=         (1 - \varepsilon/3)^2 \cdot R_{\Phi^{(1)}}(e_v) \cdot R_{\Phi^{(2)}}(e_u) \\
       \left(\text{by $\epsilon\in (0,1/4)$ and \eqref{eq-rphi-rphi12}}\right)\quad & \ge (1 - \varepsilon) R_{\Phi}(e).
    \end{align*}
    Similarly, we also have 
    $\wh{R} \le (1 + \varepsilon) R_{\Phi}(e)$. 
    In summary, we have $(1 - \varepsilon) R_{\Phi}(e) \le \wh{R} \le (1 + \varepsilon) R_{\Phi}(e)$. 
    
    Now we consider the time cost for calculating $\wh{R}$.
    One can verify that the time cost for constructing $\Phi^{(1)}$ and $\Phi^{(2)}$ is $O(\abs{V}\cdot \Delta(G))$.
    Recall that the time cost for calculating $\wh{R}_1$ and $\wh{R}_2$ is 
    $O\left(\abs{V} \cdot \varepsilon^{-\RunningTimeExponent}\right)$.
    In summary, the total time cost for calculating $\wh{R}$ is 
    $O\left(\abs{V} \cdot \varepsilon^{-\RunningTimeExponent}\right)$. Hence the theorem is proved.
\end{proof}

\section{Approximate Counting}\label{sec:counting}
In this section, we present our deterministic algorithm for approximating the partition functions of instances satisfying~\Cref{cond:Holant-condition}.
This algorithm is based on the marginal ratio estimator in~\Cref{thm:Holant-marginal-ratio-estimator}.

\begin{theorem} \label{thm:formal-counting-Holant}
    There is a deterministic algorithm such that given as input any $\varepsilon>0$ and any instance $\Phi = \left(G = (V, E), \vecf = \set{f_v}_{v \in V}\right)$ satisfying~\Cref{cond:Holant-condition}, it outputs a number $\wh{Z}$ such that
    $$
        (1 - \varepsilon) Z_\Phi \le \wh{Z} \le (1 + \varepsilon) Z_{\Phi}
    $$
    within time $O\left(\abs{V} \cdot (\abs{E} \cdot \varepsilon^{-1})^{\RunningTimeExponent}\right)$.
\end{theorem}

\begin{proof}
    Without loss of generality, assume that $\varepsilon < 1/4$, $\abs{E} = m$ and $E = \{e_1, \ldots, e_m\}$. Let $\Phi_1 = \Phi$ and for each $2 \le i \le m + 1$, define $\Phi_i = \Phi_{i - 1}^{e_i \gets 0}$. Note that the underlying graph in $\Phi_{m + 1}$ contains no edges and thus by the definition of the partition function $Z_{\Phi_{m + 1}}$,
    \begin{align} \label{eq:empty-partition-function}
        Z_{\Phi_{m + 1}} = \prod_{v \in V} f_v(0).
    \end{align}
    By the assumption $\Phi$ satisfies~\Cref{cond:Holant-condition}, we also have 
    $\Phi_i$ satisfies~\Cref{cond:Holant-condition} for each $i \in [m+1]$. Since $\Phi_{i}$ is a pinning of $\Phi_{i - 1}$ for $2 \le i \le m + 1$, by~\Cref{lem:monotonicity-under-pinning}, $B(\Phi) = B(\Phi_1) \le B(\Phi_2) \le \ldots \le B(\Phi_{m + 1})$.
    Thus, we have 
    \begin{equation}
    \begin{aligned}\label{eq-z-decompose}
        Z_{\Phi} &= Z_{\Phi}^{e_1 \gets 0} + Z_{\Phi}^{e_1 \gets 1} \\
        \left(\text{by \eqref{def-rphie}}\right) \quad
        &= \left(1 + R_{\Phi}(e_1)\right) Z_{\Phi}^{e_1 \gets 0} \\
        \left(\text{by definitions of $\Phi_1,\Phi_2$}\right) \quad
        &= \left(1 + R_{\Phi_1}(e_1)\right) Z_{\Phi_2} \\
        &= \left(1 + R_{\Phi_1}(e_1)\right) \left(Z_{\Phi_2}^{e_2 \gets 0} + Z_{\Phi_2}^{e_2 \gets 1}\right) \\
        \left(\text{by \eqref{def-rphie}}\right) \quad
        &= \left(1 + R_{\Phi_1}(e_1)\right) \left(1 + R_{\Phi_2}(e_2)\right) Z_{\Phi_2}^{e_2 \gets 0} \\
        \left(\text{by definitions of $\Phi_2,\Phi_3$}\right) \quad
        &= \left(1 + R_{\Phi_1}(e_1)\right) \left(1 + R_{\Phi_2}(e_2)\right) Z_{\Phi_3} \\
        &= \cdots \\
        &= Z_{\Phi_{m + 1}}\prod_{i = 1}^{m} \left(1 + R_{\Phi_i}(e_i)\right) \\
        \left(\text{by~\eqref{eq:empty-partition-function}}\right) \quad
        &= \prod_{v \in V} f_v(0) \prod_{i = 1}^{m} \left(1 + R_{\Phi_i}(e_i)\right).
    \end{aligned} 
    \end{equation}
    For each $1 \le i \le m$, we run the algorithm $\+A$ in~\Cref{thm:Holant-marginal-ratio-estimator} with tolerance error $\varepsilon/(2m)$ on the input $\Phi_i, e_i$ to obtain an $\wh{R}_i$ where 
    \begin{align}\label{eq-condition-ri}
        (1-\varepsilon/(2m))R_{\Phi_i}(e_i) \leq \wh{R}_i \leq (1+\varepsilon/(2m))R_{\Phi_i}(e_i).
    \end{align}
    Define 
    \begin{align}\label{eq-def-z}
        \wh{Z} \triangleq \prod_{v \in V} f_v(0) \prod_{i = 1}^{m} (1 + \wh{R}_i).
    \end{align}
    Thus, we have
    \begin{align*}
        &\symbolwidth \frac{\wh{Z}}{Z_\Phi} \\
        \left(\text{by \eqref{eq-z-decompose} and \eqref{eq-def-z}}\right)\quad
        &= \prod_{i = 1}^{m} \left( \frac{1 + \wh{R}_i}{1 + R_{\Phi_i}(e_i)}\right)\\ 
        \left(\text{by \eqref{eq-condition-ri}}\right) \quad
        &\ge \prod_{i = 1}^{m}  \left( \frac{1+(1-\varepsilon/(2m))R_{\Phi_i}(e_i) }{1 + R_{\Phi_i}(e_i)}\right)\\
        &\ge \prod_{i = 1}^{m} \left(1- \frac{\varepsilon\cdot R_{\Phi_i}(e_i)}{2m(1 + R_{\Phi_i}(e_i))}\right)\\ 
        &\ge \prod_{i = 1}^{m} \left(1- \frac{\varepsilon}{2m}\right)\\
        \left(\text{by $\varepsilon\in (0,1/4)$}\right)\quad
        &\ge 1 - \varepsilon.
    \end{align*}
    Similarly, we also have 
    \begin{align*}
        &\symbolwidth \frac{\wh{Z}}{Z_\Phi} \\
        \left(\text{by \eqref{eq-z-decompose} and \eqref{eq-def-z}}\right)\quad 
        &= \prod_{i = 1}^{m} \left( \frac{1 + \wh{R}_i}{1 + R_{\Phi_i}(e_i)}\right)\\
        \left(\text{by \eqref{eq-condition-ri}}\right) \quad
        &\leq \prod_{i = 1}^{m} \left( \frac{1+(1+\varepsilon/(2m))R_{\Phi_i}(e_i) }{1 + R_{\Phi_i}(e_i)}\right) \\
        &\leq \prod_{i = 1}^{m} \left(1+ \frac{\varepsilon\cdot R_{\Phi_i}(e_i)}{2m(1 + R_{\Phi_i}(e_i))}\right)\\ 
        &\leq \prod_{i = 1}^{m} \left(1+ \frac{\varepsilon}{2m}\right) \\
        \left(\text{by $\varepsilon\in (0,1/4)$}\right) \quad
        &\le 1 + \varepsilon.
    \end{align*}

    Now we turn to the running time of the algorithm. Firstly we can obtain $\prod_{v \in V} f_v(0)$ in time $O(\abs{V})$. For every $i \in [m]$, by~\Cref{thm:Holant-marginal-ratio-estimator},
    we can obtain $\wh{R}_i$ in time
    $$
        O\left(\abs{V(G_i)} (m/\varepsilon)^{\poly(\Delta(G_i), 1/B(\Phi_i))}\right) = O\left(\abs{V} \cdot (m/\varepsilon)^{\RunningTimeExponent}\right),
    $$
    where $G_i$ denotes the underlying graph of $\Phi_i$ and we use the fact $B(\Phi) \leq B(\Phi_i)$ due to \Cref{lem:monotonicity-under-pinning}.
    Hence we can calculate $\wh{Z}$ in time
    $$
        O\left(\abs{V} + \abs{V} \cdot m \cdot (m/\varepsilon)^{\RunningTimeExponent}\right) = O\left(\abs{V} \cdot (m \cdot \varepsilon^{-1})^{\RunningTimeExponent}\right).
    $$
    The theorem is proved.
\end{proof}

The main result for counting $\vecb$-matchings is an immediate corollary of~\Cref{thm:formal-counting-Holant}.

\begin{theorem} \label{thm:formal-counting-b-matchings}
    There exists a deterministic algorithm such that given any graph $G = (V, E)$ with maximum degree $\Delta$, any positive integer $b$, any vector $\vecb = \set{b_v}_{v \in V}$ satisfying $1 \le b_v \le b$ for every $v \in V$ and any $\varepsilon \in (0, 1)$ as input, it outputs a number $\wh{Z}$ such that
    \begin{align*}
        (1 - \varepsilon)Z_{G, \vecb} \le \wh{Z} \le (1 + \varepsilon) Z_{G, \vecb}
    \end{align*}
    within time $O\left(\abs{V} \cdot (\abs{E} \cdot \varepsilon^{-1})^{\poly(\Delta^b)}\right)$ where $Z_{G, \vecb}$ is the number of $\vecb$-matchings on $G$.
\end{theorem}

\begin{proof}
    Given any graph $G = (V, E)$, any positive integer $b$ and any vector $\vecb = \set{b_v}_{v \in V}$ satisfying $1 \le b_v \le b$,
    for each $v \in V$, define a signature $f_v \triangleq [f_v(0),f_v(1),\cdots,f_v(\deg_G(v))]$ where $f_v(i) = \id{i \le b_v}$ for every $0 \le i \le \deg_G(v)$.
    Consider the Holant instance $\Phi_{G, \vecb} \triangleq (G, \vecf)$.
    One can verify that $\Phi_{G, \vecb}$ satisfies~\Cref{cond:Holant-condition} and $Z_{\Phi_{G, \vecb}} = Z_{G, \vecb}$. Moreover, by~\eqref{eq-def-bmin} and~\eqref{eq-def-rmax}, it holds that $r_{\max}(\Phi_{G, \vecb}) = 1$ and $B(\Phi_{G, \vecb}) \ge 1/\left(\sum_{i = 0}^{b} \binom{\Delta}{i}\right) = \Delta^{-\poly(b)}$. 
    Thus, by~\Cref{thm:formal-counting-Holant}, there is a deterministic algorithm such that given as input $\varepsilon$ and $\Phi_{G, \vecb}$, it outputs a number $\wh{Z}$ where $(1 - \varepsilon) Z_{G, \vecb} \le \wh{Z} \le (1 + \varepsilon) Z_{G, \vecb}$ 
    within time
    $$
        O\left(\abs{V} \cdot (\abs{E} \cdot \varepsilon^{-1})^{\poly(\Delta(G), 1/B(\Phi_{G, \vecb}))}\right) = O\left(\abs{V} \cdot (\abs{E} \cdot \varepsilon^{-1})^{\poly(\Delta^b)}\right).
    $$
    The theorem is proved.
\end{proof}

\bibliographystyle{alpha}
\bibliography{refs}

\appendix

\section{Counterexample} \label{sec:counter-example}

In this section, we give a counterexample to show that for general Holant instances, an arbitrary unpinned edge in $E_v^\sigma$ is \emph{not necessarily} amendable in the coupling process $\Couple(\Phi, \sigma, \tau, v)$.

Consider the Holant instance $\Phi = \left(G, \vecf\right)$ where the graph $G = (V = \set{v_\bot, v_1, v_2, \ldots, v_5}, E = \set{e_1, \ldots, e_5} \cup \set{e_\bot})$ is shown as~\Cref{fig:edge-selection-counterexample} and the signatures $\vecf = \set{f_v}_{v \in V}$ are defined as $f_{v_\bot} = [1, 1, 1, 0]$, $f_{v_2} = [1, 1]$ and $f_{v_1} = f_{v_3} = f_{v_4} = f_{v_5} = [1, 10, 0]$.


\input{counterexample}

Define two partial assignments $\sigma, \tau$ on $\set{e_\bot}$ as $\sigma = (e_\bot \gets 1)$ and $\tau = (e_\bot \gets 0)$. One can easily verify that the tuple $(\Phi, \sigma, \tau, v_\bot)$ satisfies~\Cref{cond-instancepair} and for $e_2 = \set{v_\bot, e_2}$, it holds that $\Ham(\sigma, E_{v_\bot}) > \Ham(\tau, E_{v_\bot})$. However, by calculation,
$$
    \mu_{e_2}^{\sigma}(1) = \frac{10301}{24622} > \frac{14321}{38742} = \mu_{e_2}^{\tau}(1)
$$
meaning that $e_2$ is \emph{not} an amendable edge. This phenomenon indicates that the choice of the edge in Line~\ref{line:pick-dominating-edge-1} and Line~\ref{line:pick-dominating-edge-2} through the coupling process \emph{cannot} be arbitrary.


\section{Proofs for Properties of Holant Instances} \label{sec:proof-Holant-properties}
In this section, we complete the deferred proofs for properties of Holant instances. Recall that $\Phi = (G = (V, E = E_1 \cup E_2), \vecf = \set{f_v})$ is a Boolean domain symmetric log-concave Holant instance satisfying~\Cref{cond:Holant-condition} and $\mu = \mu_\Phi$ is its Gibbs distribution.

\subsection{Monotonicity under pinnings}

We prove the monotonicity of quantities $r_{\max}$ and $B$ under the pinnings.
\Monotonicity*
\begin{proof}
    First of all, since $Z_{\Phi}^{e \gets c} > 0$, $f_v^{e \gets c}(0) > 0$ for every $v \in V$ and thus quantities $r_{\max}(\Phi^{e \gets c})$ and $B(\Phi^{e \gets c})$ are well-defined. By the assumption $\Phi$ satisfies~\Cref{cond:Holant-condition}, it holds that for every $v \in V$, $f_v(0) > 0$ and $f_v$ is log-concave. For convenience, assume that $f_v(k) = 0$ for $k > \deg_G(v)$. By log-concavity, it holds that
    \begin{align} \label{eq:log-concave-inequality}
        \frac{f_v(i)}{f_v(0)} \ge \frac{f_v(\ell + i)}{f_v(\ell)}, \quad \forall i, \ell \in \mathbb{N}
    \end{align}
    with convention $0/0 = 0$. Recall the definition of $r_{\max}$ in~\eqref{eq-def-rmax}. It holds that
    \begin{align*}
        r_{\max}\left(\Phi^{e \gets c}\right) &= \max_{v \in V} \frac{f_v^{e \gets c}(1)}{f_v^{e \gets c}(0)} \\
        &= \max\set{\max_{v \in e} \frac{f_v^{e \gets c}(1)}{f_v^{e \gets c}(0)}, \max_{v \in V \setminus e}\frac{f_v^{e \gets c}(1)}{f_v^{e \gets c}(0)}} \\
        (\mbox{by definition of $\vecf^{e \gets c}$}) \quad
        &= \max\set{\max_{v \in e} \frac{f_v(1 + c)}{f_v(c)}, \max_{v \in V \setminus e}\frac{f_v(1)}{f_v(0)}} \\
        (\mbox{by~\eqref{eq:log-concave-inequality}}) \quad
        &\le \max\set{\max_{v \in e} \frac{f_v(1)}{f_v(0)}, \max_{v \in V \setminus e}\frac{f_v(1)}{f_v(0)}} = r_{\max}(\Phi).
    \end{align*}

    To show the inequality for $B$, by~\eqref{eq-def-bmin}, it suffices to show that for every $v \in V$, we have
    \begin{align} \label{eq:local-ratio-inequality}
        \frac{P_{f_v^{e \gets c}}(0)}{P_{f_v^{e \gets c}}(r_{\max}(\Phi^{e \gets c}))} \ge \frac{P_{f_v}(0)}{P_{f_v}(r_{\max}(\Phi))}.
    \end{align}
    By~\eqref{eq:local-polynomial}, $P_{f_v^{e \gets c}}(x)$ is increasing on $x \ge 0$. Recall that $0 \le r_{\max}(\Phi^{e \gets c}) \le r_{\max}(\Phi)$. To prove~\eqref{eq:local-ratio-inequality}, we only need to show
    \begin{align} \label{eq:intermediate-local-ratio-inequality}
        \frac{P_{f_v^{e \gets c}}(0)}{P_{f_v^{e \gets c}}(r_{\max}(\Phi))} \ge \frac{P_{f_v}(0)}{P_{f_v}(r_{\max}(\Phi))}.
    \end{align}
    For every $v \in V \setminus \set{u}$, by the definition of $\vecf^{e \gets c}$, we have $f_v^{e \gets c} = f_v$ and~\eqref{eq:intermediate-local-ratio-inequality} holds trivially. For $v \in e$, let $d = \deg_G(v)$ and thus $\deg_{G^{e \gets c}}(v) = d - 1$. Since $Z_{\Phi}^{e \gets c} > 0$, it holds that $f_v(c) = f_v^{e \gets c}(0) > 0$. By~\eqref{eq:local-polynomial}, we have
    \begin{align*}
        &\symbolwidth \frac{P_{f_v^{e \gets c}}(0)}{P_{f_v^{e \gets c}}(r_{\max}(\Phi))} \\
        (\mbox{by~\eqref{eq:local-polynomial}}) \quad
        &= \frac{\sum_{i = 0}^{d - 1} \binom{d - 1}{i} f_v^{e \gets c}(i) 0^i}{\sum_{i = 0}^{d - 1} \binom{d - 1}{i} f_v^{e \gets c}(i) r_{\max}(\Phi)^i} \\
        &= \frac{f_v^{e \gets c}(0)}{\sum_{i = 0}^{d - 1} \binom{d - 1}{i} f_v^{e \gets c}(i) r_{\max}(\Phi)^i} \\
        \left(\mbox{by definition of $f_v^{e \gets c}$}\right) \quad
        &= \frac{f_v(c)}{\sum_{i = 0}^{d - 1} \binom{d - 1}{i} f_v(i + c) r_{\max}(\Phi)^i} \\
        &= \frac{1}{\sum_{i = 0}^{d - 1} \binom{d - 1}{i} (f_v(i + c) / f_v(c)) r_{\max}(\Phi)^i} \\
        \left(\mbox{by $\binom{d}{i} \ge \binom{d - 1}{i}$ and~\eqref{eq:log-concave-inequality}}\right) \quad
        &\ge \frac{1}{\sum_{i = 0}^{d - 1} \binom{d}{i} (f_v(i) / f_v(0)) r_{\max}(\Phi)^i} \\
        \left(\mbox{by $f_v(d) \ge 0$}\right) \quad
        &\ge \frac{f_v(0)}{\sum_{i = 0}^{d} \binom{d}{i} f_v(i) r_{\max}(\Phi)^i} = \frac{P_{f_v}(0)}{P_{f_v}(r_{\max}(\Phi))}.
    \end{align*}

    Combining all of the above arguments, we conclude the lemma.
\end{proof}

\subsection{Feasibility of partial assignments}

We recall the theorem for the feasibility of partial assignments.
\PartialFeasibility*

To show the feasibility of a partial assignment $\sigma$, we first make the following claim.
\begin{claim} \label{claim:partial-assignment-feasibility}
    A partial assignment $\sigma$ is feasible if and only if the assignment $\sigma'$ on $E$ is feasible where $\sigma'$ is defined as
    \begin{align} \label{eq:0-assignment}
        \sigma'(e) = \begin{cases}
            \sigma(e) & \sigma \in \Lambda(\sigma) \\
            0 & \mbox{otherwise}
        \end{cases}\;.
    \end{align}
\end{claim}
\begin{proof}
    When $\sigma'$ is feasible, obviously $\sigma$ is feasible since $\sigma' \in \sigma$. To see the other side, when $\sigma$ is feasible, there exists $\tau : E \to \set{0, 1}$, $\tau \in \sigma$ such that $\mu(\tau) > 0$. Hence we have $\prod_{v \in V} f_v\left(\abs{\tau(E_v)}\right) > 0$ meaning that $f_v\left(\abs{\tau(E_v)}\right) > 0$ for every $v \in V$. Since $f_v$ is log-concave, $f_v(0) > 0$ and $\abs{\sigma'(E_v)} \le \abs{\tau(E_v)}$, it holds that $f_v\left(\abs{\sigma'(E_v)}\right) > 0$. Therefore we obtain that $\mu(\sigma') > 0$ by the definition of $\mu(\sigma')$.
\end{proof}

\begin{proof}[Proof of~\Cref{lem:partial-assignment-feasibility}]
    By~\Cref{claim:partial-assignment-feasibility}, we check the feasibility of $\sigma'$ defined as~\eqref{eq:0-assignment}. In other words, we check whether $f\left(\abs{\sigma'(E_v)}\right) > 0$ for every $v \in V$. By definition, for every $v \in V$, $\abs{\sigma'(E_v)} = \Ham(\sigma, v)$ and thus $f\left(\abs{\sigma'(E_v)}\right) = f\left(\Ham(\sigma, v)\right)$. When $\Lambda(\sigma) \cap E_v = \emptyset$, we need to do nothing from the assumption $f_v(0) > 0$. Let $V(\sigma) \defeq \set{v \in V \cmid \Lambda(\sigma) \cap E_v \neq \emptyset}$. It is not hard to see that we can find $V(\sigma)$ and compute $\Ham(\sigma, v)$ for every $v \in V(\sigma)$ in time $O(\abs{\Lambda(\sigma)})$ by enumerating all edges in $\Lambda(\sigma)$. Hence we can check the feasibility of a partial assignment $\sigma$ in time $O(\abs{\Lambda(\sigma)})$. 
\end{proof}

\subsection{Marginal ratio bounds}

In this part, we prove the upper bound for marginal ratios.

\MarginalRatioBound*
\begin{proof}
    Let $v$ be the unique vertex incident to $e$. Since $f_v$ is log-concave and $f_v(0) > 0$, it holds that $f_v(i + 1) \le r_{\max}(\Phi) f_v(i)$. Hence we have, for every $\sigma \in \set{0, 1}^E$ with $\sigma(e) = 1$, $\prod_{u \in V} f_u\left(\abs{\sigma(E_u)}\right) \le r_{\max}(\Phi) \prod_{u \in V} f_u\left(\abs{\sigma'(E_u)}\right)$ where $\sigma'(e') = \sigma(e')$ for every $e' \in E \setminus \set{e}$ and $\sigma'(e) = 0$. Then by direct calculation,
    \begin{align*}
        R_{\Phi}(e) &= \frac{Z_{\Phi}^{e \gets 1}}{Z_{\Phi}^{e \gets 0}} \\
        &= \frac{\sum_{\sigma \in \set{0, 1}^E : \sigma(e) = 1} \prod_{u \in V} f_u\left(\abs{\sigma(E_u)}\right)}{\sum_{\sigma \in \set{0, 1}^E : \sigma(e) = 0} \prod_{u \in V} f_u\left(\abs{\sigma(E_u)}\right)} \\
        &\le \frac{r_{\max}(\Phi) \sum_{\sigma \in \set{0, 1}^E : \sigma(e) = 0} \prod_{u \in V} f_u\left(\abs{\sigma(E_u)}\right)}{\sum_{\sigma \in \set{0, 1}^E : \sigma(e) = 0} \prod_{u \in V} f_u\left(\abs{\sigma(E_u)}\right)} \le r_{\max}(\Phi).
    \end{align*}
    Then we conclude the upper bound.
\end{proof}

\section{Omitted Proofs}\label{sec:omitted Proofs}

\subsection{Omitted proofs in the coupling process and tree} \label{sec:omitted-proofs-coupling}

\PropertyDefrpc*
\begin{proof}
To prove this lemma, it is sufficient to prove that for each $(\sigma, \tau, \seqS,v,t)$,
\begin{align}\label{eq-def-trp-perperty-sigma-delete-vt}
\Pr[\!{cp}]{(\sigma, \tau, \seqS,v,t)}\leq \mu^{\sigma_{\bot}}_{\seqS}(\sigma).
\end{align}
Similarly, one can also prove that $\Pr[\!{cp}]{(\sigma, \tau, \seqS,v,L)}\leq \mu^{\tau_{\bot}}_{\seqS}(\tau)$.
Thus, \eqref{eq-def-trp-perperty-sigma} is proved.
Meanwhile, for each $e\in E^{\sigma}_v$,
by \Cref{def-notation-trp} and \eqref{eq-def-trp-perperty-sigma-delete-vt}, 
we have
\begin{align*}
\Pr[\!{cp}]{(\sigma, \tau, \seqS,v,L,e)} \leq  \Pr[\!{cp}]{(\sigma, \tau, \seqS,v,L)} \leq \mu^{\sigma_{\bot}}_{\seqS}(\sigma).
\end{align*}
Similarly, one can also prove that $\Pr[\!{cp}]{(\sigma, \tau, \seqS,v,L,e)}\leq \mu^{\tau_{\bot}}_{\seqS}(\tau)$.
Thus, \eqref{eq-def-trp-perperty-sigma-e} is also proved and the lemma is immediate.

In the following, we prove \eqref{eq-def-trp-perperty-sigma-delete-vt} by induction on the length of $\seqS$.
The induction basis is when $\seqS = \varnothing$.
In this case, by \eqref{eq-def-pro-stsvl} we have 
\begin{align*}
 \Pr[\!{cp}]{(\sigma, \tau, \seqS,v,L)} = \Pr{\left(T\geq 0\right)\land \left((\sigma, \tau, \seqS, v, L) = (\sigma_\bot, \tau_\bot, \varnothing, v_\bot, 0)\right)}\leq 1 = \mu^{\sigma_{\bot}}_{\seqS}(\sigma).
\end{align*}
The base case is proved. For the induction step,  assume $\abs{\seqS} = t>0$, $\seqS = \seqS'\circ e$, $\sigma = \sigma'\land(e\leftarrow a)$, $\tau = \tau'\land(e\leftarrow b)$ for some $e\in E(G),a,b\in \{0,1\}$.
By \eqref{eq-def-pro-stsvl} we have 
\begin{align*}
 \Pr[\!{cp}]{(\sigma, \tau, \seqS, v, L)} 
 = \Pr{\left(T\geq t\right)\land \left((\sigma, \tau, \seqS, v, L) = (\sigma_t, \tau_t, \seqS_t, v_t, L_t)\right)}.
\end{align*}
In addition, by \Cref{def:truncated-random-process},
we have the event $\left(T\geq t\right)\land \left((\sigma, \tau, \seqS, v, L) = (\sigma_t, \tau_t, \seqS_t, v_t, L_t)\right)$ happens only if 
\begin{itemize}
\item the event $\+E\triangleq \left(T\geq t-1\right)\land \left((\sigma', \tau', \seqS', v(\sigma',\tau'), L(\sigma',\tau')) = (\sigma_{t-1}, \tau_{t-1}, \seqS_{t-1}, v_{t-1}, L_{t-1})\right)$ happens;
\item the chosen edge in the state $(\sigma_{t-1}, \tau_{t-1}, \seqS_{t-1}, v_{t-1}, L_{t-1})$ is $e$ and the sample $(\sigma_e,\tau_e)$ from an optimal coupling of $(\mu_e^{\sigma_{t-1}},\mu_e^{\tau_{t-1}})$ is $(a,b)$.
\end{itemize}
Thus, we have
\begin{align*}
 &\symbolwidth \Pr[\!{cp}]{(\sigma, \tau, \seqS, v, L)}\\ 
 &= \Pr{\left(T\geq t\right)\land \left((\sigma, \tau, \seqS, v, L) = (\sigma_t, \tau_t, \seqS_t, v_t, L_t)\right)}\\
&\leq 
\Pr{\+E}\cdot\Pr{(\sigma_e,\tau_e) =(a,b)\mid \+E}
\\&\leq \Pr{\left(T\geq t-1\right)\land \left((\sigma', \tau', \seqS', v(\sigma',\tau'), L(\sigma',\tau')) = (\sigma_{t-1}, \tau_{t-1}, \seqS_{t-1}, v_{t-1}, L_{t-1})\right)}\cdot\Pr{(\sigma_e,\tau_e) =(a,b)\mid \+E}
\\&\leq \mu^{\sigma_{\bot}}_{\seqS'}(\sigma')\cdot\Pr{(\sigma_e,\tau_e) =(a,b)\mid \+E},
\end{align*}
where the last inequality is by the induction hypothesis.
Moreover, we have 
\begin{equation*}
\begin{aligned}
\Pr{(\sigma_e,\tau_e) =(a,b)\mid \+E}
\leq \Pr{\sigma_e =a\mid \+E}
= \mu_e^{\sigma_{t-1}}\left(a \mid \+E\right)
= \mu_e^{\sigma'}(a),
\end{aligned}
\end{equation*}
where the first equality is by $(\sigma_e,\tau_e)$ is a coupling of $(\mu_e^{\sigma_{t-1}},\mu_e^{\tau_{t-1}})$ and the second equality is by $\sigma_{t-1} = \sigma'$ if $\+E$ happens.
Therefore, we have 
\begin{align*}
\Pr[\!{cp}]{(\sigma, \tau, \seqS, v, L)}\leq \mu^{\sigma_{\bot}}_{\seqS'}(\sigma')\cdot \mu_e^{\sigma'}(a) = \mu^{\sigma_{\bot}}_{\seqS}(\sigma).
\end{align*}
This completes the induction step.
Then \eqref{eq-def-trp-perperty-sigma-delete-vt} is proved and the lemma is immediate.
\end{proof}

\PropertyTruncateTree*

\begin{proof}
    \underline{Proof of (1).} By \Cref{def:truncated-coupling-tree},  each infeasible node is a leaf in $\+T$.
    Formally, $V(\+T) \setminus \+V \subseteq \+L$.
    Thus,  $V(\+T) \setminus \+L \subseteq \+V$.
    Therefore, $V(\+T) \setminus \+L \subseteq \+V\setminus \+L$.
    Combined with $\+V\subseteq V(\+T)$, we have $V(\+T) \setminus \+L = \+V\setminus \+L$.
    \vspace{0.2cm}
    
    \underline{Proof of (2).} 
    Let $(\sigma_0,\tau_0,\seqS_0,v_0,L_0),\cdots,(\sigma_t,\tau_t,\seqS_t,v_t,L_t)$ be a path from the root to any leaf in $\+T$, where $(\sigma_0,\tau_0,\seqS_0,v_0,L_0)$ is the root $(\sigma_\bot, \tau_\bot, \varnothing, v_\bot, 0)$ and $(\sigma_t,\tau_t,\seqS_t,v_t,L_t)$ is the leaf.
    At first, we prove the bound on the depth of $\+T$.
    According to ~\Cref{def:truncated-coupling-tree}, for each $0\leq i <t$, one can verify that either $L_{i+1}= L_{i} + 1$ or the following holds:
    \[L_{i+1} = L_{i}, v_{i+1} = v_{i}, \Lambda(\sigma_{i+1}) = \Lambda(\sigma_{i})\cup \{e_{i+1}\}, e_{i}\in \Lambda(\sigma_{i}) \text{ for some } e_{i}\in E_{v_i},e_{i+1}\in E^{\sigma_i}_{v_i}.\]
    Thus, for any $j>i$ where $L_j = L_i$, by induction
    one can verify that  
    $\Lambda(\sigma_{j}) = \Lambda(\sigma_{i})\cup S$ where $\abs{S} = j-i$ and $S\subseteq E^{\sigma_i}_{v_i}$. 
    In addition, by $E^{\sigma_i}_{v_i}\subseteq E_{v_i} \setminus \{e_i\}$, we have $\abs{E^{\sigma_i}_{v_i}}\leq \abs{E_{v_i}\setminus \{e_i\}} = \Delta - 1$.
    Thus, we have $j - i= \abs{S} \leq \Delta - 1$.
    Therefore, $j\leq i+\Delta - 1$.
    Thus, for each $0\leq i <t$ and each $i+\Delta \leq k\leq t$, we have $L_{k}>L_{i}$.
    Therefore, $L_{t}\geq 0 + \lfloor t/\Delta\rfloor$.
    In addition, by ~\Cref{def:truncated-coupling-tree} we have $L_t\leq \ell$.
    Thus, we have $t\leq \Delta\ell$.
    Therefore, the depth of $\+T$ is no more than $\Delta\ell$.

    Now we prove the conclusion $\abs{\Lambda(\sigma)} = \abs{\Lambda(\tau)} \leq \Delta \ell + 1$ for each node $(\sigma,\tau,\seqS,v,L)\in \+T$.
    By $\abs{\Lambda(\sigma_{0})} = 1$ and $\abs{\Lambda(\sigma_{i+1})} = \abs{\Lambda(\sigma_{i})} + 1$ for each 
    $0\leq i <t \leq \Delta\ell$, we have $\abs{\Lambda(\sigma_{i+1})} \leq \Delta\ell + 1$.
    Combined with \Cref{condition-sigma-tau},
    we also have $\abs{\Lambda(\tau_{i+1})} = \abs{\Lambda(\sigma_{i+1})} \leq \Delta\ell + 1$. 
    The conclusion is proved.

    In the next, we prove the bound on the degree of $\+T$.
    For each $0\leq i <t$, assume \emph{w.l.o.g.} ${\!{Ham}\left(\sigma_i,{E_{v_i}}\right)}<{\!{Ham}\left(\tau_i,{E_{v_i}}\right)}$.
    According to ~\Cref{def:truncated-coupling-tree},     
    one can verify that there exists some $e\in E^{\sigma_i}_{v_i}$ such that
    \[\seqS_{i+1} = \seqS_{i}\circ e, (\sigma_{i+1},\tau_{i+1})\in \{(\sigma_i\land(e\leftarrow 0),\tau_i\land(e\leftarrow 0)),(\sigma_i\land(e\leftarrow 1),\tau_i\land(e\leftarrow 0)),(\sigma_i\land(e\leftarrow 1),\tau_i\land(e\leftarrow 1))\}.\]
    Moreover, by \Cref{condition-sigma-tau} we have $v_{i+1} = v(\sigma_{i+1},\tau_{i+1})$ and
    $L_{i+1} = L(\sigma_{i+1},\tau_{i+1})$.
    Thus, by $\abs{E^{\sigma_i}_{v_i}}\leq \abs{E_{v_i}}\leq \Delta$, we have $(\sigma_{i+1},\tau_{i+1},\seqS_{i+1},v_{i+1},L_{i+1})$ has at most $3\Delta$ possibilities for each fixed  $(\sigma_{i},\tau_{i},\seqS_{i},v_{i},L_{i})$.
    Therefore, the degree of $\+T$ is no more than $3\Delta$.
    
    Finally, $\abs{V(\+T)}\leq \left(3\Delta\right)^{\Delta \ell + 1}$ is immediate by the depth of $\+T$ is at most $\Delta \ell$ and the degree is at most $3\Delta$.
 
    \vspace{0.2cm}
    \underline{Proof of (3).} For each node $(\sigma, \tau, \seqS,v,L) \in \+L_{\!{good}}\cap \+V$, 
    define a mapping $g:\left\{x\mid x\in \sigma\right\}\rightarrow \{0,1\}^E$ as follows.
    For every $x \in \sigma$, 
    \begin{align}\label{eq-def-g}
        (g(x))(e) \triangleq \begin{cases}
            \tau(e) & e \in \Lambda(\tau) \\
             x(e)& \mbox{otherwise}
        \end{cases}\;.
    \end{align}
    We claim that $g(\cdot)$ is a bijection between $\left\{x\mid x\in\sigma\right\}$ and $\left\{y\mid x\in\tau\right\}$.
    Because for each $x\in \sigma$, by \eqref{eq-def-g} we have $g(x)\in \tau$.
    In addition, by \Cref{condition-sigma-tau} we have $\Lambda(\sigma) = \Lambda(\tau)$. 
    Thus, one can also verify that for each $y\in \tau$, there exists a unique $x\in \sigma$ such that $g(x) = y$.  

    By \Cref{condition-sigma-tau}, we have $\abs{\sigma(E_u)} = \Ham(\sigma, E_u) = \Ham(\tau, E_u) = \abs{\tau(E_u)}$ for each $u \in V \setminus \set{v}$.
     Thus, we have
     \begin{equation*}
      \begin{aligned}\label{eq-xemu-eq-gxemu}
   & \abs{x(E_u)} \\
\left(\text{by $x\in \sigma$}\right) \quad   = &\abs{\sigma(E_u)} + \sum_{e\in E_u\setminus \Lambda(\sigma)}x(e) \\
\left(\text{by $\abs{\sigma(E_u)} = \abs{\tau(E_u)}, \Lambda(\sigma) = \Lambda(\tau)$}\right) \quad  = & \abs{\tau(E_u)} + \sum_{e\in E_u\setminus \Lambda(\tau)}x(e) \\
\left(\text{by \eqref{eq-def-g}}\right) \quad  = & \abs{(g(x))(E_u)}.
    \end{aligned}
    \end{equation*}
    Therefore, we have 
    \begin{align}\label{eq-mapping-g}
            \sum_{x\in \sigma} \prod_{u \in V\setminus \{v\}} f_u\left(\abs{x ({E_u})}\right)            =\sum_{x\in \sigma} \prod_{u \in V\setminus \{v\}} f_u\left(\abs{(g(x))({E_u})}\right)
            =\sum_{y\in \tau} \prod_{u \in V\setminus \{v\}} f_u\left(\abs{y({E_u})}\right),
        \end{align}
    where the last equality is by that $g(\cdot)$ is a bijection between $\left\{x\mid x\in \sigma\right\}$ and $\left\{y\mid y\in \tau\right\}$.
    
    Meanwhile, by $(\sigma, \tau, \seqS,v,L) \in \+L_{\!{good}}\cap \+V$, 
    we have $L<\ell$ and $(\sigma, \tau, \seqS,v,L)$ is a feasible leaf in $\+T$.
    Combing with \Cref{def:truncated-coupling-tree},
    we have $E_v^{\sigma} = \emptyset$. 
    Therefore, $x(E_v) = \sigma(E_v)$ for each $x\in \sigma$.
    Thus, by \eqref{def-holant-mu} we have
    \begin{equation}\label{eq-z-mu-sigma}
    \begin{aligned}
        Z\cdot \mu(\sigma)&= \sum_{x\in \sigma} \prod_{u \in V} f_u\left(\abs{x ({E_u})}\right)=\sum_{x\in \sigma} f_v\left(\abs{x ({E_v})}\right)\prod_{u \in V\setminus v} f_u\left(\abs{x ({E_u})}\right)\\
        &=f_v\left(\abs{\sigma ({E_v})}\right)\sum_{x\in \sigma} \prod_{u \in V\setminus \{v\}} f_u\left(\abs{x ({E_u})}\right).
    \end{aligned}
    \end{equation}
    Similarly, by $E_v^{\sigma} = \emptyset$ and $\Lambda(\sigma) = \Lambda(\tau)$, we have $E_v^{\tau} = \emptyset$. 
    Therefore, $y(E_v) = \tau(E_v)$ for each $y\in \tau$.
    Thus, by \eqref{def-holant-mu} we have 
    \begin{equation}\label{eq-z-mu-tau}
      \begin{aligned}
        Z\cdot\mu(\tau)&= \sum_{y\in \tau} \prod_{u \in V} f_u\left(\abs{y ({E_u})}\right)=\sum_{y\in \tau} f_v\left(\abs{y ({E_v})}\right)\prod_{u \in V\setminus v} f_u\left(\abs{y ({E_u})}\right)\\
        &=f_v\left(\abs{\tau ({E_v})}\right)\sum_{y\in \tau} \prod_{u \in V\setminus \{v\}} f_u\left(\abs{y ({E_u})}\right).
    \end{aligned}
    \end{equation}
    Moreover, recall that $(\sigma,\tau,\seqS,v,L)$ is feasible.
    We have $\mu(\sigma)>0$. Combined with \eqref{eq-mapping-g}, \eqref{eq-z-mu-sigma} and \eqref{eq-z-mu-tau}, we have
    $$
        \frac{\mu(\tau)}{\mu(\sigma)} =  \frac{f_v\left(\abs{\tau(E_v)}\right)}{f_v\left(\abs{\sigma({E_v})}\right)}.
    $$
\end{proof}

\LinearConstraints*
\begin{proof}
\underline{Proof of (1).} 
    We prove $p^{\sigma}_{\sigma, \tau, \seqS}\in [0,1]$ by considering two separate cases:
    \begin{itemize}
    \item $(\sigma, \tau, \seqS)\in \+V$. By \eqref{eqn-marginal-all} and \eqref{eq-def-trp-perperty-sigma}, we have 
    \begin{align*}
        p^{\sigma}_{\sigma, \tau, \seqS} = \frac{\Pr[\!{cp}]{(\sigma,\tau, \seqS)}}{\mu^{\sigma_{\bot}}_{\seqS}(\sigma)}\leq 1.
    \end{align*}
    Moreover, by $\Pr[\!{cp}]{(\sigma,\tau, \seqS)}\geq 0$ and $\mu^{\sigma_{\bot}}_{\seqS}(\sigma)>0$,
    we also have $p^{\sigma}_{\sigma, \tau, \seqS}\geq 0$.
    Therefore, we have $p^{\sigma}_{\sigma, \tau, \seqS} \in [0,1]$.  Similarly, we also have $p^{\tau}_{\sigma, \tau, \seqS} \in [0,1]$. 
    \item $(\sigma, \tau, \seqS)\in V(\+T)\setminus \+V$.  we have $p^{\sigma}_{\sigma, \tau, \seqS} =  p^{\tau}_{\sigma, \tau, \seqS} = 0$.
    \end{itemize}
    In summary, we always have $p^{\sigma}_{\sigma, \tau, \seqS} \in [0,1]$. Similarly, one can also prove that 
    $p^{\tau}_{\sigma, \tau, \seqS},p^{\sigma}_{\sigma, \tau, \seqS,e},p^{\tau}_{\sigma, \tau, \seqS,e}\in [0,1]$.
    Moreover, note that
    $\Pr[\!{cp}]{(\sigma_\bot, \tau_\bot, \varnothing)} = 1$.
    Combining with $\mu^{\sigma_{\bot}}_{\varnothing}(\sigma_\bot)$ = 1,
    we have 
    \[p^{\sigma_\bot}_{\sigma_\bot, \tau_\bot, \varnothing} = \frac{\Pr[\!{cp}]{(\sigma_\bot, \tau_\bot, \varnothing)}}{\mu^{\sigma_{\bot}}_{\varnothing}(\sigma_\bot)} = \frac{1}{1} = 1 .\]
    Similarly, we also have $p^{\tau_\bot}_{\sigma_\bot, \tau_\bot, \varnothing} = 1$.

    
    \underline{Proof of (2).} 
    It is sufficient to prove 
    \begin{align}\label{eqn-inter-sum1-first}
    p^{\sigma}_{\sigma,\tau,\seqS} = \sum_{e \in E_v^{\sigma}} p^{\sigma}_{\sigma, \tau, \seqS, e}.
    \end{align}
    Similarly, one can also prove 
    \[
    p^{\tau}_{\sigma,\tau,\seqS}=\sum_{e \in  E_v^{\sigma}} p^{\tau}_{\sigma,\tau, \seqS, e}.
    \]
    Then \eqref{eqn-inter-sum1} is immediate.
    In the following, we prove \eqref{eqn-inter-sum1-first}.
    For each $(\sigma, \tau, \seqS)$ in $\+V\setminus \+L$ and $e \in E_v^{\sigma}$ where $v = v(\sigma,\tau)$, we claim that 
    \begin{align}\label{eq-pr-stsvl-sum-stsvle}
    \Pr[\!{cp}]{(\sigma, \tau, \seqS)} = \sum_{e\in E_v^{\sigma}}\Pr[\!{cp}]{(\sigma, \tau, \seqS,e)}.
    \end{align}  
    Combining with \eqref{eqn-marginal-all} and \eqref{eqn-marginal-inner}, \eqref{eqn-inter-sum1-first} is immediate. 
    At last, we prove \eqref{eq-pr-stsvl-sum-stsvle}, which completes the proof of \eqref{eqn-inter-sum1-first}.
     We prove \eqref{eq-pr-stsvl-sum-stsvle} by considering two separate cases.
    \begin{itemize}
    \item $\Pr[\!{cp}]{(\sigma, \tau, \seqS)} = 0$.
    By \eqref{eq-def-pro-stsvl} and \eqref{eq-def-pro-stsvle} we have
    \[\forall e\in E_v^{\sigma},\quad \Pr[\!{cp}]{(\sigma, \tau, \seqS, e)}\leq \Pr[\!{cp}]{(\sigma, \tau, \seqS)} = 0 .\]
    Therefore,
    \begin{align*}
    \Pr[\!{cp}]{(\sigma, \tau, \seqS)} = 0 = \sum_{e\in E_v^{\sigma}}\Pr[\!{cp}]{(\sigma, \tau, \seqS,e)}.
    \end{align*}  
    Thus, \eqref{eq-pr-stsvl-sum-stsvle} is immediate.
    \item $\Pr[\!{cp}]{(\sigma, \tau, \seqS)} > 0$. Assume \emph{w.l.o.g.} that ${\!{Ham}\left(\sigma, {E_{v}}\right)} < {\!{Ham}\left(\tau, {E_{v}}\right)}$.
    By \Cref{def:truncated-random-process},
    if $e$ is the first edge in $E_{v}^{\sigma}$ with 
            $\mu^{\sigma}_e(1) \geq \mu^{\tau}_e(1)$,
    then 
    \begin{align*}
\Pr[\!{cp}]{(\sigma, \tau, \seqS,e)} &= \Pr{\left(T > t\right)\land \left((\sigma, \tau, \seqS) = (\sigma_t, \tau_t, \seqS_t)\right)\land \left(\seqS_{t+1} =\seqS\circ e\right) }\\
&=\Pr{\left(T > t\right)\land \left((\sigma, \tau, \seqS) = (\sigma_t, \tau_t, \seqS_t)\right)}\\
&=\Pr[\!{cp}]{(\sigma, \tau, \seqS)}.
\end{align*}
    Otherwise, $\Pr[\!{cp}]{(\sigma, \tau, \seqS,e)} = 0$.
    Thus, \eqref{eq-pr-stsvl-sum-stsvle} is immediate.
    \end{itemize}

    \underline{Proof of (3).} It is sufficient to prove \eqref{eqn-inner-child-sum1}. Then \eqref{eqn-inner-child-sum2}, \eqref{eqn-inner-child-sum3} and \eqref{eqn-inner-child-sum4} can be proved similarly.
   In the following, we prove \eqref{eqn-inner-child-sum1}.

   Assume 
   ${\!{Ham}\left(\sigma, {E_v}\right)} < {\!{Ham}\left(\tau,{E_v}\right)}$ and $\abs{\seqS} =t$.
   Let $\sigma^0 = \sigma \land (e\gets 0)$, 
   $\sigma^1 = \sigma \land (e\gets 1)$, $\tau^0 = \tau \land (e\gets 0)$, 
   $\tau^1 = \tau \land (e\gets 1)$, $\seqS'=\seqS\circ e$.
   By \Cref{def:truncated-random-process}, under the condition that $(\sigma_t,\tau_t,\seqS_t) = (\sigma,\tau,\seqS)$ and the chosen edge at this state is $e$, we have $\mu^{\sigma_t}_e(1) \geq \mu^{\tau_t}_e(1)$.
   Thus, for each $(\sigma_e,\tau_e)$ sampling from an optimal coupling of $(\mu_e^{\sigma_t},\mu_e^{\tau_t})$,
   we have 
   \[\Pr{\sigma_e = \tau_e = 0} = \mu^{\sigma_t}_e(0) = \mu^{\sigma}_e(0),\quad \Pr{(\sigma_e = \tau_e = 1)\lor(\sigma_e = 1,\tau_e = 0)} = \mu^{\sigma_t}_e(1) = \mu^{\sigma}_e(1).\]
   Formally,
   \[\Pr{(\sigma_{t+1},\tau_{t+1}) = (\sigma^0,  \tau^0)\mid \left(T > t\right)\land \left((\sigma, \tau, \seqS) = (\sigma_t, \tau_t, \seqS_t)\right)\land \left(\seqS_{t+1} =\seqS\circ e\right) } = \mu^{\sigma}_e(0),\]
   \[\Pr{(\sigma_{t+1},\tau_{t+1}) \text{ is } (\sigma^1,  \tau^1)  \text{ or } (\sigma^1,  \tau^0) \mid \left(T > t\right)\land \left((\sigma, \tau, \seqS) = (\sigma_t, \tau_t, \seqS_t)\right)\land \left(\seqS_{t+1} =\seqS\circ e\right) } = \mu^{\sigma}_e(1).\]
   Therefore, we have 
   \begin{equation}\label{eq-expansion-tplus1-te}
   \begin{aligned}
   &\symbolwidth \Pr{\left(T> t\right)\land \left((\sigma^0, \tau^0, \seqS') = (\sigma_{t+1}, \tau_{t+1}, \seqS_{t+1})\right)} \\
   &=\Pr{\left(T > t\right)\land \left((\sigma, \tau, \seqS) = (\sigma_t, \tau_t, \seqS_t)\right)\land \left(\seqS_{t+1} =\seqS\circ e\right) }\cdot \mu^{\sigma}_e(0),
   \end{aligned}
   \end{equation}
    \begin{equation}\label{eq-expansion-tplus1-te-one}
   \begin{aligned}
   &\symbolwidth \Pr{\left(T> t\right)\land \left( (\sigma_{t+1}, \tau_{t+1}, \seqS_{t+1}) \text{ is } (\sigma^1, \tau^1, \seqS') \text{ or } (\sigma^1, \tau^0, \seqS'\right)} \\
   &=\Pr{\left(T > t\right)\land \left((\sigma, \tau, \seqS) = (\sigma_t, \tau_t, \seqS_t)\right)\land \left(\seqS_{t+1} =\seqS\circ e\right) }\cdot \mu^{\sigma}_e(1),
   \end{aligned}
   \end{equation}

   At first, we prove 
   \[p^{\sigma}_{\sigma,\tau, \seqS,e} = p^{\sigma \land (e\gets 0)}_{\sigma \land (e\gets 0),\tau\land (e\gets 0), \seqS \circ e}.\]
   Note that 
   \begin{equation}\label{eq-pr-stsprime-pr-stse}
    \begin{aligned}
   &\symbolwidth \Pr[\!{cp}]{(\sigma^0, \tau^0, \seqS')}\\
   (\text{by \eqref{eq-def-pro-stsvl}})\quad&= \Pr{\left(T\geq t+1\right)\land \left((\sigma^0, \tau^0, \seqS') = (\sigma_{t+1}, \tau_{t+1}, \seqS_{t+1})\right)}\\
   (\text{by \eqref{eq-expansion-tplus1-te}})\quad&= \Pr{\left(T > t\right)\land \left((\sigma, \tau, \seqS) = (\sigma_t, \tau_t, \seqS_t)\right)\land \left(\seqS_{t+1} =\seqS\circ e\right) }\cdot \mu^{\sigma}_e(0)\\
   (\text{by \eqref{eq-def-pro-stsvle}})\quad   &=\Pr[\!{cp}]{(\sigma, \tau, \seqS,e)}\cdot \mu^{\sigma}_e(0).
   \end{aligned}
   \end{equation}
   Moreover, we also have 
   \begin{align}\label{eq-xsimmu-sigmaprime-sigma}
   \mu^{\sigma_{\bot}}_{\seqS'}(\sigma^0) =\mu^{\sigma_{\bot}}_{\seqS}(\sigma) \cdot \mu^{\sigma}_e(0).
   \end{align}
   Thus, we have 
   \begin{align*}
    &\symbolwidth p^{\sigma}_{\sigma,\tau, \seqS,e} \\
(\text{by \eqref{eqn-marginal-inner}})\quad    &= \Pr[\!{cp}]{(\sigma, \tau, \seqS,e)}/\mu^{\sigma_{\bot}}_{\seqS}(\sigma)\\
 (\text{by \eqref{eq-pr-stsprime-pr-stse}})\quad&= \Pr[\!{cp}]{(\sigma^0, \tau^0, \seqS')}/\left(\mu^{\sigma_{\bot}}_{\seqS}(\sigma^0)\cdot \mu^{\sigma}_e(0)\right) \\
 (\text{by \eqref{eq-xsimmu-sigmaprime-sigma}})\quad&= \Pr[\!{cp}]{(\sigma^0, \tau^0, \seqS')}/\mu^{\sigma_{\bot}}_{\seqS'}(\sigma^0) 
 \\(\text{by \eqref{eqn-marginal-all}})\quad&= p^{\sigma^0}_{\sigma^0,\tau^0, \seqS'}\\
 &= p^{\sigma \land (e\gets 0)}_{\sigma \land (e\gets 0),\tau\land (e\gets 0), \seqS \circ e}.
   \end{align*}

    In the next, we prove 
    \[
    p^{\sigma}_{\sigma, \tau, \seqS, e}=p^{\sigma\land (e\gets 1)}_{\sigma\land (e\gets 1),\tau\land (e\gets 0), \seqS\circ e} + p^{\sigma\land (e\gets 1)}_{\sigma\land (e\gets 1),\tau\land (e\gets 1),\seqS \circ e}.\]
    By \eqref{eq-def-pro-stsvl}, we have
    \begin{equation}\label{eq-pr-stsprime-pr-stse-11}
    \begin{aligned}
    \Pr[\!{cp}]{(\sigma^1, \tau^1, \seqS')}= \Pr{\left(T\geq t+1\right)\land \left((\sigma^1, \tau^1, \seqS') = (\sigma_{t+1}, \tau_{t+1}, \seqS_{t+1})\right)}
   \end{aligned}
   \end{equation}
   Similarly, we also have 
     \begin{equation}\label{eq-pr-stsprime-pr-stse-10}
    \begin{aligned}
   \Pr[\!{cp}]{(\sigma^1, \tau^0, \seqS')}
   = \Pr{\left(T\geq t+1\right)\land \left((\sigma^1, \tau^0, \seqS') = (\sigma_{t+1}, \tau_{t+1}, \seqS_{t+1})\right)}
   \end{aligned}
   \end{equation}
   Thus, we have
     \begin{equation}\label{eq-pr-stsprime-pr-stse-11-10}
    \begin{aligned}
   &\symbolwidth \Pr[\!{cp}]{(\sigma^1, \tau^1, \seqS')} + \Pr[\!{cp}]{(\sigma^1, \tau^0, \seqS')}\\
   (\text{by \eqref{eq-pr-stsprime-pr-stse-11}, \eqref{eq-pr-stsprime-pr-stse-10} and \eqref{eq-expansion-tplus1-te-one}})\quad&= \Pr{\left(T > t\right)\land \left((\sigma, \tau, \seqS) = (\sigma_t, \tau_t, \seqS_t)\right)\land \left(\seqS_{t+1} =\seqS\circ e\right) }\cdot \mu^{\sigma}_e(1)\\
   (\text{by \eqref{eq-def-pro-stsvle}})\quad   &=\Pr[\!{cp}]{(\sigma, \tau, \seqS,e)}\cdot \mu^{\sigma}_e(1).
   \end{aligned}
   \end{equation}
   Moreover, we also have 
   \begin{align}\label{eq-xsimmu-sigmaprime-sigma-1}
   \mu^{\sigma_{\bot}}_{\seqS'}(\sigma^1)
   = \mu^{\sigma_{\bot}}_{\seqS}(\sigma) \cdot \mu^{\sigma}_e(1).
   \end{align}
   Thus, we have 
   \begin{align*}
    &\symbolwidth p^{\sigma}_{\sigma,\tau, \seqS,e} \\
(\text{by \eqref{eqn-marginal-inner}})\quad    &= \Pr[\!{cp}]{(\sigma, \tau, \seqS,e)}/\mu^{\sigma_{\bot}}_{\seqS}(\sigma)\\
 (\text{by \eqref{eq-pr-stsprime-pr-stse-11-10}})\quad&= \left(\Pr[\!{cp}]{(\sigma^1, \tau^1, \seqS')}+\Pr[\!{cp}]{(\sigma^1, \tau^0, \seqS')}\right)/\left(\mu^{\sigma_{\bot}}_{\seqS}(\sigma) \cdot \mu^{\sigma}_e(1)
 \right)\\
 (\text{by \eqref{eq-xsimmu-sigmaprime-sigma-1}})\quad&= \left(\Pr[\!{cp}]{(\sigma^1, \tau^1, \seqS')}+\Pr[\!{cp}]{(\sigma^1, \tau^0, \seqS')}\right)/\mu^{\sigma_{\bot}}_{\seqS'}(\sigma^1)
 \\(\text{by \eqref{eqn-marginal-all}})\quad&= p^{\sigma^1}_{\sigma^1,\tau^1, \seqS'}+p^{\sigma^1}_{\sigma^1,\tau^0, \seqS'}\\
 &=p^{\sigma\land (e\gets 1)}_{\sigma\land (e\gets 1),\tau\land (e\gets 0), \seqS\circ e} + p^{\sigma\land (e\gets 1)}_{\sigma\land (e\gets 1),\tau\land (e\gets 1),\seqS \circ e}.
   \end{align*}
   
 \underline{Proof of (4).} By \eqref{eqn-marginal-all}, we have 
\begin{align*}
    p^{\tau}_{\sigma, \tau, \seqS} \cdot \frac{\mu_{e_{\bot}}(1)}{\mu_{e_{\bot}}(0)}\cdot \frac{ \mu(\tau)}{ \mu(\sigma)} = \frac{\Pr[\!{cp}]{(\sigma, \tau, \seqS)}\cdot\mu_{e_{\bot}}(1)\cdot\mu(\tau)}{\mu^{\tau_{\bot}}_{\seqS}(\tau)\cdot\mu_{e_{\bot}}(0) \cdot\mu(\sigma)}= \frac{\Pr[\!{cp}]{(\sigma, \tau, \seqS)}\cdot\mu_{e_{\bot}(1)}}{\mu(\sigma)}= \frac{\Pr[\!{cp}]{(\sigma,\tau, \seqS)}}{\mu^{\sigma_{\bot}}_{\seqS}(\sigma)}= p^{\sigma}_{\sigma, \tau, \seqS}.
\end{align*}
\end{proof}

\CouplingError*
\begin{proof}
    Given any $(\sigma, \tau, \seqS) \in \+V\setminus \+L$, let $v = v(\sigma,\tau)$. 
    By \Cref{def:truncated-coupling-tree} and  $(\sigma, \tau, \seqS) \not \in \+L$, we have $(\sigma,\tau,\seqS)$ is feasible and $E^{\sigma}_v \neq \emptyset$.
    Assume $\emph{w.l.o.g.}$ ${\!{Ham}\left(\sigma, {E_{v}}\right)} < {\!{Ham}\left(\tau, {E_{v}}\right)}$.
    Let $t$ denote $\abs{\seqS}$ and $r$ denote $\abs{E^{\sigma}_v}$. 
    Define a sequence of edges $e_1,e_2,\cdots,e_r\in E^{\sigma}_v$ recursively as follows.
    For each $0\leq i<r$, let $e_{i+1}$ be the first edge in $E_{v}^{\sigma^{i}}$ with $\mu^{\sigma^{i}}_{e_{i+1}}(1) \geq \mu^{\tau^{i}}_{e_{i+1}}(1)$ where 
    \[\sigma^i \triangleq \sigma\land(e_1\leftarrow 0)\land \cdots\land (e_{i}\leftarrow 0), \quad \tau^i \triangleq \tau\land(e_1\leftarrow 0)\land \cdots\land (e_{i}\leftarrow 0).\]
    Specifically, we have $\sigma^0 = \sigma$ and $\tau^0 = \tau$.
    We remark that there must be an edge $e\in E_{v}^{\sigma^{i}}$ with $\mu^{\sigma^{i}}_{e}(1) \geq \mu^{\tau^{i}}_{e}(1)$, by 
    $\!{Ham}\left(\sigma^i, {E_v}\right) = \!{Ham}\left(\sigma, {E_v}\right) < \!{Ham}\left(\tau, {E_v}\right) = \!{Ham}\left(\tau^i, {E_v}\right)$,
    $\abs{E_{v}^{\sigma^{i}}} = \abs{E_{v}^{\sigma}} - i = r - i>0$ and the third item of \Cref{prop:coupling-correctness}. 
    Thus, $e_{i+1}$ is well-defined.
    For each $0\leq i<r$, Define 
    \[\seqS^i \triangleq \seqS\circ e_1\circ\cdots\circ e_{i-1}.\]
    Specifically, we have $\seqS^0 = \seqS$.
    
    We claim that 
    \begin{align}\label{eq-sigmar-sigma}
     \Pr[\!{cp}]{(\sigma^r, \tau^r, \seqS^r)} 
    \geq \Pr[\!{cp}]{(\sigma, \tau, \seqS)}\cdot B.
        \end{align}
    Meanwhile, by $(\sigma,\tau,\seqS)$ is feasible and $(\sigma,\tau,\seqS)\in V(\+T)$, we have $(\sigma,\tau,\seqS)\in \+V$. 
    By $(\sigma,\tau,\seqS)$ is feasible, $\sigma^r = \sigma\land  (E^{\sigma}_v \gets \boldsymbol{0}),\tau^r = \tau\land  (E^{\sigma}_v \gets \boldsymbol{0})$ and  \Cref{lem:marginal-bound}, one can also verify that 
    $(\sigma^r, \tau^r, \seqS^r)$ is also feasible.
    Moreover, by the condition of this lemma, we have 
    $(\sigma,\tau,\seqS)\in V(\+T)$ and $L(\sigma, \tau)<\ell$.
    Combined with \Cref{def:truncated-coupling-tree},
    one can also verify that $(\sigma^r, \tau^r, \seqS^r)\in V(\+T)$
    by induction.
    Therefore, we also have $(\sigma^r, \tau^r, \seqS^r)\in \+V$.
    Thus, we have 
    \begin{align*}
        &\quad p^{\sigma^r}_{\sigma^r, \tau^r, \seqS^r}\\
    (\text{by $(\sigma^r, \tau^r, \seqS^r)\in \+V$ and \eqref{eqn-marginal-all}}) \quad   &= \Pr[\!{cp}]{(\sigma^r,\tau^r, \seqS^r)}/ \mu^{\sigma_{\bot}}_{\seqS^r}(\sigma^r)\\
        &\geq \Pr[\!{cp}]{(\sigma^r,\tau^r, \seqS^r)}/ \mu^{\sigma_{\bot}}_{\seqS}(\sigma)\\
    (\text{by \eqref{eq-sigmar-sigma}})\quad    &\geq B\cdot \Pr[\!{cp}]{(\sigma,\tau, \seqS)}/ \mu^{\sigma_{\bot}}_{\seqS}(\sigma)\\
    (\text{by $(\sigma, \tau, \seqS)\in \+V$ and \eqref{eqn-marginal-all}}) \quad   &= B\cdot p^{\sigma}_{\sigma, \tau, \seqS}
    \end{align*}
    In addition, by $(\sigma^r, \tau^r, \seqS^r)\in V(\+T)$ and $\sigma^r = \sigma\land  (E^{\sigma}_v \gets \boldsymbol{0}),\tau^r = \tau\land  (E^{\sigma}_v \gets \boldsymbol{0})$, we have 
    $(\sigma^r, \tau^r, \seqS^r)\in \+D$ by \Cref{def-notation-v-tct}.
    Thus, we have $\+D\neq \emptyset$ and 
    \[\sum_{(\sigma', \tau', \seqS') \in \+D} p^{\sigma'}_{\sigma', \tau', \seqS'} \geq p^{\sigma^r}_{\sigma^r, \tau^r, \seqS^r} \geq  B \cdot p^{\sigma}_{\sigma, \tau, \seqS}.\]
    Similarly, one can also prove 
    \[\sum_{(\sigma',\tau', \seqS') \in \+D} p^{\tau'}_{\sigma', \tau', \seqS} \geq  B \cdot p^{\tau}_{\sigma,\tau, \seqS}.\]     
    Therefore, the lemma is immediate.

    In the following, we prove \eqref{eq-sigmar-sigma}, which completes the proof of the lemma.    
    Recall the $\ell$-truncated 
    random process $P^{\!{cp}} \triangleq P^{\!{cp}}_\ell(\Phi, \sigma_\bot, \tau_\bot, v_\bot) = \set{(\sigma_t, \tau_t, \seqS_t, v_t, L_t)}_{0\leq t \leq T}$.
    Let $\+E_i$ denote the event $(T\geq t+i)\land \left((\sigma_{t+i},\tau_{i+i}, \seqS_{t+i})= (\sigma^i,\tau^i,\seqS^i)\right)$.
    Assume $\+E_i$ happens.
    We have 
    $L(\sigma_{t+i},\tau_{t+i}) = L(\sigma^i,\tau^i) = L(\sigma,\tau)<\ell$.
    Meanwhile, recalling that $\abs{E_{v}^{\sigma^{i}}}>0$,
    we have $E_{v}^{\sigma_{t+i}} = E_{v}^{\sigma^{i}} \neq \emptyset$.
    Combining $L(\sigma_{t+i},\tau_{t+i})<\ell, E_{v}^{\sigma_{t+i}}  \neq \emptyset$ with  \Cref{def:truncated-random-process}, 
    we have 
    \begin{itemize}
    \item[$\circ$] $T>t+i$,
    \item[$\circ$] $e_{i+1}$ is the selected edge at the state $(\sigma_{t+i},\tau_{t+i},\seqS_{t+i})= (\sigma^i,\tau^i,\seqS^i)$,
    \item[$\circ$] $(\sigma_{t+i+1}(e_{i+1}),\tau_{t+i+1}(e_{i+1}))$ is sampled from an optimal coupling of $\left(\mu^{\sigma_{t+i}}_{e_{i+1}},\mu^{\tau_{t+i}}_{e_{i+1}}\right) =\left(\mu^{\sigma^i}_{e_{i+1}},\mu^{\tau^i}_{e_{i+1}}\right)$. 
    \end{itemize}    
    Thus, under the condition that $\+E_i$ happens,
    $\+E_{i+1}$ happens only if $\sigma_{t+i+1}(e_{i+1})=\tau_{t+i+1}(e_{i+1}) = 0$. Formally, 
    \begin{align*}
    \Pr{\+E_{i+1}\mid \+E_{i}} 
    =& \Pr{(\sigma_{t+i+1}(e_{i+1})=\tau_{t+i+1}(e_{i+1}) = 0)\mid \+E_{i}} = \min\left\{\mu^{\sigma^{i}}_{e_{i+1}}(0), \mu^{\sigma^{i}}_{e_{i+1}}(0)\right\} = \mu^{\sigma^{i}}_{e_{i+1}}(0),
    \end{align*}
    where the last inequality is by $\mu^{\sigma^{i}}_{e_{i+1}}(1) \geq \mu^{\tau^{i}}_{e_{i+1}}(1)$.
    Thus, we have 
    \begin{align}\label{eq-pr-er}
         \Pr{\+E_{r}} = \Pr{\+E_{0}} \cdot\prod_{i=0}^{r-1}\Pr{\+E_{i+1}\mid \+E_{i}}
        =& \Pr{\+E_{0}} \cdot \prod_{i=0}^{r-1}\mu^{\sigma^{i}}_{e_{i+1}}(0)= \Pr{\+E_{0}} \cdot  \mu_{E_v^\sigma}^{\sigma}(\zero)
    \geq \Pr{\+E_{0}}\cdot B,
        \end{align}
where the last inequality is by \Cref{lem:marginal-bound}.
In addition, by \Cref{def-notation-trp} we have 
    \begin{align}\label{eq-prcp-e0}
    \Pr[\!{cp}]{(\sigma, \tau, \seqS)} =\Pr{\left(T\geq t\right)\land \left((\sigma, \tau, \seqS) = (\sigma_t, \tau_t, \seqS_t)\right)} = \Pr{\+E_{0}},
    \end{align}
    \begin{align}\label{eq-prcp-er}
    \Pr[\!{cp}]{(\sigma^r, \tau^r, \seqS^r)} = \Pr{\left(T\geq t+r\right)\land \left((\sigma_r, \tau_r, \seqS_r) = (\sigma_{t+r}, \tau_{t+r}, \seqS_{t+r})\right)} = \Pr{\+E_{r}}.
    \end{align}
Therefore, \eqref{eq-sigmar-sigma} is immediate by \eqref{eq-pr-er}, \eqref{eq-prcp-e0} and \eqref{eq-prcp-er}.
\end{proof}

\subsection{Omitted proofs for properties of LP} \label{sec:omitted-proofs-LP}
\BuildingCost*
\begin{proof}
    To prove this lemma, it is sufficient to show that the LP can be constructed in time $\poly\left(\Delta^{\Delta \ell}\right)$. 
    Thus, the total number of variables and constraints in the LP is no more than $\poly\left(\Delta^{\Delta \ell}\right)$.
    Therefore, the feasibility of the LP can be checked in time $\poly\left(\Delta^{\Delta \ell}\right)$, because LP has polynomial-time algorithms.

    In the following, we show the LP can be constructed in time $\poly\left(\Delta^{\Delta \ell}\right)$. 
    By \Cref{lem:partial-assignment-feasibility},
    the feasibility of each node $(\sigma,\tau,\seqS)\in V(\+T)$ can be checked within time $O(\abs{\Lambda(\sigma)}+\abs{\Lambda(\tau)})$.
    Combined with Item~\eqref{item:CT-property-size} in~\Cref{prop:property-of-truncated-tree}, we have 
    $O(\abs{\Lambda(\sigma)}+\abs{\Lambda(\tau)})\leq O(\Delta\ell)$.
    Thus, the feasibility of each node in $\+T$ can be checked within time $O(\Delta\ell)$.
    In addition, by Item~\eqref{item:CT-property-size} of~\Cref{prop:property-of-truncated-tree}, 
    we have $\abs{V(\+T)}\leq \left(3\Delta\right)^{\Delta \ell + 1}$.
    Thus, by \Cref{def:truncated-coupling-tree}, $\+T$ can be constructed within cost $O\left(\Delta \ell \cdot \BuildTime\right) = \BuildTime$,
    where the term $\Delta \ell$ is due to the cost of checking the feasibility of a node, and the term
    $\BuildTime$ is due to the size of $\+T$.
    In the process of constructing $\+T$,
    one can also obtain the set $V(\+T),\+V,\+L$ and $\+L_{\good}$.
    

    Note that the number of variables in the LP is no more than $4\abs{V(\+T)}$, because each pair of variables $\widehat{p}_{\sigma, \tau, \seqS}^{\sigma}$ and $\widehat{p}_{\sigma, \tau, \seqS}^{\tau}$ are corresponding to a node $(\sigma,\tau,\seqS)\in V(\+T)$,
    and each pair of variables $\widehat{p}_{\sigma,\tau,\seqS,e}^{\sigma}$ and $\widehat{p}_{\sigma,\tau,\seqS,e}^{\tau}$ where $e\in E_{v(\sigma,\tau)}^{\sigma}$
    are corresponding to a node $(\sigma\land(e\leftarrow 0),\tau\land(e\leftarrow 0),\seqS\circ e)\in V(\+T)$.
    Combined with $\abs{V(\+T)}\leq \left(3\Delta\right)^{\Delta \ell + 1}$, we have
    Constraint~\ref{item-first-LP} can be constructed in time $\BuildTime$.
    Similarly, one can also verify that Constraints~\ref{item-second-LP}, \ref{item-third-LP} and \ref{item-sixth-LP} can be constructed in time $\BuildTime$.
    Now we turn to Constraint~\ref{item-forth-LP}.
    For any $(\sigma, \tau, \seqS) \in \+L_{\!{good}} \cap \+V$, by Item~\eqref{item:CT-property-ratio} in~\Cref{prop:property-of-truncated-tree}, 
    we have $\mu(\tau)/\mu(\sigma) = f_v\left(\abs{\tau(E_v)}\right)/f_v\left(\abs{\sigma(E_v)}\right)$.
    Thus, one can verify that $\mu(\tau)/\mu(\sigma)$ can be calculated in time $O(\abs{E_v}) = O(\Delta)$.
    Thus Constraint~\ref{item-forth-LP} can be constructed in time 
    $O(\Delta)\cdot \abs{V(\+T)} = \poly\left(\Delta^{\Delta \ell}\right)$.
    At last, we consider Constraint~\ref{item-fifth-LP}.
    For each node $(\sigma,\tau,\seqS)\in \+V\setminus\+L$, we claim that the set $\+D(\sigma,\tau,\seqS)$ can be constructed in time $O(\Delta\ell\cdot V(\+T)) = \poly\left(\Delta^{\Delta \ell}\right)$. 
    Thus, Constraints~\ref{item-fifth-LP} can be constructed in time $\poly\left(\Delta^{\Delta \ell}\right)\cdot \poly\left(\Delta^{\Delta \ell}\right) = \poly\left(\Delta^{\Delta \ell}\right)$.
    In the following, we prove the claim that $\+D(\sigma,\tau,\seqS)$ can be constructed in time $O(\Delta\ell\cdot V(\+T))$.
    Note that $\+D(\sigma,\tau,\seqS)\subseteq V(\+T)$.
    In addition, for each $(\sigma',\tau',\seqS')\in V(\+T)$,
    by \Cref{def-notation-v-tct} and Item~\eqref{item:CT-property-size} in~\Cref{prop:property-of-truncated-tree}, one can check whether $(\sigma',\tau',\seqS')\in \+D(\sigma,\tau,\seqS)$ within time $O(\Lambda(\sigma) + \Lambda(\tau) + \Lambda(\sigma') + \Lambda(\tau') + 2\Delta) = O(\Delta\ell)$.
    Therefore, the time cost for constructing $\+D(\sigma,\tau,\seqS)$ is $O(\Delta\ell\cdot V(\+T))$.
    In summary, the LP can be constructed in time $\poly\left(\Delta^{\Delta \ell}\right)$. The lemma is proved.
\end{proof}

\LPFeasibility*
\begin{proof}
    Set the variables of the LP as follows:
    \begin{itemize}
        \item For each $(\sigma, \tau, \seqS) \in V(\+T)$, let $\widehat{p}_{\sigma, \tau, \seqS}^{\sigma} = p_{\sigma, \tau, \seqS}^{\sigma}$ and $\widehat{p}_{\sigma,\tau, \seqS}^{\tau} = p_{\sigma,\tau, \seqS}^{\tau}$.
        \item For each $(\sigma,\tau,\seqS)$ in $\+V\setminus \+L$ and  $e\in E_v^{\sigma}$ where $v = v(\sigma,\tau)$, let  $\widehat{p}_{\sigma,\tau,\seqS,e}^{\sigma} = p_{\sigma,\tau,\seqS,e}^{\sigma}$, and $\widehat{p}_{\sigma,\tau,\seqS,e}^{\tau} = p_{\sigma,\tau,\seqS,e}^{\tau}$.
    \end{itemize}
    By \Cref{prop:coupling-linear-constraint}, Items \ref{item-first-LP}-\ref{item-forth-LP} of the constraints are satisfied.
    By \Cref{lem:coupling-error}, Item \ref{item-fifth-LP} of the constraints is also satisfied.
    In addition, by \Cref{def-key-quantity}, we have $ p^{\sigma}_{\sigma, \tau, \seqS}= p^{\tau}_{\sigma, \tau, \seqS} = 0$ for each node $(\sigma, \tau, \seqS)\in V(\+T)\setminus \+V$.
    Thus, Item \ref{item-sixth-LP} of the constraints is also satisfied.
    In summary, all the constraints of the LP are satisfied.
    The lemma is proved.
\end{proof}

\RatioSum*

\begin{proof}
    To prove this lemma, it is sufficient to prove \eqref{eqn-sum-leftside}. Similarly, one can also prove \eqref{eqn-sum-rightside}.
    In the following, we prove \eqref{eqn-sum-leftside}.
    Recall that $\+L$ is the set of leaf nodes in $V(\+T)$, where $\+T$ is the $\ell$-truncated extended coupling tree. 
    Let $\+T^{i}$ be the subtree of $\+T$ obtained by deleting all nodes at a depth greater than $i$
    and $\+L^{i}$ be the leaf nodes of $\+T^{i}$.
    We will prove 
    \begin{align}\label{eqn-subtree-l-tct}
        \forall x\in \sigma_\bot,0\leq i\leq \ell, \quad \sum_{(\sigma,\tau, \seqS)\in \+L^i: \ x\in \sigma} \widehat{p}^{\sigma}_{\sigma,\tau,\seqS}=1,
    \end{align}
    Then \eqref{eqn-sum-leftside} is immediate.    
    We prove \eqref{eqn-subtree-l-tct} by induction on $i$.
    The induction basis is when $i = 0$.
    In this case, by \Cref{def:truncated-coupling-tree} we have $\+T^0$ is a tree with a unique node $(\sigma_\bot, \tau_\bot, \varnothing)$.
    In addition, by \Cref{def:induced-LP}, we have 
    $\widehat{p}^{\sigma_\bot}_{\sigma_\bot,\tau_\bot, \varnothing} = 1$.
    Therefore, 
     \begin{align*}
        \forall x\in \sigma_\bot, \quad \sum_{(\sigma,\tau, \seqS)\in \+L^0: \ x\in \sigma} \widehat{p}^{\sigma}_{\sigma,\tau,\seqS}=\widehat{p}^{\sigma_\bot}_{\sigma_\bot,\tau_\bot, \varnothing} = 1.
    \end{align*}
    The base case is proved. For the induction step, we fix an $x\in \sigma_\bot$.    
    For each node $(\sigma,\tau,\seqS)\in \+L^i$,
    let $\+C \triangleq \+C((\sigma,\tau,\seqS))$ be the children of $(\sigma,\tau,\seqS)$ in $\+T^{i+1}$. 
    By $(\sigma,\tau,\seqS) \in \+L^{i}$, we have 
    $\+C\subseteq \+L^{i+1}$.
    Define 
    \begin{align}\label{def-cplus}
    \+C^{+} \triangleq \+C^{+}((\sigma,\tau,\seqS)) = \left(\+C\cup \{(\sigma,\tau,\seqS)\}\right)\cap \+L^{i+1}.
    \end{align}
    We claim that 
    \begin{align}\label{eq-extension-sum}
    \forall (\sigma,\tau,\seqS)\in \+L^i,\quad \widehat{p}^{\sigma}_{\sigma,\tau,\seqS}\cdot \id{x \in \sigma} = \sum_{(\sigma',\tau', \seqS')\in \+C^{+}} \widehat{p}^{\sigma'}_{\sigma',\tau',\seqS'}\cdot \id{\ x\in \sigma'}.
    \end{align}
    Thus, we have 
    \begin{align*}
        \sum_{(\sigma,\tau, \seqS)\in \+L^{i+1}} \widehat{p}^{\sigma}_{\sigma,\tau,\seqS}\cdot \id{\ x\in \sigma}&=\sum_{u \in \+L^{i}}\sum_{(\sigma',\tau', \seqS')\in \+C^{+}(u)} \widehat{p}^{\sigma'}_{\sigma',\tau',\seqS'}\cdot \id{ x\in \sigma'}\\
        &=\sum_{(\sigma,\tau, \seqS)\in \+L^{i}}\widehat{p}^{\sigma}_{\sigma,\tau,\seqS}\cdot \id{x \in \sigma}\\
        &=\sum_{(\sigma,\tau, \seqS)\in \+L^i:\  x\in \sigma} \widehat{p}^{\sigma}_{\sigma,\tau,\seqS}\\
        &=1,
    \end{align*}
    where the second equality is by \eqref{eq-extension-sum} and the last equality is by the induction hypothesis.
    This completes the induction step and \eqref{eqn-subtree-l-tct} is immediate.

    In the following, we prove \eqref{eq-extension-sum}, which completes the proofs of \eqref{eqn-subtree-l-tct} and the lemma. Fix a $(\sigma,\tau,\seqS)\in \+L^i$.  
    Let $v=v(\sigma,\tau)$ and $L=L(\sigma,\tau)$.
    If $x\not \in \sigma$, we have $x\not\in \sigma'$ for each $(\sigma',\tau',\seqS')\in \+C$.
    Thus, 
    \begin{align*}
    \sum_{(\sigma',\tau', \seqS')\in \+C^{+} } \widehat{p}^{\sigma'}_{\sigma',\tau',\seqS'}\cdot \id{x\in \sigma'}=\sum_{(\sigma',\tau', \seqS')\in \+C} \widehat{p}^{\sigma'}_{\sigma',\tau',\seqS'}\cdot\id{ x\in \sigma'} =0 = \widehat{p}^{\sigma}_{\sigma,\tau,\seqS}\cdot \id{x \in \sigma}.
    \end{align*}
    Then \eqref{eq-extension-sum} is proved.
    Otherwise, $x \in \sigma$.     
    By \Cref{def:truncated-coupling-tree},
    there are only two possibilities.
    \begin{itemize}
    \item If $L\geq \ell$ or $E_v^{\sigma}=\emptyset$ or $(\sigma, \tau, \seqS)$ is infeasible, then $(\sigma,\tau,\seqS)$ is a leaf node in $\+T$.
    Therefore, we have $(\sigma,\tau,\seqS)\in \+L^{i+1}$ and $\+C = \emptyset$.
    Thus, $\+C^+ = (\sigma, \tau, \seqS)$.
    We have 
    \begin{align*}
    \sum_{(\sigma',\tau', \seqS')\in \+C^{+}} \widehat{p}^{\sigma'}_{\sigma',\tau',\seqS'}\cdot \id{x\in \sigma'} = \widehat{p}^{\sigma}_{\sigma,\tau,\seqS} = \widehat{p}^{\sigma}_{\sigma,\tau,\seqS}\cdot \id{x \in \sigma}.
    \end{align*}
    Then \eqref{eq-extension-sum} is proved.    
    \item Otherwise, $(\sigma,\tau,\seqS)$ is not a leaf node in $\+T$. Combining with $(\sigma,\tau,\seqS) \in \+L^{i}$, we have 
    $(\sigma,\tau,\seqS) \not\in \+L^{i+1}$.
    Thus, by \eqref{def-cplus} we have $(\sigma,\tau,\seqS) \not\in \+C^{+}$.
    We claim that 
    \begin{align}\label{eq-extension-sum-edge}
     \forall e = \set{u,v} \in E_v^{\sigma}, \quad \widehat{p}^{\sigma}_{\sigma,\tau,\seqS,e}\cdot \id{x \in \sigma} = \sum_{(\sigma',\tau', \seqS\circ e)\in \+C } \widehat{p}^{\sigma'}_{\sigma',\tau',\seqS\circ e}\cdot \id{x\in \sigma'}.
    \end{align}
    In addition, by $(\sigma,\tau,\seqS)$ is not a leaf node in $\+T$ and \Cref{prop:property-of-truncated-tree}, we have 
    $(\sigma,\tau,\seqS)\in \+V\setminus \+L$.
    Moreover, by \Cref{def:truncated-coupling-tree}, we have each $(\sigma',\tau', \seqS')\in \+C$ satisfies $\seqS' = \seqS \circ e$ for some $e \in E_v^{\sigma}$. Formally,
    \begin{align}\label{eq-c-ce}
    \+C = \left\{(\sigma',\tau', \seqS\circ e)\mid (\sigma',\tau', \seqS\circ e)\in \+C,e\in E_v^{\sigma}\right\}.
    \end{align}
    Thus, we have 
    \begin{align*}
       &\symbolwidth \widehat{p}^{\sigma}_{\sigma, \tau, \seqS} \cdot \id{x \in \sigma} \\
       (\text{by $(\sigma,\tau,\seqS)\in \+V\setminus \+L$ and \eqref{eqn-hat-inter-sum1}})\quad &= \sum_{e \in  E_v^{\sigma}} \left(\widehat{p}^{\sigma}_{\sigma, \tau, \seqS, e}\cdot \id{x \in \sigma}\right)\\ 
       (\text{by \eqref{eq-extension-sum-edge}})\quad &= \sum_{e \in  E_v^{\sigma}}\sum_{(\sigma',\tau', \seqS\circ e)\in \+C} \widehat{p}^{\sigma'}_{\sigma',\tau',\seqS\circ e}\cdot\id{ x\in \sigma'}\\
       (\text{by \eqref{eq-c-ce}})\quad  &= \sum_{(\sigma',\tau', \seqS')\in \+C} \widehat{p}^{\sigma'}_{\sigma',\tau',\seqS'}\cdot\id{ x\in \sigma'}\\
(\text{by $(\sigma,\tau,\seqS) \not\in \+C^{+}$})\quad &= \sum_{(\sigma',\tau', \seqS')\in \+C^{+}} \widehat{p}^{\sigma'}_{\sigma',\tau',\seqS'}\cdot\id{ x\in \sigma'}.
    \end{align*}
    Then \eqref{eq-extension-sum} is immediate.
    
    In the following, we prove \eqref{eq-extension-sum-edge}, which completes the proof of \eqref{eq-extension-sum}.
    Recall that $x\in \sigma$, $(\sigma,\tau,\seqS)$ is not a leaf node in $\+T$, $(\sigma,\tau,\seqS) \in \+L^{i}$, $v=v(\sigma,\tau)$ and $L=L(\sigma,\tau)$.
    Fix an edge $ e = \set{u,v} \in E_v^{\sigma}$.
    Let $\sigma^{a} = \sigma \land (e \gets a)$ and $\tau^{a} = \tau \land (e \gets a)$ for each $a\in \{0,1\}$.    
    We prove \eqref{eq-extension-sum-edge} by considering two separate cases. 
        \begin{enumerate}[(i)]
            \item ${\!{Ham}\left(\sigma,{E_v}\right)}<{\!{Ham}\left(\tau,{E_v}\right)}$. In this case, by \Cref{def:truncated-coupling-tree}, $(\sigma,\tau,\seqS)$ has three children related to $e$ in $\+L^{i+1}$, i.e., $(\sigma^0, \tau^0, \seqS \circ e)$, $(\sigma^1, \tau^0, \seqS \circ e)$, and $(\sigma^1, \tau^1, \seqS \circ e)$ . Assume \emph{w.l.o.g.} $x(e) = 0$.
            We have 
            \[\left\{(\sigma',\tau', \seqS\circ e)\mid (\sigma',\tau', \seqS\circ e)\in \+C,x\in \sigma'\right\} = \{ (\sigma^0, \tau^0, \seqS \circ e)\}.\]          
            Combining with \eqref{eqn-hat-inner-child-sum1}, we have 
             \begin{align*}
            \widehat{p}^{\sigma}_{\sigma,\tau,\seqS,e}\cdot \id{x \in \sigma} = \widehat{p}^{\sigma}_{\sigma,\tau,\seqS,e}=\widehat{p}^{\sigma^0}_{\sigma^0,\tau^0, \seqS \circ e} = \sum_{(\sigma',\tau', \seqS\circ e)\in \+C} \widehat{p}^{\sigma'}_{\sigma',\tau',\seqS\circ e}\cdot \id{ x\in \sigma'}. 
            \end{align*}
            Then \eqref{eq-extension-sum-edge} is proved. \item ${\!{Ham}\left(\sigma,{E_v}\right)}\geq{\!{Ham}\left(\tau,{E_v}\right)}$. In this case, by \Cref{def:truncated-coupling-tree},  $(\sigma,\tau,\seqS)$ has three children related to $e$ in $\+L^{i+1}$, i.e., $(\sigma^0, \tau^0, \seqS \circ e)$, $(\sigma^0, \tau^1, \seqS \circ e)$, and $(\sigma^1, \tau^1, \seqS\circ e)$. Assume \emph{w.l.o.g.} $x(e) = 0$.
            We have 
            \[\left\{(\sigma',\tau', \seqS\circ e)\mid (\sigma',\tau', \seqS\circ e)\in \+C,x\in \sigma'\right\} = \{ (\sigma^0, \tau^0, \seqS \circ e),(\sigma^0, \tau^1, \seqS \circ e)\}.\]          
            Combining with \eqref{eqn-hat-inner-child-sum3}, we have 
             \begin{align*}
            \widehat{p}^{\sigma}_{\sigma,\tau,\seqS,e}\cdot \id{x \in \sigma} = \widehat{p}^{\sigma}_{\sigma,\tau,\seqS,e}=\widehat{p}^{\sigma^0}_{\sigma^0,\tau^0, \seqS \circ e} +\widehat{p}^{\sigma^0}_{\sigma^0,\tau^1, \seqS \circ e} = \sum_{(\sigma',\tau', \seqS\circ e)\in \+C} \widehat{p}^{\sigma'}_{\sigma',\tau',\seqS\circ e}\cdot \id{ x\in \sigma'}.  
            \end{align*}
            Then \eqref{eq-extension-sum-edge} is proved.
        \end{enumerate} 
    \end{itemize}
\end{proof}

\LPTruncatedError*
\begin{proof}
To prove this lemma, it is sufficient to prove \eqref{eqn-error1}.
Similarly, one can also prove \eqref{eqn-error2}.
In the following, we prove $\eqref{eqn-error1}$.
Recall the $\ell$-truncated extended coupling tree $\+T_{\ell}$ in \Cref{def:truncated-coupling-tree}, 
and the $\+L,\+L_{\!{good}},\+L_{\!{bad}},\+D(\cdot)$ related to $\+T_{\ell}$ in \Cref{def-notation-v-tct}.
In the following proof, we will use $\+L(\ell),\+L_{\!{good}}(\ell),\+L_{\!{bad}}(\ell),\+D_{\ell}(\cdot)$ to denote $\+L,\+L_{\!{good}},\+L_{\!{bad}},\+D(\cdot)$, such that
these concepts related to different $\ell$ can be distinguished.
To prove $\eqref{eqn-error1}$,
it is equivalent to prove
\begin{align}\label{eqn-sum-leftside-error-induction}
    \forall \ell\geq 0, \quad \sum_{(\sigma,\tau, \seqS)\in \+L_{\!{bad}}(\ell)} \widehat{p}^{\sigma}_{\sigma,\tau,\seqS} \cdot \mu^{\sigma_{\bot}}_{\seqS}(\sigma)\leq (1 - B^2)^{\ell}.
\end{align}
In the following, we prove \eqref{eqn-sum-leftside-error-induction} by induction on $\ell$.
The induction basis is when $\ell = 0$.
In this case, by \Cref{def:truncated-coupling-tree} we have $\+T_{0}$ is a tree with a unique node $(\sigma_\bot, \tau_\bot, \varnothing)$.
In addition, by \Cref{def:induced-LP}, we have 
$\widehat{p}^{\sigma_\bot}_{\sigma_\bot,\tau_\bot, \varnothing} = 1$.
Therefore, 
 \begin{align*}
    \sum_{(\sigma,\tau, \seqS)\in \+L_{\!{bad}}(\ell)} \widehat{p}^{\sigma}_{\sigma,\tau,\seqS}\cdot\mu^{\sigma_{\bot}}_{\seqS}(\sigma)=\widehat{p}^{\sigma_\bot}_{\sigma_\bot,\tau_\bot, \varnothing} = 1 = (1- B^2)^{0}.
\end{align*}
The base case is proved. For the induction step, 
let $\+S$ denote $\+L_{\!{bad}}(\ell)\setminus \+L(\ell+1)$.
We have 
\[\+S \subseteq V(\+T_{\ell+1})\setminus \+L(\ell+1)= \+V_{\ell+1}\setminus \+L(\ell+1),\]
where the last equality is by \Cref{prop:property-of-truncated-tree}.
Thus, by \Cref{def-notation-v-tct}, 
$\+D_{\ell+1}(\sigma,\tau,\seqS)$ is well-defined for each $(\sigma,\tau,\seqS)\in \+S$.
we claim that 
\begin{align}\label{eq-d-subset-1}
\+L(\ell)\setminus \+L_{\!{good}}(\ell+1) \subseteq \+S,
\end{align}
\begin{align}\label{eq-d-subset-2}
\forall (\sigma,\tau,\seqS)\in \+S, \quad \+D_{\ell+1}(\sigma,\tau,\seqS)\subseteq \+L_{\!{good}}(\ell+1)\setminus \+L(\ell),
\end{align}
\begin{align}\label{eq-d-subset-3}
\forall u=(\sigma,\tau,\seqS)\in \+S, \quad \widehat{p}^{\sigma}_{\sigma,\tau,\seqS}\mu^{\sigma_{\bot}}_{\seqS}(\sigma) - \left(\sum_{(\sigma',\tau',\seqS')\in \+D_{\ell+1}(u)} \widehat{p}^{\sigma'}_{\sigma',\tau',\seqS'}\cdot \mu^{\sigma_{\bot}}_{\seqS'}(\sigma') \right)\leq \left(1 - B^2\right)\cdot\widehat{p}^{\sigma}_{\sigma,\tau,\seqS}\mu^{\sigma_{\bot}}_{\seqS}(\sigma).
\end{align}
Thus, by \eqref{eq-d-subset-1},  \eqref{eq-d-subset-2}, \eqref{eq-d-subset-3} and the induction hypothesis (I.H.), we have
\begin{equation}\label{eq-d-diff}
\begin{aligned}
 &\quad \sum_{(\sigma,\tau, \seqS)\in \+L(\ell)} \widehat{p}^{\sigma}_{\sigma,\tau,\seqS}\cdot \mu^{\sigma_{\bot}}_{\seqS}(\sigma) - \sum_{(\sigma,\tau, \seqS)\in \+L_{\!{good}}(\ell+1)}\widehat{p}^{\sigma}_{\sigma,\tau,\seqS}\cdot \mu^{\sigma_{\bot}}_{\seqS}(\sigma)\\
\left(\text{by \eqref{eq-d-subset-1} and \eqref{eq-d-subset-2}}\right)\quad &\leq \sum_{(\sigma,\tau,\seqS)\in \+S}\widehat{p}^{\sigma}_{\sigma,\tau,\seqS}\cdot \mu^{\sigma_{\bot}}_{\seqS}(\sigma) - \sum_{u\in \+S}\ \sum_{(\sigma',\tau',\seqS')\in \+D_{\ell+1}(u)} \widehat{p}^{\sigma'}_{\sigma',\tau',\seqS'}\cdot \mu^{\sigma_{\bot}}_{\seqS'}(\sigma') \\
\left(\text{by \eqref{eq-d-subset-3}}\right)\quad &\leq
\left( 1 -B^2\right)\cdot\sum_{(\sigma,\tau,\seqS)\in \+S}\widehat{p}^{\sigma}_{\sigma,\tau,\seqS}\cdot \mu^{\sigma_{\bot}}_{\seqS}(\sigma)\\
\left(\text{by $\+S\subseteq \+L_{\!{bad}}(\ell)$}\right)
\quad &\leq
\left( 1 -B^2\right)\cdot\sum_{(\sigma,\tau,\seqS)\in \+L_{\!{bad}}(\ell)}\widehat{p}^{\sigma}_{\sigma,\tau,\seqS}\cdot \mu^{\sigma_{\bot}}_{\seqS}(\sigma)\\
\left(\text{by I.H.}\right)\quad &\leq \left(1-B^2\right)^{\ell+1}.
\end{aligned}
\end{equation}
Moreover, by \Cref{lem:ratio-identity}, we have
\begin{equation}\label{eq-lemma-ratio-identity-transform}
\begin{aligned}
    \sum_{(\sigma,\tau, \seqS)\in \+L(\ell)} \widehat{p}^{\sigma}_{\sigma,\tau,\seqS}\cdot \mu^{\sigma_{\bot}}_{\seqS}(\sigma) &= \sum_{(\sigma,\tau, \seqS)\in \+L(\ell)}\widehat{p}^{\sigma}_{\sigma,\tau,\seqS}\cdot \mu(\sigma)/\mu_{e_{\bot}}(1) = 1 \\
    &= \sum_{(\sigma,\tau, \seqS)\in \+L(\ell+1)}\widehat{p}^{\sigma}_{\sigma,\tau,\seqS}\cdot \mu(\sigma)/\mu_{e_{\bot}}(1) =  \sum_{(\sigma,\tau, \seqS)\in \+L(\ell+1)} \widehat{p}^{\sigma}_{\sigma,\tau,\seqS}\cdot \mu^{\sigma_{\bot}}_{\seqS}(\sigma)
\end{aligned}
\end{equation}
Therefore, we have
\begin{align*}
    \sum_{(\sigma,\tau, \seqS)\in \+L_{\!{bad}}(\ell+1)} \widehat{p}^{\sigma}_{\sigma,\tau,\seqS}\cdot \mu^{\sigma_{\bot}}_{\seqS}(\sigma)
    =&\sum_{(\sigma,\tau, \seqS)\in \+L(\ell+1)} \widehat{p}^{\sigma}_{\sigma,\tau,\seqS}\cdot \mu^{\sigma_{\bot}}_{\seqS}(\sigma) - \sum_{(\sigma,\tau, \seqS)\in \+L_{\!{good}}(\ell+1)}\widehat{p}^{\sigma}_{\sigma,\tau,\seqS}\cdot \mu^{\sigma_{\bot}}_{\seqS}(\sigma) \\
\left(\text{by \eqref{eq-lemma-ratio-identity-transform}}\right) \quad      =&\sum_{(\sigma,\tau, \seqS)\in \+L(\ell)} \widehat{p}^{\sigma}_{\sigma,\tau,\seqS}\cdot \mu^{\sigma_{\bot}}_{\seqS}(\sigma) - \sum_{(\sigma,\tau, \seqS)\in \+L_{\!{good}}(\ell+1)}\widehat{p}^{\sigma}_{\sigma,\tau,\seqS}\cdot \mu^{\sigma_{\bot}}_{\seqS}(\sigma)\\
\left(\text{by \eqref{eq-d-diff}}\right) \quad  \leq &\left(1-B^2\right)^{\ell+1}.
\end{align*}
This completes the induction step and \eqref{eqn-sum-leftside-error-induction} is immediate.

In the following, we prove \eqref{eq-d-subset-1},  \eqref{eq-d-subset-2} and \eqref{eq-d-subset-3}, which completes the proofs of \eqref{eqn-sum-leftside-error-induction} and the lemma. 
At first, we prove \eqref{eq-d-subset-1}.
Given any $(\sigma,\tau,\seqS) \in \+L_{\!{good}}(\ell)$, 
let $v = v(\sigma,\tau)$ and $L = L(\sigma,\tau)$.
By \Cref{def-notation-v-tct}, we have $L<\ell$.
Combining $(\sigma,\tau,\seqS)$ is a leaf node in $\+T_{\ell}$, $L<\ell$ with \Cref{def:truncated-coupling-tree},
we have either $E_v^{\sigma}=\emptyset$ or $(\sigma, \tau, \seqS)$ is infeasible.
Thus, $(\sigma, \tau, \seqS)$ is also a leaf node in $\+T_{\ell+1}$.
Formally, $(\sigma,\tau,\seqS)\in \+L(\ell+1)$.
Combined with $L<\ell$, 
we have 
$(\sigma, \tau, \seqS)\in \+L_{\!{good}}(\ell+1).$
Thus, we have 
$\+L_{\!{good}}(\ell)\subseteq \+L_{\!{good}}(\ell+1)$.
Therefore, 
$\+L(\ell)\setminus \+L_{\!{good}}(\ell+1)\subseteq \+L_{\!{bad}}(\ell)$.
In addition, given any $(\sigma,\tau,\seqS)\in \+L(\ell)\setminus \+L_{\!{good}}(\ell+1)$, we have $L(\sigma,\tau)\leq \ell < \ell+1$ by \Cref{def-notation-v-tct}.
Combined with $(\sigma,\tau,\seqS)\not \in \+L_{\!{good}}(\ell+1)$,
we have $(\sigma,\tau,\seqS)\not\in \+L(\ell+1)$.
Therefore, we have
$(\+L(\ell)\setminus \+L_{\!{good}}(\ell+1))\cap \+L(\ell+1) = \emptyset$.
In summary, we have $\+L(\ell)\setminus \+L_{\!{good}}(\ell+1) \subseteq \+L_{\!{bad}}(\ell)\setminus \+L(\ell+1)$.
Then, \eqref{eq-d-subset-1} is proved.

In the next, we prove \eqref{eq-d-subset-2}.
Given any $(\sigma,\tau,\seqS)\in \+S$, let $\+D_{\ell+1}$ denote 
$\+D_{\ell+1}(\sigma,\tau,\seqS)$ and $v$ denote $v(\sigma,\tau)$.
By $(\sigma,\tau,\seqS)\in \+S \subseteq \+L_{\!{bad}}(\ell)$,
we have $L(\sigma,\tau) = \ell$. Recall that $\+D_{\ell+1}\neq \emptyset$ by \Cref{lem:coupling-error}.
Fix an arbitrary $(\sigma',\tau',\seqS')\in \+D_{\ell+1}$.
By \Cref{condition-sigma-tau} and \Cref{def-notation-v-tct},
we have 
\[v' = v(\sigma',\tau')= v(\sigma \land (E^{\sigma}_v \gets \boldsymbol{0}),\tau \land (E^{\sigma}_v \gets \boldsymbol{0}))=v(\sigma,\tau)=v,\]
\[L' = \abs{\{e\in \Lambda(\sigma')\mid (e\neq e_{\bot})\land (\sigma'(e)\neq \tau'(e))\}} = \abs{\{e\in \Lambda(\sigma)\mid (e\neq e_{\bot})\land (\sigma(e)\neq \tau(e))\}} = L(\sigma,\tau) = \ell.\]
Thus, we have
$E^{\sigma'}_{v'} =E^{\sigma'}_{v} = \emptyset$.
Combined with \Cref{def:truncated-coupling-tree},
we have  $(\sigma',\tau',\seqS')$ is a leaf node in $\+T_{\ell+1}$.
Formally, $(\sigma',\tau',\seqS')\in \+L(\ell+1)$.
Combined with $L'=\ell<\ell+1$, we have 
$(\sigma',\tau',\seqS') \in \+L_{\!{good}}(\ell+1)$.
Recall that $(\sigma',\tau',\seqS')$ is an arbitrary node in $\+D_{\ell+1}$.
We have $\+D_{\ell+1}\subseteq \+L_{\!{good}}(\ell+1)$.
Moreover, by $(\sigma,\tau,\seqS)\in 
\+S = \+L_{\!{bad}}(\ell)\setminus \+L(\ell+1)$,
we have $(\sigma,\tau,\seqS)$ is not a leaf in $\+T_{\ell+1}$.
Combined with \Cref{def:truncated-coupling-tree}, we have 
$E^{\sigma}_{v} \neq \emptyset$.
Recalling that $E^{\sigma'}_{v} = \emptyset$,
we have $\sigma'\neq \sigma$.
Thus, $(\sigma',\tau',\seqS') \neq (\sigma,\tau,\seqS)$.
Combining with \Cref{def:truncated-coupling-tree} and $(\sigma,\tau,\seqS)\in \+L_{\!{bad}}(\ell)$,
we have $(\sigma',\tau',\seqS')\not \in \+T_{\ell}$.
Therefore, $(\sigma',\tau',\seqS')\not\in  \+L(\ell)$.
Thus, we have $\+D_{\ell+1}\cap \+L(\ell) = \emptyset$.
In summary, we have
$\+D_{\ell+1}\subseteq \+L_{\!{good}}(\ell+1)\setminus \+L(\ell)$.
Then, \eqref{eq-d-subset-2} is proved.

At last, we prove \eqref{eq-d-subset-3}.
Fix arbitrary $(\sigma,\tau,\seqS) \in \+S$ and  $(\sigma',\tau',\seqS')\in \+D_{\ell+1}$. 
Recall that $\+S \subseteq  \+V_{\ell+1}\setminus \+L(\ell+1)$.
Thus, we have $(\sigma,\tau,\seqS)$ is in $\+V_{\ell+1}$ and feasible.
Combined with \Cref{lem:marginal-bound},
we have 
$\mu^{\sigma_{\bot}}_{\seqS'}(\sigma') = \mu^{\sigma_{\bot}}_{\seqS}(\sigma)\cdot \mu_{E_v^\sigma}^{\sigma}(\zero) \ge B\cdot\mu^{\sigma_{\bot}}_{\seqS}(\sigma)$.
Therefore, we have 
\begin{align*}
\sum_{(\sigma',\tau',\seqS')\in \+D_{\ell+1}} \widehat{p}^{\sigma'}_{\sigma',\tau',\seqS'}\cdot \mu^{\sigma_{\bot}}_{\seqS'}(\sigma') 
\geq \sum_{(\sigma',\tau',\seqS')\in \+D_{\ell+1}} B\cdot\widehat{p}^{\sigma'}_{\sigma',\tau',\seqS'}\cdot \mu^{\sigma_{\bot}}_{\seqS}(\sigma) \geq B^2 \cdot \widehat{p}^{\tau}_{\sigma,\tau,\seqS}\mu^{\sigma_{\bot}}_{\seqS}(\sigma),
\end{align*}
where the last inequality is by $(\sigma,\tau,\seqS) \in \+S\subseteq  \+V_{\ell+1}\setminus \+L(\ell+1)$ and \eqref{eqn-hat-error-bound}. Thus, \eqref{eq-d-subset-3} is proved.
\end{proof}

\end{document}

%% file: split-edge.tex
\begin{figure}[htbp]
    \centering
    \begin{tikzpicture}[scale=0.9]
        \node[circle, fill, inner sep=1pt, minimum size=6pt] (lu) at (0, 0) {};
        \node[circle, fill, inner sep=1pt, minimum size=6pt] (lv) at (3, 0) {};
        \node[circle, fill, inner sep=1pt, minimum size=6pt] (l1) at (-1, 1) {};
        \node[circle, fill, inner sep=1pt, minimum size=6pt] (l2) at (-1, -1) {};
        \node[circle, fill, inner sep=1pt, minimum size=6pt] (l3) at (1, 1) {};
        \node[circle, fill, inner sep=1pt, minimum size=6pt] (l4) at (1, -1) {};
        \node[circle, fill, inner sep=1pt, minimum size=6pt] (l5) at (4, 1) {};
        \node[circle, fill, inner sep=1pt, minimum size=6pt] (l6) at (4, -1) {};
        \node[circle, fill, inner sep=1pt, minimum size=6pt] (l7) at (2, 1) {};
        \node[circle, fill, inner sep=1pt, minimum size=6pt] (l8) at (2, -1) {};
    
        \draw (lu) -- (l1);
        \draw (lu) -- (l2);
        \draw (lu) -- (l3);
        \draw (lu) -- (l4);
        \draw (lu) -- (lv) node[midway, above] {$e$};
        \draw (lv) -- (l5);
        \draw (lv) -- (l6);
        \draw (lv) -- (l7);
        \draw (lv) -- (l8);

        \node at (0, -0.3) {$u$};
        \node at (3, -0.3) {$v$};

        \node at (5.5, 0) {$\stackrel{\mbox{split $e$}}{\Longrightarrow}$};
    
        \node[circle, fill, inner sep=1pt, minimum size=6pt] (ru1) at (8, 0) {};
        \node[circle, fill, inner sep=1pt, minimum size=6pt] (r1) at (7, 1) {};
        \node[circle, fill, inner sep=1pt, minimum size=6pt] (r2) at (7, -1) {};
        \node[circle, fill, inner sep=1pt, minimum size=6pt] (r3) at (9, 1) {};
        \node[circle, fill, inner sep=1pt, minimum size=6pt] (r4) at (9, -1) {};

        \draw (ru1) -- (r1);
        \draw (ru1) -- (r2);
        \draw (ru1) -- (r3);
        \draw (ru1) -- (r4);
        \draw[dashed] (ru1) -- (9, 0) node[midway, below] {$e_u$};

        \node at (8, -0.3) {$u$};
    
        \node[circle, fill, inner sep=1pt, minimum size=6pt] (rv1) at (11, 0) {};
        \node[circle, fill, inner sep=1pt, minimum size=6pt] (r5) at (10, 1) {};
        \node[circle, fill, inner sep=1pt, minimum size=6pt] (r6) at (10, -1) {};
        \node[circle, fill, inner sep=1pt, minimum size=6pt] (r7) at (12, 1) {};
        \node[circle, fill, inner sep=1pt, minimum size=6pt] (r8) at (12, -1) {};

        \draw (rv1) -- (r5);
        \draw (rv1) -- (r6);
        \draw (rv1) -- (r7);
        \draw (rv1) -- (r8);
        \draw[dashed] (rv1) -- (10, 0) node[midway, below] {$e_v$};

        \node at (11, -0.3) {$v$};
    \end{tikzpicture}
    \caption{An example of splitting the edge $e = \set{u, v}$. $e_u, e_v$ are the half-edges after splitting $e = \set{u, v}$.}
    \label{fig:splitting-edge}
\end{figure}
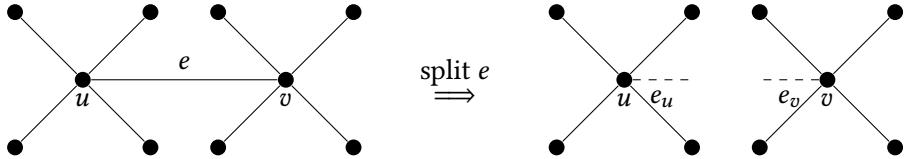

%% file: counterexample.tex
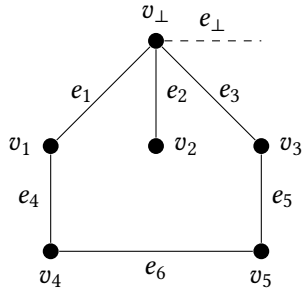
\begin{figure}[htbp]
    \centering
    \begin{tikzpicture}[scale=0.7]
        \node (v0) at (0, 3) [circle, fill, inner sep=1pt, minimum size=6pt, label=above:$v_\bot$] {};
        \node (v1) at (-2, 1) [circle, fill, inner sep=1pt, minimum size=6pt, label=left:$v_1$] {};
        \node (v2) at (0, 1) [circle, fill, inner sep=1pt, minimum size=6pt, label=right:$v_2$] {};
        \node (v3) at (2, 1) [circle, fill, inner sep=1pt, minimum size=6pt, label=right:$v_3$] {};
        \node (v4) at (-2, -1) [circle, fill, inner sep=1pt, minimum size=6pt, label=below:$v_4$] {};
        \node (v5) at (2, -1) [circle, fill, inner sep=1pt, minimum size=6pt, label=below:$v_5$] {};
        \draw[dashed] (v0) -- (2, 3) node[midway, above] {$e_{\bot}$};
        \draw (v0) -- (v1) node[midway, left] {$e_1$};
        \draw (v0) -- (v2) node[midway, right] {$e_2$};
        \draw (v0) -- (v3) node[midway, right] {$e_3$};
        \draw (v1) -- (v4) node[midway, left] {$e_4$};
        \draw (v3) -- (v5) node[midway, right] {$e_5$};
        \draw (v4) -- (v5) node[midway, below] {$e_6$};
    \end{tikzpicture}
    \caption{A counterexample against the arbitrary choice of edges.}
    \label{fig:edge-selection-counterexample}
\end{figure}